\documentclass[pra,twocolumn,a4paper,showpacs,aps,11pt,nofootinbib,allowtoday]{quantumarticle}
\pdfoutput=1
\newcommand{\ket}[1]{|#1\rangle}                        % |a>
\newcommand{\bra}[1]{\langle #1|}                       % <a|
\usepackage[normalem]{ulem}
\pdfoutput=1
\usepackage{graphicx}
\usepackage[section]{placeins}
\usepackage{amsmath}
\usepackage{amsthm}
\usepackage{amssymb}
\usepackage{dsfont}
\usepackage{bbold}
\usepackage{color}
\usepackage{comment}
\usepackage{appendix}
\usepackage{dutchcal}
\usepackage{ulem}
\usepackage{tikz}
\usepackage{subcaption}
\usepackage{url}

\newtheorem{theorem}{Theorem}
\newtheorem{lemma}{Lemma}
\newtheorem{corollary}{Corollary}

\newcommand{\Tr}{\text{Tr}}

\newcommand\id{\leavevmode\hbox{\small1\kern-3.3pt\normalsize1}}

\newcommand{\ack}{\subsection*{\normalsize \sf \textbf{Acknowledgment}}}

\newcommand{\bit}{\begin{itemize}}
	\newcommand{\eit}{\end{itemize}\par\noindent}
\newcommand{\ben}{\begin{enumerate}}
	\newcommand{\een}{\end{enumerate}\par\noindent}
\newcommand{\beq}{\begin{equation}}
	\newcommand{\eeq}{\end{equation}\par\noindent}
\newcommand{\beqa}{\begin{eqnarray*}}
	\newcommand{\eeqa}{\end{eqnarray*}\par\noindent}
\newcommand{\beqn}{\begin{eqnarray}}
	\newcommand{\eeqn}{\end{eqnarray}\par\noindent}

\usepackage{amsmath}
\usepackage{amsthm}
\usepackage{graphicx}
\usepackage{braket}
\usepackage[makeroom]{cancel}

\newtheorem{definition}{Definition}

\usepackage[colorlinks=true,allcolors=magenta,breaklinks=true]{hyperref}

\begin{document}

\title{Nonclassicality in correlations without causal order
}
	
\author{Ravi Kunjwal}
\email{quaintum.research@gmail.com}
\affiliation{Universit\'e libre de Bruxelles, QuIC, Brussels, Belgium}
\affiliation{Aix-Marseille University, CNRS, LIS, Marseille, France}

\author{Ognyan Oreshkov}
\email{ognyan.oreshkov@ulb.be}
\affiliation{Universit\'e libre de Bruxelles, QuIC, Brussels, Belgium}
	
\date{\today}                                           
\maketitle
\begin{abstract}

Bell scenarios are multipartite scenarios that exclude any signalling between parties. This leads to a strict hierarchy of classical, quantum, and non-signalling correlations in such scenarios.
Here we consider a minimal relaxation of non-signalling: each party is allowed to receive a system once, implement any local intervention on it, and  send out the resulting system once. 
Crucially, unlike Bell, we make no \textit{global} assumption about causal relations \textit{between} parties, \textit{e.g.}, they could be embedded in some exotic spacetime with indefinite causal order. We do make a causal assumption \textit{local} to each party, \textit{i.e.}, the input received by it causally precedes the output it sends out.  We then ask: Can we device-independently certify the nonclassicality of multipartite correlations in such scenarios, just as Bell inequality violations do so in Bell scenarios? \textit{A priori}, this is not clear: without some assumptions on the underlying physics (\textit{e.g.}, non-signalling), parties can realize arbitrary correlations. We therefore make a minimal assumption of logical consistency on the underlying physics, \textit{i.e.}, it must be free of time-travel antinomies \textit{without} imposing any restrictions 
on the local interventions of the parties. We then define antinomicity as a device-independent notion of nonclassicality and prove a strict hierarchy between correlation sets based on their antinomicity. An antinomic correlation cannot be explained by a classical physical theory compatible with free local interventions on pain of logical contradictions in the theory. On the other hand, parties exchanging quantum systems can witness antinomicity while respecting logical consistency. Antinomicity reduces to Bell nonlocality for non-signalling parties. It also resolves a conceptual puzzle, namely, the failure of causal inequalities as witnesses of nonclassicality: antinomicity implies causal inequality violations, but not conversely.
	\end{abstract}
\tableofcontents
	
\section{Introduction}

The scientific enterprise hinges on uncovering causal explanations for observed correlations. Usually these explanations assume definiteness of causal order between the relevant events. If, however, causal order is subject to quantum indefiniteness in a similar sense as physical properties like position and momentum, what would such explanations look like?

Aside from its intrinsically foundational motivations, this question also stems from Hardy's observation that a deeper theory that recovers both quantum theory and general relativity in appropriate physical regimes must meet the challenge of combining the radical aspects of both theories: the intrinsically \textit{probabilistic} character of quantum theory and the intrinsically \textit{dynamical} nature of causal structure in general relativity  \cite{Hardy05,Hardy07,Hardy16}.

The process-matrix framework \cite{OCB12} is a relatively mild extension of standard quantum theory that models a scenario with local quantum operations, but without assuming a global causal order between these operations.
The framework asks for the most general way of assigning probabilities to the outcomes of such operations: this leads to the central object of the framework, namely, the process matrix \cite{OCB12,SCM18}. The process matrix is a generalization of states and channels to a higher-order object that encodes correlations between local quantum operations. It can be conceptualized as a description of the environment that surrounds the local operations and establishes causal links between them. It will be useful in subsequent discussion to refer to each local operation as being specific to a `party' that implements this operation in a closed `lab', with an input port to receive quantum systems from the environment and an output port to send quantum systems to the environment.\footnote{A `lab' here is simply a placeholder for the locus of a freely chosen local operation.} We will then refer to correlations in the process-matrix framework as multi\textit{partite} correlations.\footnote{This is in keeping with the usual tradition of associating parties to such local operations, \textit{e.g.}, in Bell scenarios. This does not, of course, mean that one necessarily needs to think in terms of parties: one could, more generally, adopt a circuit-based perspective \cite{WBO23} and avoid all mention of `parties'. For our purposes, however, thinking in terms of parties is useful when drawing analogies with Bell scenarios.}

A key feature of the process-matrix framework is that process matrices never lead to correlations that can generate logical contradictions, even though the local interventions are assumed to be freely chosen \cite{BCR19}. In general, dropping the assumption of a definite causal order can generate grandfather or information antinomies \cite{BT21}, but the condition of \textit{logical consistency} in the framework excludes the possibility of such contradictions \cite{OCB12}. Perhaps the most surprising thing about process matrices is that, while satisfying logical consistency, they permit the violation of \textit{causal inequalities}, namely, constraints on multipartite correlations that follow from assuming a definite (but possibly unknown) causal order between the labs. The first example of a causal inequality violation was demonstrated in a bipartite scenario where each party receives a qubit from the environment, implements a local quantum instrument, and sends out a qubit to the environment \cite{OCB12}. It was also shown that in the limit where each set of local quantum operations is diagonal in a fixed basis, 
correlations with definite causal order are recovered, \textit{i.e.}, in this classical limit of the framework, causal inequalities cannot be violated. This led to the conjecture that definite causal order always arises in such a classical limit in any multipartite scenario \cite{OCB12}.

Interestingly, it was shown soon after that beyond the bipartite scenario, the classical limit of the process-matrix framework \textit{does} admit classical processes that violate causal inequalities \cite{BFW14, BW16}. In fact, already in a tripartite scenario, it becomes possible to violate causal inequalities in this classical limit. Furthermore, a causal inequality violation occurs in the classical \textit{deterministic} limit of the framework \cite{BW16}. This then puts into question the status of a causal inequality violation as a signature of nonclassicality, \textit{e.g.}, as an operational signature of an underlying theory that combines the aforementioned radical aspects of quantum theory and general relativity. Our goal is to address this deficiency of causal inequalities as witnesses of nonclassicality. Before we describe our approach, we motivate it by considering nonclassicality in the case of Bell scenarios.

In the case of Bell scenarios, non-signalling is a condition imposed on multipartite correlations that follows from the impossibility of signalling across spacelike separation. Signalling correlations are therefore excluded in Bell scenarios because they violate relativistic causality. However, non-signalling on its own is famously unhelpful \cite{PR94} in pinning down the set of quantum correlations: these form a strict subset of the set of non-signalling correlations. Yet non-signalling \textit{quantum} correlations are still strong enough to resist classical explanation, \textit{i.e.}, they can violate Bell inequalities \cite{Bell64,CHSH}. Here the relevant notion of classicality---namely, \textit{local causality}---for non-signalling correlations singles out the Bell polytope. Bell inequalities bound the strength of correlations within this polytope.

In the case of multipartite correlations without any non-signalling constraints, as long as they are realizable by operations localised in a spacetime that respects  relativistic causality, they are (by definition) causal.
This constraint---realizability in a spacetime that respects relativistic causality---leads to causal inequalities, which are respected by quantum theory \cite{WDA21, PS21}. However, as we noted, the process-matrix framework surprisingly allows for a violation of these causal inequalities, both in the case where the local operations are quantum operations \cite{OCB12} and in the case where they are classical stochastic operations \cite{BFW14,BW16}.
Hence, although causal inequalities are often considered similar in spirit to Bell inequalities \cite{OCB12}, this analogy fails in the sense of certifying nonclassicality: any non-signalling correlation that violates a Bell inequality is nonclassical, but there exist signalling correlations that violate causal inequalities without requiring any nonclassical resources. Put differently, in a world without definite causal order, the difference between a classical description and a nonclassical one is not always witnessed by causal inequality violations.

This inability of causal inequalities to generically separate quantum from classical correlations in the process-matrix framework means that these inequalities do not directly probe \textit{those aspects} of these multipartite correlations that are key to their \textit{nonclassicality}.
The notion of classicality we propose aims to pin down the quantum/classical distinction in a scenario without global causal assumptions in the same sense that Bell inequalities do so in scenarios where the causal assumption is one of a global common cause in the past of all the non-communicating parties.

We refer to our notion of classicality as \textit{deterministic consistency} or \textit{nomicity} \cite{KO24}. This notion is motivated by a particular way of taking the classical limit of correlations in the process-matrix framework and is intimately related to the classical \textit{deterministic} limit of processes in the framework \cite{BW16}. Taking this notion seriously leads us to the recognition that not every causal inequality is an appropriate analogue of a Bell inequality when viewed as a witness of nonclassicality. For example, we will see that the causal inequality violated by the Araújo-Feix/Baumeler-Wolf (``AF/BW" or ``Lugano") process \cite{BW16} is, from this perspective, not a device-independent witness of nonclassicality even as it witnesses noncausality device-independently. On the other hand, all bipartite causal inequalities serve as such witnesses of nonclassicality. The inequalities that follow from our notion of classicality are respected by causal correlations but are not, in general, saturated by them.

We now outline the structure of the paper. In Section \ref{sec:2}, we introduce the relevant notions from the process-matrix framework as well as introduce some definitions that are needed for our purposes. In Section \ref{sec:3}, we formally introduce our notion of classicality, define different sets of correlations with different degrees of nonclassicality (which we term \textit{antinomicity}), prove some general properties of these sets, and propose a weight-based measure of antinomicity that we term \textit{antinomy weight} of a correlation. We prove fundamental limits on the strength of correlations realizable in the process-matrix framework, showing that they cannot be arbitrarily antinomic. Specifically, we show that if a deterministic correlation is achievable in the (quantum) process-matrix framework, it is achievable with a deterministic classical process (or process function \cite{BT21}) whose causal structure faithfully corresponds to the signalling exhibited by the correlation (Theorem \ref{thm:det}). In particular, this implies that the set of process-matrix correlations is strictly contained in the set of all possible correlations (Corollary \ref{cor:strictincl1}) \cite{BAB19}. 
We also prove a \textit{sufficient} condition for the antinomicity of deterministic correlations (Corollary \ref{cor:siboncyc}) and a \textit{sufficient} condition for the causality of deterministic correlations (Corollary \ref{cor:causalityofvertices}), building on results of Ref.~\cite{TB22}. Furthermore, we establish a strict hierarchy of antinomicity between four different correlation sets that naturally arise in our approach (\textit{cf.}~Eq.~\eqref{eq:strictinclusions}). In Section \ref{sec:4}, we consider the bipartite scenario with a binary setting and a binary outcome per party. We use signalling graphs introduced in Section \ref{sec:2} to represent and classify vertices of the bipartite correlation polytope, followed by the identification of vertices that maximally violate the facet causal inequalities (obtained in Ref.~\cite{BAF15}) in this scenario. We then consider process-matrix violations of these inequalities from Ref.~\cite{BAF15} and compute the antinomy weight of the correlation achieved by the process-matrix constructions. After thus training our intuitions in the bipartite scenario, in Section \ref{sec:5}, we move on to the tripartite scenario with a binary setting and a binary outcome per party. We classify the vertices in this scenario via signalling graphs before moving on to the question of characterizing those noncausal vertices that admit classical realizations. These latter vertices are termed \textit{nomic vertices} since their convex hull defines the \textit{nomic polytope}, \textit{i.e.}, the polytope of nomic (classical) correlations. In the case where the parties discard their output systems (\textit{i.e.}, a non-signalling scenario), the nomic polytope is exactly the Bell polytope. Unlike the bipartite scenario, where causal inequalities are in one-to-one correspondence with our witnesses of nonclassicality, in the tripartite scenario, we identify an antinomicity witness, namely, the `Guess Your Neighbor's Input, or NOT' (GYNIN) inequality. This inequality is not saturated by any causal strategy and it displays a gap between the causal, classical, and nonclassical correlations in our approach. In particular, the inequality reveals the nonclassicality of the Baumeler-Feix-Wolf (BFW) process \cite{BFW14} compared to the AF/BW process \cite{BW16}. We conclude with open questions and outlook for future work in Section \ref{sec:discussion}.

\section{Preliminaries}\label{sec:2}
We introduce some preliminary notions that we will use in the rest of the paper.

\subsection{The operational paradigm}

The scenario where our notion of classicality is applicable is described by the following operational paradigm (following Ref.~\cite{OCB12}): there exist $N$ labs, labelled by $k\in\{1,2,\dots,N\}$, embedded in some environment such that each lab receives an input system from the environment and subsequently sends out an output system to the environment; party $S_k$ (in the $k$-th lab) receives a classical setting (or question) denoted $a_k\in A_k$ and reports a classical outcome (or answer) denoted $x_k\in X_k$. To determine the answer for a given question, $S_k$ can implement a local intervention depending on $a_k$ on the input system it receives from the environment and depending on the result of this local intervention report the corresponding answer $x_k$. 
The communication between the labs is mediated entirely by the environment, with the parties limited to  local interventions in their respective labs. The central object of investigation is the multipartite correlation $p(\vec{x}|
\vec{a})$ that the $N$ parties can achieve using their local interventions. Note that the parties can execute at most a \textit{single-round} communication protocol in the operational paradigm we envisage since each party receives and sends out a system exactly once.\footnote{
	Our scenario specializes to a Bell scenario when the environment is such that it does not allow signalling between the parties for any local interventions on the local systems.}
\subsection{Process-matrix framework}
In keeping with the general operational paradigm, we denote the different parties by $\{S_k\}_{k=1}^N$, but we now assume that they perform local \textit{quantum} operations, \textit{i.e.}, party $S_k$ has an incoming quantum system denoted $I_k$ with Hilbert space $\mathcal{H}^{I_k}$, an outgoing quantum system $O_k$ with Hilbert space $\mathcal{H}^{O_k}$, and the party can perform arbitrary quantum operations from $I_k$ to $O_k$. A local quantum operation is described by a quantum instrument, \textit{i.e.}, a set of completely positive (CP) maps $\{\mathcal{M}^{S_k}_{x_k|a_k}:\mathcal{L}(\mathcal{H}^{I_k})\rightarrow\mathcal{L}(\mathcal{H}^{O_k})\}_{x_k\in X_k}$, where the setting $a_k\in A_k$ labels the instrument and $x_k\in X_k$ labels the classical outcome associated with each CP map.\footnote{Without loss of generality, we assume that the outcome set $X_k$ is identical for all settings $a_k\in A_k$: this can be ensured by including, if needed, outcomes that never occur, \textit{i.e.}, those represented by the null CP map, for settings that have fewer non-null outcomes than some other setting.} Summing over the classical outcomes yields a completely positive and trace-preserving (CPTP) map $\mathcal{M}^{S_k}_{a_k}:=\sum_{x_k}\mathcal{M}^{S_k}_{x_k|a_k}$. The correlations between the classical outcomes of the different labs given their classical settings are given by
\begin{align}\label{eq:corrpm}
	&p(x_1,x_2,\dots,x_N|a_1,a_2,\dots,a_N)\nonumber\\
	=&\Tr(W M^{I_1O_1}_{x_1|a_1}\otimes M^{I_2O_2}_{x_2|a_2}\otimes\dots\otimes M^{I_NO_N}_{x_N|a_N}),
\end{align}
where 
\begin{align}
	&M_{x_k|a_k}^{I_kO_k}:=[\mathcal{I}^{I_k}\otimes\mathcal{M}^{S_k}_{x_k|a_k}(d_{I_k}\ket{\Phi^+}\bra{\Phi^+})]^{\rm T}\\
	&\in\mathcal{L}(\mathcal{H}^{I_k}\otimes\mathcal{H}^{O_k})
\end{align}
is the Choi-Jamiołkowski (CJ) matrix associated to the CP map $\mathcal{M}^{S_k}_{x_k|a_k}$, $\ket{\Phi^+}\in\mathcal{H}^{I_k}\otimes\mathcal{H}^{I_k}$ being the maximally entangled state, \textit{i.e.}, $\ket{\Phi^+}=\frac{1}{\sqrt{d_{I_k}}}\sum_{j=1}^{d_{I_k}}\ket{j}\ket{j}$, and $\mathcal{I}^{I_k}:\mathcal{L}(\mathcal{H}^{I_k})\rightarrow\mathcal{L}(\mathcal{H}^{I_k})$ being the identity channel. We have, for the CJ matrix associated to the CPTP map $\mathcal{M}^{S_k}_{a_k}$,
\begin{align}
	M_{a_k}^{I_kO_k}\geq 0, \Tr_{O_k}M_{a_k}^{I_kO_k}=\id_{I_k}
\end{align}

The operator $W\in\mathcal{L}(\mathcal{H}^{I_1}\otimes\mathcal{H}^{O_1}\otimes\mathcal{H}^{I_2}\otimes\mathcal{H}^{O_2}\otimes\dots\otimes\mathcal{H}^{I_N}\otimes\mathcal{H}^{O_N})$ is called a process matrix and it establishes correlations between the local interventions of the labs. $W$ satisfies the following constraints:
\begin{align}\label{eq:validprocess1}
	&W\geq 0,\\
	&\forall M^{S_k}\geq 0\textrm{ where } M^{S_k}\in\mathcal{L}(\mathcal{H}^{I_k}\otimes\mathcal{H}^{O_k})\nonumber\\
	&\textrm{and }\Tr_{O_k}M^{S_k}=\id_{I_k}: \quad\Tr(W \bigotimes_{k=1}^N M^{S_k})=1.\label{eq:validprocess2}
\end{align}

\subsection{Process functions, classical processes, and classical quasi-processes}\label{subsec:classproc}
Although the classical process framework arises in the `diagonal limit' of the process-matrix framework, namely, when all the inputs and outputs of the parties are taken to be diagonal in a fixed basis, it is helpful for our purposes to present it from first principles, \textit{a priori} independent of the process-matrix framework, following Ref.~\cite{BW16}.

In the classical process framework, each party $S_k$ has an incoming classical system represented by a random variable $I_k$ that takes values $i_k\in\{0,1,\dots,d_{I_k}-1\}$ and an outgoing classical system represented by a random variable $O_k$ that takes values $o_k\in\{0,1,\dots,d_{O_k}-1\}$.\footnote{With a slight but standard abuse of notation, we will often also use $I_k$ and $O_k$ to represent the respective sets in which these random variables take values.} The local operations of party $S_k$ are specified by the conditional probability distribution $P_{X_k,O_k|A_k,I_k}: X_k\times O_k\times A_k\times I_k\rightarrow[0,1]$, where the random variables $A_k$ and $X_k$ denote, respectively, the setting and outcome for party $S_k$. In the following, we will denote $P_{X_k,O_k|A_k,I_k}(x_k,o_k|a_k,i_k)$ as $p(x_k,o_k|a_k,i_k)$ and follow similar shorthand for all other probability distributions over random variables. Using the notation $\vec{X}:=(X_1,X_2,\dots,X_N)$, which takes values $\vec{x}:=(x_1,x_2,\dots,x_N)$, and $\vec{A}:=(A_1,A_2,\dots,A_N)$, which takes values $\vec{a}:=(a_1,a_2,\dots,a_N)$, the multipartite correlations, $P_{\vec{X}|\vec{A}}$, are then given by 
\begin{align}\label{eq:clcorr}
	&p(x_1,x_2,\dots,x_N|a_1,a_2,\dots,a_N)\nonumber\\
	=&\sum_{\vec{i},\vec{o}}\left(\prod_{k=1}^Np(x_k,o_k|a_k,i_k)\right)p(\vec{i}|\vec{o}),
\end{align}
where $P_{\vec{I}|\vec{O}}:\vec{O}\times\vec{I}\rightarrow[0,1]$ denotes the \textit{classical process} describing the environment, where $\vec{I}:=(I_1,I_2,\dots,I_N)$ takes values $\vec{i}:=(i_1,i_2,\dots,i_N)$, $\vec{O}=(O_1,O_2,\dots,O_N)$ takes values $\vec{o}:=(o_1,o_2,\dots,o_N)$, and we denote $P_{\vec{I},\vec{O}}(\vec{i}|\vec{o})$ by $p(\vec{i}|\vec{o})$.

The classical process $P_{\vec{I}|\vec{O}}$ cannot be an arbitrary conditional probability distribution: to be a classical process, $P_{\vec{I}|\vec{O}}$ must satisfy the condition of \textit{logical consistency}, \textit{i.e.}, for any set of local interventions, $P_{\vec{X},\vec{I},\vec{O}|\vec{A}}$ defined via
\begin{align}\label{eq:lc}
	p(\vec{x},\vec{i},\vec{o}|\vec{a}):=\left(\prod_{k=1}^Np(x_k,o_k|a_k,i_k)\right)p(\vec{i}|\vec{o})
\end{align}
should be a valid conditional probability distribution: $p(\vec{x},\vec{i},\vec{o}|\vec{a})\geq 0$ and $\sum_{\vec{x},\vec{i},\vec{o}}p(\vec{x},\vec{i},\vec{o}|\vec{a})=1$. This is a non-trivial requirement: it rules out the possibility of time-travel antinomies \cite{BW16,BT21} by excluding arbitrary conditional probability distributions $P_{\vec{I}|\vec{O}}$ from being classical processes. Indeed, it was shown in Ref.~\cite{BW16} that the condition of logical consistency for $P_{\vec{I}|\vec{O}}$ in Eq.~\eqref{eq:lc} is equivalent to the requirement that Eq.~\eqref{eq:clcorr} defines a valid multipartite correlation (\textit{i.e.}, satisfying positivity and normalization) for every set of local interventions, providing the classical counterpart of Eqs.~\eqref{eq:validprocess1} and \eqref{eq:validprocess2}.

In this paper, we will refer to an arbitrary conditional probability distribution $P_{\vec{I}|\vec{O}}$ as a \textit{classical quasi-process} and when this distribution is deterministic, we will represent it via a \textit{quasi-process function} $\omega:\vec{O}\rightarrow\vec{I}$, where  $\omega:=(\omega_1,\omega_2,\dots,\omega_N)$, $\omega_k:\vec{O}\rightarrow I_k$ for all $k\in\{1,2,\dots,N\}$, and $P_{\vec{I}|\vec{O}}=\delta_{\vec{I},\omega(\vec{O})}=\prod_{k=1}^N\delta_{i_k,\omega_k(\vec{o})}$.\footnote{This is non-standard terminology, but we will later find it useful in describing the most general set of correlations in multipartite scenarios.} Hence, a classical quasi-process that satisfies logical consistency will be called a \textit{classical process} \cite{BW16}. If a classical process is deterministic, \textit{i.e.}, $P_{\vec{I}|\vec{O}}:\vec{I}\otimes\vec{O}\rightarrow \{0,1\}$, then the quasi-process function associated with it is called a \textit{process function} \cite{BT21}.

Since we will often refer to the the fixed-point chacterization of process functions obtained in Ref.~\cite{BW16fp}, we recall this characterization below.

\begin{theorem}[Fixed point characterization of process functions \cite{BW16fp}]\label{thm:fixedpt}
	A function $\omega:\vec{O}\rightarrow\vec{I}$ is a process function if and only if it satisfies the following constraint: for every set of local interventions  $h:=(h_1,h_2,\dots,h_N)$, where $\forall k, h_k:I_k\rightarrow O_k$, the composition of $\omega$ with $h$, \textit{i.e.}, $\omega\circ h$, admits a unique fixed point, \textit{i.e.},
	\begin{align}
		\forall h, \exists!\vec{i}: \vec{i}=\omega\circ h(\vec{i}).
	\end{align}	
\end{theorem}

The convex hull of the set of all deterministic classical processes (specified by all possible process functions) is referred to as the \textit{deterministic-extrema polytope} of classical processes \cite{BW16}.
\subsection{Causal vs.~noncausal correlations}
A \textit{correlational scenario} consists of $N$ parties, where party $S_k$ has settings $A_k$ and each setting $a_k\in A_k$ has a set of possible outcomes $X_{a_k}$. Here $k\in[N]:=\{1,2,\dots,N\}$. Without loss of generality, we consider the situation where $|A_k|=M$ and $|X_{a_k}|=D$ for all $k, a_k$.\footnote{There is no loss of generality in the following sense: In any correlational scenario, one can define 
	$D:=\max_{a_k\in A_k, k\in[N]} |X_{a_k}|$ and $M:=\max_{k\in [N]}|A_k|$. One can then always add trivial outcomes---those that never occur---for settings that have fewer outcomes than $D$ in order to make sure that all settings have $D$ outcomes; similarly, one can always add trivial settings---those with a fixed outcome that always occurs, supplemented with $D-1$ trivial outcomes that never occur---for any party that has fewer settings than $M$ settings.} Hence, we will denote a given correlational scenario via the triple $(N,M,D)$, in a similar manner as in the case of Bell scenarios \cite{BCP14}. Here $N$ is the number of parties, $M$ is the number of measurement settings per party, and $D$ is the number of outcomes for each measurement setting.
For party $S_k$, the setting $A_k$ takes values $a_k\in\{0,1,\dots,M-1\}$, the outcome $X_k$ takes values $x_k\in\{0,1,\dots,D-1\}$. We summarize the $N$-party settings and outcomes below:
\begin{align}
	\textrm{Settings: }\vec{A}&:=(A_1,A_2,\dots,A_N),\nonumber\\
	\vec{a}&:=(a_1,a_2,\dots,a_N)\nonumber\\
	\textrm{Outcomes: }\vec{X}&:=(X_1,X_2,\dots,X_N)\nonumber\\
	\vec{x}&:=(x_1,x_2,\dots,x_N).
\end{align}
The observed probabilistic behaviour, or correlation, between the parties is given by $P_{\vec{X}|\vec{A}}: \vec{X}\times \vec{A} \rightarrow [0,1]$, satisfying non-negativity and normalization, \textit{i.e.},
\begin{align}
	P_{\vec{X}|\vec{A}}(\vec{x}|\vec{a})&\geq 0 \quad\forall \vec{x},\vec{a},\nonumber\\
	\sum_{\vec{x}}P_{\vec{X}|\vec{A}}(\vec{x}|\vec{a})&=1\quad\forall\vec{a}.
\end{align}
It will be convenient to think of this correlation as a $D^N\times M^N$ column-stochastic matrix of probabilities. Since there are no constraints beyond non-negativity and normalization on the correlation, the different probability distributions (columns of the matrix) comprising the correlation are independent. The column-stochasticity means that we can take each correlation $P_{\vec{X}|\vec{A}}$ as defining a point in the $M^N(D^N-1)$-dimensional real vector space of multipartite correlations. The set of all such correlations defines a \textit{correlation polytope} with deterministic vertices that are given by Boolean column-stochastic matrices, \textit{i.e.}, those given by $P_{\vec{X}|\vec{A}}:\vec{X}\times \vec{A} \rightarrow \{0,1\}$ satisfying normalization for each choice of setting $\vec{A}$. As noted earlier, we will often use the shorthand $p(\vec{x}|\vec{a}):=P_{\vec{X}|\vec{A}}(\vec{x}|\vec{a})$ to denote the entries of the correlation matrix $P_{\vec{X}|\vec{A}}$. 

The correlation $P_{\vec{X}|\vec{A}}$ is said to be causal if and only if it can be expressed as
\begin{align}\label{eq:causalcorr}
	p(\vec{x}|\vec{a})=\sum_{k=1}^Nq_k p(x_k|a_k)p_{x_k}^{a_k}(\vec{x}_{\backslash k}|\vec{a}_{\backslash k}),
\end{align}
where $q_k\geq 0$, $\sum_{k=1}^Nq_k=1$, and $p_{x_k}^{a_k}(\vec{x}_{\backslash k}|\vec{a}_{\backslash k})$ is a causal correlation between $(N-1)$ parties. Here $\vec{x}_{\backslash k}$ denotes the tuple $(x_1,x_2,\dots,x_{k-1},x_{k+1},\dots,x_N)$ (\textit{i.e.}, $\vec{x}$ without $x_k$) and similarly for $\vec{a}_{\backslash k}$. A $1$-party correlation $p(x_k|a_k)$ is, by definition, causal. This recursive form of causal correlations was originally derived from a principle of causality which essentially says that a freely chosen setting cannot be correlated with properties of its causal past or causal elsewhere \cite{OG16}. An intuitive way to interpret Eq.~\eqref{eq:causalcorr} is the following: In a scenario with definite (but possibly unknown) causal order, (at least) one party must be in the causal past or causal elsewhere of every other party. The probability $q_k$ for this to be the case for a specific party $S_k$ cannot depend on anyone's settings, while the probability $p(x_k|a_k)$ for the outcome of that party to take a particular value cannot depend on the settings of the others \cite{OG16}.

For the case of deterministic correlations---namely, those where $p(\vec{x}|\vec{a})\in\{0,1\}$ for all $\vec{x},\vec{a}$---the above definition reduces to 
\begin{align}\label{eq:causalcorrdet}
	&\exists k\in\{1,2,\dots, N\}:\nonumber\\
	&p(\vec{x}|\vec{a})=\delta_{x_k, f_k(a_k)}\delta_{\vec{x}_{\backslash k},f_{\backslash k}^{a_k}(\vec{a}_{\backslash k})},
\end{align}
where $p(\vec{x}_{\backslash k}|\vec{a}):=\delta_{\vec{x}_{\backslash k},f_{\backslash k}^{a_k}(\vec{a}_{\backslash k})}$ is a deterministic causal correlation between $N-1$ parties, specified by the set of functions 
\begin{align}
	f_{\backslash k}^{a_k}:=(f_1^{a_k},f_2^{a_k},\dots,f_{k-1}^{a_k},f_{k+1}^{a_k},f_{k+2}^{a_k},\dots, f_N^{a_k}),
\end{align}
where $f_j^{a_k}:\vec{A}_{\backslash k}\rightarrow X_j$ for all $j\in[N]\backslash\{k\}$.
In this deterministic case, the recursive nature of the definition entails that every
subset of parties must admit at least one party whose outcome is independent of the other parties' settings in this subset for all settings of parties in the complementary subset.\footnote{The complementary subset is empty if we are considering the full set of parties.}

The causal correlations in a correlational scenario define its \textit{causal polytope}, \textit{i.e.}, the set of all correlations that admit decompositions of the form in Eq.~\eqref{eq:causalcorr}. Facets of the causal polytope define the set of (facet) causal inequalities whose satisfaction is necessary and sufficient for membership in the causal polytope. More generally, a causal inequality (facet or otherwise) provides a necessary condition that all causal correlations must satisfy. Hence, violation of a causal inequality rules out membership in the causal polytope, thus witnessing indefinite causal order in a device-independent manner.

\subsection{Causal structures}\label{subsec:qcm}
We recall here some facts about quantum causal models \cite{ABH17, BLO19, BLO21} that will be relevant for our results and refer the reader to Refs.~\cite{BLO19,BLO21,TB22} for a more detailed presentation.

A notion that will be particularly relevant is that of the \textit{causal structure} associated with a quasi-process function. This causal structure may or may not be \textit{admissible} in the sense that it arises from a quantum causal model per Refs.~\cite{BLO19,BLO21,TB22} and we define it below.	
\begin{definition}[Causal structure of a quasi-process function]\label{def:causstructure}
	The causal structure of a quasi-process function $P_{\vec{I}|\vec{O}}=\delta_{\vec{I},\omega(\vec{O})}$, where $\omega:=(\omega_1,\omega_2,\dots,\omega_N)$ and  $\omega_k:\vec{O}\rightarrow I_k$ for all $k\in\{1,2,\dots,N\}$, is given by the directed graph $G=(V,E)$ with vertices $V:=\{1,2,\dots,N\}$ labelling the parties and directed edges $E:=\{(k,l)|\exists \vec{o}_{\backslash k}: \omega_l(\vec{o}_{\backslash k}, o_k)\neq \omega_l(\vec{o}_{\backslash k}, o'_k) \textrm{ for some } o_k\neq o'_k\}$. 
\end{definition}
In particular, there is no directed edge between two vertices in $G$---representing the causal structure of some quasi-process function---if and only if the associated parties can never signal to each other for \textit{any} choice of local interventions. When the quasi-process function under consideration is a process function, its causal structure describes a faithful and consistent causal model in the sense of Refs.~\cite{BLO19,BLO21,TB22}.\footnote{Here `faithful' refers to the fact that a directed edge $(k,l)$ in the causal structure is equivalent to the possibility of signalling from party $S_k$ to party $S_l$ under \textit{some} choice of local interventions. This property is satisfied for the causal structure of a quasi-process function by definition. On the other hand, `consistent' refers to the fact that the quasi-process function is logically consistent, \textit{i.e.}, it is a process function.}

We now recall a particular family of causal structures, namely, the \textit{siblings-on-cycles graphs} \cite{TB22}, that we will later use in our results.

\begin{definition}[Siblings-on-cycles graph]\label{def:csm}
	A siblings-on-cycles graph is a directed graph $G=(V,E)$ such that each directed cycle in $G$ contains siblings, \textit{i.e.}, for each cycle of directed edges $C:=\{(v_1,v_2),(v_2,v_3),\dots,(v_{c-1},v_c), (v_c,v_1)\}\subseteq E$, there exists a pair of vertices $v_{s_1},v_{s_2}\in\{v_1,v_2,\dots,v_c\}$ that have a common parent $v_p\in V$, \textit{i.e.}, $(v_p,v_{s_1}),(v_p,v_{s_2})\in E$.
\end{definition}
We now recall Theorem 2 from Ref.~\cite{TB22} (paraphrased below):

\begin{theorem}[Siblings-on-cycles condition]\label{thm:siboncyc}
	A necessary condition for a causal structure $G=(V,E)$ to be \textit{admissible} is that it is a siblings-on-cycles graph.
\end{theorem}
It then follows that all directed graphs that fail the siblings-on-cycles condition also fail to arise as causal structures of process functions.

Theorem 6 of Ref.~\cite{TB22} also provides a sufficient condition under which process functions are incapable of violating causal inequalities. Namely, if the causal structure of a process function is a \textit{chordless siblings-on-cycles graph}---a subset of siblings-on-cycles graphs---then the process function only admits causal correlations. We reproduce the relevant definitions from Ref.~\cite{TB22} below.

\begin{definition}[Induced graph]
	Given a directed graph $G=(V,E)$ and a subset of vertices $V'\subseteq V$, the induced graph $G[V']$ is the graph $G'=(V',E')$, where all edges $E'\subseteq E$ have endpoints in $V'$, \textit{i.e.}, $(k,l)\in E'\Leftrightarrow k,l\in V'\textrm{ and }(k,l)\in E$.
\end{definition}
\begin{definition}[Chordless siblings-on-cycles graph]\label{def:chordless}
	A chordless siblings-on-cycles graph $G$ is a siblings-on-cycles graph where every directed cycle is induced, \textit{i.e.}, the induced graph $G[C]$ is a cycle graph for every directed cycle $C$.
\end{definition}

We can now paraphrase Theorem 6 of Ref.~\cite{TB22} as follows.
\begin{theorem}\label{thm:chordless}
	Any process function with a chordless siblings-on-cycles graph as its causal structure always produces causal correlations, \textit{i.e.}, it can never violate causal inequalities.
\end{theorem}

A simple example of a chordless siblings-on-cycles graph is, for example, the one associated with the causal structure of the quantum SWITCH \cite{CDP13,BLO21}, \textit{i.e.}, $G=(V,E)$ with $V:=\{1,2,3,4\}$ and $E:=\{(3,1), (3,2), (1,2),(2,1), (1,4), (2,4), (3,4)\}$, where party $S_3$ serves as the common parent of parties $S_1$ and $S_2$ that are in a cycle. The only cycle in this graph, namely $1\leftrightarrow 2$, is an induced cycle.\footnote{For convenience, we represent the fact that there are two directions of signalling, $1\rightarrow 2$ and $1\leftarrow 2$ simply via $1\leftrightarrow 2$.} From Theorem \ref{thm:chordless}, we have that any process function with the causal structure of the quantum SWITCH can only produce causal correlations.
\subsection{Signalling graphs}
We now define the notion of a signalling graph associated with any vertex of the correlation polytope in a given correlational scenario.
\begin{definition}[Signalling graph of a vertex]\label{def:signallinggraph}
	The signalling graph of a vertex is a \textit{labelled} directed graph where each node represents a party and a directed edge from one party to another represents the fact that the outcome of the latter varies as a function of the setting of the former. It is also assumed (but not manifest in the signalling graph) that each party's setting \textit{can} affect its outcome.
\end{definition}

For example, in the $(2,2,2)$ scenario, the signalling graph of a vertex that corresponds to the function $f=(f_1,f_2)$, where $x_1=f_1(a_1,a_2)=a_1$ and $x_2=f_2(a_1,a_2)=a_1\oplus a_2$ is given by $S_1\rightarrow S_2$ since $x_2$ is a nontrivial function of $a_1$ and $x_1$ is unaffected by $a_2$, depending only on $a_1$. This is a \textit{causal vertex} because it 
is 
consistent with the definition of deterministic causal correlations in Eq.~\eqref{eq:causalcorrdet} (party $S_1$ is in the global past).  On the other hand, $S_1\leftrightarrow S_2$ is the signalling graph of a vertex that corresponds to $x_1=f_1(a_1,a_2)=a_2$ and $x_2=f_2(a_1,a_2)=a_1$. This is a \textit{noncausal vertex} because it 
is
inconsistent with the definition of causal correlations (every party's outcome depends on some other party's setting, hence there is no global past). Note that, unlike directed graphs which represent the causal structure of a process \cite{BLO19,BLO21}, 
any directed edge in the signalling graph of a vertex represents actual signalling and not merely potential signalling.\footnote{As such, these signalling graphs \textit{should not} be confused with directed graphs arising from classical split-node causal models or quantum causal models of Ref.~\cite{BLO19,BLO21,TB22}, \textit{i.e.}, the causal structures discussed in Section \ref{subsec:qcm}.}

Note that we \textit{do not }define signalling graphs for indeterministic correlations: for these correlations, it is possible to hide the signalling structure of the underlying vertices via probabilistic coarse-graining, \textit{e.g.}, consider the Popescu-Rohrlich (PR) box \cite{PR94, BCP14} that is non-signalling, \textit{i.e.}, $P_{\rm PR}(x_1,x_2|a_1,a_2)=\frac{1}{2}(\delta_{x_1x_2,00}+\delta_{x_1x_2,11})(\delta_{a_1a_2,00}+\delta_{a_1a_2,01}+\delta_{a_1a_2,10})+\frac{1}{2}(\delta_{x_1x_2,01}+\delta_{x_1x_2,10})(\delta_{a_1a_2,11})$, which marginalizes to $P_{\rm PR}(x_1|a_1)=\frac{1}{2}(\delta_{x_1,0}+\delta_{x_1,1})$ and $P_{\rm PR}(x_2|a_2)=\frac{1}{2}(\delta_{x_2,0}+\delta_{x_2,1})$. The PR-box can be seen as a convex mixture of type $\frac{1}{2}P_{12}+\frac{1}{2}P'_{12}$, where $P_{12}$ and $P'_{12}$ are (deterministic) causal vertices with the signalling graph $S_1 \rightarrow S_2$, given by 
\begin{align}
	P_{12}(x_1,x_2|a_1,a_2)&=\delta_{x_1,0}\delta_{x_2,a_1.a_2},\\
	P'_{12}(x_1,x_2|a_1,a_2)&=\delta_{x_1,1}\delta_{x_2,\overline{a_1.a_2}}.
\end{align}
It is then a plausible account (especially in a scenario that allows signalling) of PR-box correlations that at an underlying level there is 
some signalling from $S_1$ to $S_2$, but the fine-tuning inherent to the PR-box washes out this ``hidden" signalling so that it isn't manifest in its signalling properties. In fact, there are also other ways to decompose a PR-box into a mixture of deterministic signalling vertices. Following is a decomposition into \textit{noncausal vertices} (\textit{i.e.}, with signalling graph $S_1\leftrightarrow S_2$), even though PR-box correlations are \textit{causal}:
\begin{align}
	P_{\rm PR}(x_1,x_2|a_1,a_2)&=\frac{1}{2}\delta_{x_1,a_1\oplus a_2}\delta_{x_2,a_1+a_2}\nonumber\\
	&+\frac{1}{2}\delta_{x_1,\overline{a_1\oplus a_2}}\delta_{x_2,\overline{a_1}.\overline{a_2}}.
\end{align} 
We will see how signalling graphs are useful for classifying vertices of the correlation polytope in bipartite and tripartite correlational scenarios. 

We can now proceed to define our notion of classicality for correlations without causal order.

\section{Antinomicity}\label{sec:3}

To recap, we know that causal inequality violations provide a device-independent meaning to the notion of indefinite causal order, \textit{i.e.}, by ruling out that a correlation can be expressed in the form of Eq.~\eqref{eq:causalcorr}. We also know that causal inequalities can be violated with process functions, which model the deterministic exchange of classical systems between parties, \textit{e.g.}, the AF/BW (Lugano) process \cite{BW16}. Hence, causal inequality violations do not provide a device-independent notion of nonclassicality for correlations in the presence of indefinite causal order. The general question we want to answer is therefore the following: 

\textit{\textbf{What does it mean, in a device-independent sense, for indefinite causal order to be nonclassical? Or, equivalently, which \textit{noncausal} correlations are also \textit{nonclassical}?}}

A straightforward answer would seem to be the following: if a correlation lies outside the set of correlations achievable via classical processes \cite{BW16}, then it is a nonclassical correlation. The notion of classicality that this answer assumes is in many ways the most obvious one: that is, instead of quantum theory, one assumes classical probability theory holds locally and any correlations achievable within this classical process framework are, therefore, classical. This is also intuitive because it agrees with the `diagonal limit' of the process-matrix framework, \textit{i.e.}, in the limit where both the process matrix and the 
local CP maps are diagonal in a fixed product basis \cite{BW16}. However, a key subtlety arises in this limit: namely, not all classical processes can be thought of as arising from an underlying deterministic description within the classical process framework. That is, there exist classical processes that can \textit{only} be understood as \textit{fine-tuned} probabilistic mixtures over underlying deterministic descriptions 
(\textit{i.e., quasi-process functions}) 
that are incompatible with logical consistency (and thus outside the classical process framework). An example of this is the Baumeler-Feix-Wolf (BFW) process that violates a causal inequality \cite{BFW14}: this process can be understood as an equal mixture of a pair of causal loops that separately lead to logical contradictions even though their uniform mixture evades such contradictions. At the same time, there also exist classical processes violating causal inequalities that \textit{can} always be understood as a convex mixture over deterministic classical processes (\textit{i.e.}, \textit{process functions} \cite{BW16, BW16fp, BT21}) that do not lead to logical contradictions: these are exactly the classical processes within the deterministic-extrema polytope defined by Baumeler and Wolf \cite{BW16}. A key example of such a classical process is the deterministic AF/BW or Lugano process \cite{BW16,KB22}.

Our proposed answer to the question of nonclassicality of correlations departs from the straightforward answer above. We propose the following: \textit{if a correlation lies outside the set of correlations achievable via classical processes within any deterministic-extrema polytope, then it is a nonclassical correlation; otherwise it is classical}. 
As we argue below, this is a natural generalization of a well-known notion of classicality, namely, \textit{local causality}, to a scenario where there is no non-signalling constraint and, furthermore, no \textit{a priori} assumption about causal relations between the different labs. An intuition underlying our proposal is the following hallmark of classical physics (of which we take general relativity to be an example): the physical realizability of any probabilistic phenomenon within classical physics entails the physical realizability of the underlying deterministic phenomena over which we quantify our uncertainty via probabilistic coarse-graining; that is, there is no intrinsic indeterminism in classical physics.\footnote{Any indeterminism in classical physics can therefore be understood as entirely epistemic. That this kind of indeterminism can be accommodated in an agent-centric operational formulation of general relativity---Probabilistic General Relativity with Agency (PAGeR)---has been shown by Hardy \cite{Hardy16}.} For example, in the case of a Bell scenario,\footnote{That is, a scenario with multiple local experiments that are spacelike separated from each other.} the realizability of any probabilistic phenomenon within the Bell polytope in classical physics is underpinned by the fact that all the deterministic phenomena---vertices of the Bell polytope---are realizable in classical physics. Any non-signalling probabilistic phenomenon outside the Bell polytope, however, requires in its probabilistic support deterministic phenomena that are forbidden in classical physics since they exhibit signalling across spacelike separation.\footnote{This is the conceptual motivation for the term `nonlocality' that is often used to describe correlations that violate Bell inequalities when one adopts the \textit{superluminal causation paradigm} to explain them \cite{WSS20}. Bohmian mechanics \cite{Bohm52I, Bohm52II}, which motivated Bell's theorem \cite{Bell66, Bell64}, is an example of a nonlocal theory of this type that can reproduce quantum correlations in Bell scenarios. 
}
In a similar spirit, to assume that a probabilistic phenomenon achievable by the BFW process \cite{BFW14}---and unachievable by any process within the deterministic-extrema polytope---is realizable in a classical physical theory is to also assume that the deterministic phenomena achieved by the underlying deterministic causal loops are also realizable in the theory, \textit{i.e.}, the theory allows for information or grandfather antinomy \cite{BT21} at a fundamental level. However, the fact that such deterministic phenomena lead to logical contradictions makes them impossible in the classical process framework! Such deterministic phenomena, therefore, shouldn't be accessible in a classical physical theory whose probabilistic structure mirrors the classical process framework.

Before we proceed further, it is important to distinguish two notions of signalling that will appear in subsequent discussion: firstly, the usual notion (\textit{e.g.}, in Bell scenarios) of \textit{observed signalling} at the level of correlations, $P_{\vec{X}|\vec{A}}$, and secondly, the notion of \textit{environmental signalling} at the level of the underlying (quasi-)process, $P_{\vec{I}|\vec{O}}$. The latter provides a mechanism that makes it possible to activate the signalling behaviour at the observational level with the appropriate choice of local interventions, \textit{i.e.}, it dictates the potential for signalling (which depends on the environment) rather than actual signalling (which is manifest in the observed correlations). This distinction between observed and environmental signalling is at the root of the distinction between \textit{signalling graphs} and \textit{causal structures}. Causal structures are meant to serve as causal explanations of observed correlations, \textit{i.e.}, they are phrased at the level of environmental signalling, which is always potential rather than actual and requires appropriate local interventions to be actualized at the level of correlations. Signalling graphs are concerned with deterministic observed signalling at the level of correlations.

With this distinction in mind, 
in the rest of this section we formally define our notion of classicality, how it motivates defining different sets of correlations as objects of investigation, and some general properties of these sets of correlations.

\subsection{From local causality to nomicity}
In the case of non-signalling correlations $p(\vec{x}|\vec{a})$, Bell's assumption of local causality \cite{Bell76, Wiseman14} requires that 
\begin{align}\label{eq:bell}
	p(\vec{x}|\vec{a})=\sum_{\vec{i}}\prod_{k=1}^Np(x_k|a_k,i_k)p(\vec{i}),
\end{align}
where $\vec{i}=(i_1,i_2,\dots,i_N)$ denotes a source of classical shared randomness that is distributed among the parties and $p(x_k|a_k,i_k)$ denotes the local strategy of party $S_k$ with the key feature that it is independent of the settings and outcomes of other (spacelike separated) parties.\footnote{There is no loss of generality in assuming that the local interventions of the parties are deterministic on account of Fine's theorem \cite{Fine82a,Fine82b,Kunjwal15}, \textit{i.e.}, local causality---which allows the local interventions to be probabilistic---yields the same set of correlations as local determinism, namely, the Bell polytope.} The probability distribution $p(\vec{i})$ lives in the associated probability simplex with the vertices of the simplex denoting deterministic assignments to $\vec{i}$.
In terms of classical processes, local causality just requires that the parties cannot signal to each other via the environment, \textit{i.e.},  $p(\vec{i}|\vec{o})=p(\vec{i})$ for all $\vec{o}$, so that $\vec{o}$ in Eq.~\eqref{eq:clcorr} can be coarse-grained and we recover the set of correlations within the Bell polytope. Hence, $p(\vec{i})$ is a non-signalling\footnote{Remember that this is environmental non-signalling, applicable to $p(\vec{i}|\vec{o})$.} classical process.\footnote{In particular, it can be understood as a probabilistic mixture of deterministic non-signalling classical processes, \textit{i.e.}, $p(\vec{i})=\sum_lp(l)\delta_{\vec{i},\vec{i}_l}$, where $l$ labels deterministic assignments $\vec{i}_l$ to $\vec{i}$.} The gap between locally causal correlations and correlations within the classical process framework is therefore solely due to the environmental signalling allowed in the latter. This signalling, however, is not arbitrary---\textit{i.e.}, not all classical quasi-processes are allowed---and Ref.~\cite{BW16} imposes logical consistency as a minimal constraint on this environmental signalling (thus restricting to classical processes) that is necessary and sufficient to exclude grandfather or information antinomies within the framework, \textit{i.e.}, logical consistency follows from the principle that the framework should not allow environments that admit the possibility of logical contradictions under some local interventions. 
Logical consistency results in the classical process polytope which admits extremal classical processes that are indeterministic, \textit{i.e.}, they cannot be understood as probabilistic mixtures of deterministic classical processes and therefore indicate a fundamental indeterminism in any classical physical theory that mirrors the classical process framework. This is quite unlike the case of non-signalling classical processes (in the Bell case), which can always be understood as probabilistic mixtures of deterministic non-signalling classical processes. Motivated by the idea that there is no fundamental indeterminism in classical physical theories (like general relativity), we therefore impose an additional constraint on the type of signalling allowed in a classical process besides logical consistency: namely, that the classical process be expressible as a probabilistic mixture of deterministic classical processes (as in the case of Bell scenarios). This leads us to \textit{deterministic consistency} as a constraint on multipartite correlations that can be considered classical, \textit{i.e.},
\begin{align}\label{eq:loccons1}
	p(\vec{x}|\vec{a})=\sum_{\vec{i},\vec{o}}\prod_{k=1}^Np(x_k,o_k|a_k,i_k)p(\vec{i}|\vec{o}),
\end{align}
where
\begin{align}\label{eq:loccons2}
	p(\vec{i}|\vec{o})=\sum_{\lambda}p(\lambda)\delta_{\vec{i},\omega^{\lambda}(\vec{o})},
\end{align}
and where $\lambda$ labels vertices of a  deterministic-extrema polytope, \textit{i.e.}, process functions. Here $p(x_k,o_k|a_k,i_k)$ denotes the local strategy of party $k$ given the input $i_k$ and a random setting $a_k$. 
Hence, deterministic consistency can be seen as a conjunction of two assumptions on the realizability of a correlation via some classical quasi-process under local interventions:\footnote{As we will show further on, \textit{every} correlation admits a realization with a classical quasi-process under local interventions if no further assumptions are imposed on the realization.} firstly, that the classical quasi-process satisfies \textit{logical consistency}, \textit{i.e.}, it is a classical process, and, secondly, that it satisfies \textit{determinism}, \textit{i.e.}, it lies within the deterministic-extrema polytope.\footnote{Intuitively, this amounts to requiring that, in a classical realization of a correlation, logical consistency holds not only on average (at the level of $p(\vec{i}|\vec{o})$), but also in a fine-grained sense (at the level of $\delta_{\vec{i},\omega^{\lambda}(\vec{o})}$).}

We will often refer to deterministic consistency as \textit{nomicity} and its failure to account for some correlation (analogous to the failure of local causality) as \textit{antinomicity} (analogous to nonlocality). Hence, any correlation that fails to be \textit{nomic} will be referred to as \textit{antinomic}.

\subsubsection{Nomic polytope: the classical set of correlations}
The set of nomic correlations is convex: consider two deterministically consistent correlations $p_0(\vec{x}|\vec{a})$ and $p_1(\vec{x}|\vec{a})$ such that 
\begin{align}
	p_0(\vec{x}|\vec{a})&=\sum_{\vec{i},\vec{o}}\prod_{k=1}^Np_0(x_k,o_k|a_k,i_k)p_0(\vec{i}|\vec{o}),\\
	p_1(\vec{x}|\vec{a})&=\sum_{\vec{i},\vec{o}}\prod_{k=1}^Np_1(x_k,o_k|a_k,i_k)p_1(\vec{i}|\vec{o}).
\end{align}
Given any $q\in(0,1)$, we have that the correlation $p(\vec{x}|\vec{a}):=qp_0(\vec{x}|\vec{a})+(1-q)p_1(\vec{x}|\vec{a})$ is achievable via the classical process
\begin{align}
	p(i_*,\vec{i}|\vec{o}):=q\delta_{i_*,0}p_0(\vec{i}|\vec{o})+(1-q)\delta_{i_*,1}p_1(\vec{i}|\vec{o}),
\end{align}
which, by convexity, is within the deterministic-extrema polytope in a higher-cardinality scenario, where a shared auxiliary input $I_*:=\{0,1\}$ is included for all parties. The local interventions on this classical process are given by
\begin{align}
	&p(x_k,o_k|a_k,i_k,i_*)\nonumber\\
	=&\delta_{i_*,0}p_0(x_k,o_k|a_k,i_k)+\delta_{i_*,1}p_1(x_k,o_k|a_k,i_k), 
\end{align}
so that 
\begin{align}
	p(\vec{x}|\vec{a})&=\sum_{i_*,\vec{i},\vec{o}}\prod_{k=1}^Np(x_k,o_k|a_k,i_k,i_*)p(i_*,\vec{i}|\vec{o})\nonumber\\
	&= qp_0(\vec{x}|\vec{a})+(1-q)p_1(\vec{x}|\vec{a}).
\end{align}
Furthermore, the extremal points of this set of correlations are deterministic: this follows from noting that local interventions in Eq.~\eqref{eq:clcorr} can be assumed to be deterministic without loss of generality and any classical process within the deterministic-extrema polytope is expressible as a convex mixture over process functions \cite{BW16}. Hence, the nomic correlations in any correlational scenario form a convex polytope with deterministic vertices. We term this the \textit{nomic polytope} of the scenario, generalizing the Bell polytope in the same scenario.

\subsubsection{The `determinism' in `deterministic consistency'}
We emphasize here that the `deterministic' in `deterministic consistency' refers to the environmental correlations rather than the observed correlations. Specifically, it \textit{does not} refer to the fact that local interventions can be assumed to be deterministic without loss of generality in the classical process framework. 

In the case of Bell scenarios, 
every probability distribution $p(\vec{i})$ is a valid non-signalling classical process that is expressible as a convex mixture over deterministic non-signalling classical processes; hence, any correlation obtained from such a process satisfies deterministic consistency.\footnote{Indeed, this sort of determinism is implicit (and holds trivially)  in the ontological models framework \cite{HS10}, where one assumes a prior probability distribution over ontic states that is independent of any future interventions.} In the case of arbitrary correlational scenarios (without the assumption of non-signalling), we replace the classical probability simplex that contains $p(\vec{i})$ with its appropriate generalization. This generalization is the deterministic-extrema polytope which describes those classical environments (\textit{i.e.}, $p(\vec{i}|\vec{o})$) that can be understood as arising in a classical physical theory without fundamental indeterminism and without the possibility of logical contradictions.

The fact that nomicity is the natural generalization of local causality\footnote{We will soon see how Theorem \ref{thm:det} and the surrounding discussion of how it specializes to Bell scenarios makes this claim rigorous.} lends further credence to suggestions in the literature \cite{BW16} that classical processes outside the deterministic-extrema polytope are pathological when viewed from a classical physical perspective, \textit{e.g.}, these processes require a precise fine-tuning of probabilistic mixtures over quasi-process functions to maintain their logical consistency \cite{BW16} and, furthermore, they fail to be reversibly extendible \cite{Baumeler17}.

\subsection{A tale of two classical limits and how they fail to coincide without a non-signalling constraint}

The usual classical limit of the process-matrix framework (which recovers the classical process framework) is obtained as follows: Fix the basis of the input and output systems for each party, \textit{e.g.}, $\{\ket{i_k}\}_{i_k=0}^{d_{I_k}-1}$ for the input system and $\{\ket{o_k}\}_{o_k=0}^{d_{O_k}-1}$ for the output system of party $S_k$, $k\in\{1,2,\dots,N\}$. The local operations of party $S_k$ can then be written as 
\begin{align}
	M^{I_kO_k}_{x_k|a_k}&:=\sum_{i_k,o_k}p(x_k,o_k|a_k,i_k)\ket{i_k}\bra{i_k}\otimes\ket{o_k}\bra{o_k}\\
	&=\sum_{i_k,o_k}p(x_k,o_k|a_k,i_k)\ket{i_k o_k}\bra{i_k o_k},
\end{align}
where we leave the tensor product implicit and represent the local bases at the input and the output together as $\{\ket{i_k o_k}\}_{i_k,o_k}$. The process matrix can then, without loss of generality, be assumed to be diagonal in the product basis over all the parties,\footnote{There is no loss of generality in the sense that the observed correlations will only depend on the diagonal terms in the process matrix when it is expressed in the product basis; any non-zero off-diagonal terms will not contribute to the correlations.} \textit{i.e.},
\begin{align}
	W&:=\sum_{\vec{i},\vec{o}}p(\vec{i}|\vec{o})\ket{\vec{i}\vec{o}}\bra{\vec{i}\vec{o}},\\
	\textrm{where }\ket{\vec{i}\vec{o}}&:=\bigotimes_{k=1}^N \ket{i_k o_k}\bra{i_k o_k}.
\end{align}
The correlations achievable in this classical limit are given by 
\begin{align}
	p(\vec{x}|\vec{a})&=\Tr(W\bigotimes_{k=1}^{N}M^{I_kO_k}_{x_k|a_k})\\
	&=\sum_{\vec{i},\vec{o}}\prod_{k=1}^N p(x_k,o_k|a_k,i_k)p(\vec{i}|\vec{o}),
\end{align}
where $p(\vec{i}|\vec{o})$ is a classical process, \textit{i.e.}, it must satisfy the constraint of logical consistency.

Under our proposal, the appropriate classical limit of the process-matrix framework is one where $p(\vec{i}|\vec{o})$ is a classical process within the deterministic-extrema polytope, \textit{i.e.}, it lies in the convex hull of deterministic classical processes (or process functions). The correlations arising under this restriction are the ones we deem `classical', \textit{i.e.}, they satisfy \textit{deterministic consistency}. Hence, the classical limit of the process-matrix framework under our proposal is taken in two steps:
\begin{itemize}
	\item First, take the \textit{deterministic} diagonal limit of the process-matrix framework, \textit{i.e.}, under the restriction that $\forall\vec{i},\vec{o}: p(\vec{i}|\vec{o})\in\{0,1\}$.
	\item Second, take the convex hull over deterministic classical processes arising in this limit as the set of valid classical processes.\footnote{These are the processes within the deterministic-extrema polytope.}
\end{itemize}
This contrasts with the one-step procedure of taking the classical limit as every valid process arising in the diagonal limit of the process-matrix framework. While these two ways of obtaining the classical limit of correlations differ in the process-matrix framework (as we will later show), they coincide in the case of non-signalling correlations realizable in quantum theory, \textit{i.e.}, both recover correlations within the Bell polytope.

In a Bell scenario, the process matrix must be non-signalling, \textit{i.e.}, as a channel from the parties outputs to their inputs, it must not signal any information about the outputs to the inputs. Hence, the outputs of the labs can be discarded and the process matrix reduces to a density matrix in this case. The local quantum instruments are effectively just positive operator-valued measures (POVMs) and we have
\begin{align}
	p(\vec{x}|\vec{a}):=\Tr(\rho \bigotimes_{k=1}^NE_{x_1|a_1}),
\end{align}
where $\rho\in\mathcal{L}(\bigotimes_{k=1}^N\mathcal{H}_{I_k})$ is a quantum state and $\forall k: E_{x_k|a_k}\in\mathcal{L}(\mathcal{H}_{I_k})$ represent quantum effects.

The diagonal limit of this set of correlations is given, as before, by fixing the local bases of the input systems that the parties receive and taking $\rho$---without loss of generality---to be diagonal in the product of these local bases, \textit{i.e.},
\begin{align}
	E_{x_k|a_k}&:=\sum_{i_k}p(x_k|a_k,i_k)\ket{i_k}\bra{i_k},\\
	\rho&:=\sum_{\vec{i}}p(\vec{i})\ket{\vec{i}}\bra{\vec{i}},\\
	&\textrm{where }\ket{\vec{i}}:=\bigotimes_{k=1}^N\ket{i_k}.
\end{align}
The resulting correlations in the this limit are then 
\begin{align}
	p(\vec{x}|\vec{a})=\sum_{\vec{i}}\prod_{k=1}^Np(x_k|a_k,i_k)p(\vec{i}),
\end{align}
which is exactly the set of locally causal correlations, \textit{i.e.}, correlations within the Bell polytope.

The classical limit according to our prescription in this case would first restrict to the case of deterministic classical states, where $\forall \vec{i}: p(\vec{i})\in\{0,1\}$, \textit{i.e.}, vertices of the simplex of probability distributions over  $\{\vec{i}\}$. It would then take the convex hull over these deterministic classical states, \textit{i.e.}, the full simplex of probability distributions. It is easy to see that our prescription coincides with the diagonal limit in the non-signalling case: in both cases we recover the full simplex of probability distributions over $\{\vec{i}\}$ as the set of valid classical probabilistic states. On the other hand, in the general case of the process-matrix framework, the two ways of arriving at the classical limit do not coincide: our prescription leads to a strictly smaller set of classical processes (and resulting correlations) than the diagonal prescription; furthermore, neither prescription recovers the full set of conditional probability distributions $\{p(\vec{i}|\vec{o})\}$ (since they lead to logical contradictions in general).\footnote{Note that if either prescription were to recover the full set of such conditional probability distributions, then no notion of classicality for correlations would be definable based on such a classical limit: all multipartite correlations are accessible via such classical quasi-processes, \textit{i.e.}, they belong to the set of \textit{quasi-consistent correlations}, which we will define shortly.}

From the perspective of these two classical limits, the claim that our prescription is the appropriate generalization of nonclassicality in the non-signalling case to nonclassicality in the case of arbitrarily signalling correlations rests on the following observations:
\begin{itemize}
	\item Violation of causal inequalities by process functions means that, in principle, such violations are consistent with any classical account of causal relations that satisfies two conditions: firstly, that the causal relations are fundamentally deterministic (even if they don't respect a partial order), and secondly, that they do not lead to logical contradictions under any set of local interventions. 
	
	Any multipartite correlation that can be realized by a process function is, in this sense, classical. 
	
	\item We also want to allow for epistemic uncertainty about causal relations in a classical account of causal relations, hence we consider the full deterministic-extrema polytope as the set of probabilistic causal relations that are classically achievable.
	
	\item The above account of classicality of correlations means that a correlation is nonclassical if and only if its realization requires a classical quasi-process that is outside the deterministic-extrema polytope. This leads to two possibilities for interpreting the nonclassicality of such an explanation:
	\begin{enumerate}
		\item If one holds on to the idea that classical causal relations must be fundamentally deterministic\footnote{As they are in general relativity.} (and that any probabilities are due to epistemic uncertainty), then such an explanation \textit{must} entertain quasi-process functions that are logically inconsistent. The essence of nonclassicality is then in the fact that such logically inconsistent quasi-process functions are \textit{necessary} in any account of the correlation based on classical quasi-processes. This is the position we endorse.
		
		\item If one allows for intrinsic indeterminism in classical causal relations,\footnote{Something that is already in conflict with the fact that general relativity is a deterministic theory.} then such an explanation can \textit{sometimes (but not always)} be understood without requiring the possibility of any logical contradictions, \textit{i.e.}, using classical processes outside the deterministic-extrema polytope. In this case, a correlation is nonclassical if and only if even a fundamentally indeterministic account of causal relations must invoke logical contradictions, \textit{i.e.}, using classical quasi-processes outside the classical process polytope. This is the position we \textit{do not} endorse.
	\end{enumerate}	
\end{itemize}

We now proceed to formally define a hierarchy of sets of signalling correlations in correlational scenarios, similar in spirit to the hierarchy of sets of non-signalling correlations in Bell scenarios.
\subsection{Sets of correlations}\label{subsec:hierarchy}
We define the following sets of correlations in a scenario without global causal assumptions: \textit{deterministically consistent (or nomic) correlations}, \textit{probabilistically consistent correlations}, \textit{quantum process correlations}, and \textit{quasi-consistent correlations}.\footnote{One could also consider defining a set of ``generalised probabilistic process correlations", namely, correlations that arise from a general process framework where the local theory is assumed to be neither a classical probabilistic theory nor operational quantum theory, but a generalised probabilistic theory (GPT). However, in the absence of a concrete framework for such a theory that naturally generalizes process matrices, we restrict ourselves largely to correlations within the process-matrix framework. Nonetheless, we do define quasi-consistent correlations to formally capture the possibility of correlations that are outside the process-matrix framework and that could come from a generalised probabilistic process framework.}
\subsubsection{Deterministically consistent, or nomic, correlations}

A correlation $P_{\vec{X}|\vec{A}}$ is said to be a deterministically consistent (or nomic) correlation if and only if it admits a realization via a classical process within the deterministic-extrema polytope, \textit{i.e.}, following Eqs.~\eqref{eq:loccons1} and \eqref{eq:loccons2}. We denote the set of such correlations by $\mathcal{DC}$. As we argued earlier, in any correlational scenario, $\mathcal{DC}$ is a convex polytope with deterministic vertices.
\subsubsection{Probabilistically consistent correlations}
A correlation $P_{\vec{X}|\vec{A}}$ is said to be a  probabilistically consistent correlation if and only if there exists a classical process $P_{\vec{I}|\vec{O}}=\sum_{\kappa}p(\kappa)p(i_k|\vec{o}_{\backslash k},\kappa)$ (where $\kappa$ denotes the vertices of the classical process polytope) and local stochastic maps $\{P_{X_k,O_k|A_k,I_k}\}_k$ such that
\begin{align}
	p(\vec{x}|\vec{a})=\sum_{\kappa}p(\kappa)\sum_{\vec{i},\vec{o}}\prod_{k=1}^np(x_k,o_k|a_k,i_k)p(\vec{i}|\vec{o},\kappa).
\end{align}
We denote this set of correlations by $\mathcal{PC}$ and we have $\mathcal{DC}\subseteq\mathcal{PC}$. Convexity of this set of correlations is similarly shown as in the case of the set  $\mathcal{DC}$. In any correlational scenario, $\mathcal{PC}$ defines a convex polytope that may have indeterministic vertices.

\subsubsection{Quantum process correlations}
A correlation $P_{\vec{X}|\vec{A}}$ is said to be a quantum process correlation if and only if it admits a process matrix realization according to Eq.~\eqref{eq:corrpm} for some choice of local interventions. 
We denote this set of correlations by $\mathcal{QP}$ and we have $\mathcal{DC}\subseteq\mathcal{PC}\subseteq\mathcal{QP}$. This set is known to be convex \cite{BAF15}, but is not a polytope.
\subsubsection{Quasi-consistent correlations}
A correlation $P_{\vec{X}|\vec{A}}$ is said to be a  quasi-consistent correlation if and only if there exists a conditional probability distribution (i.e., a \textit{classical quasi-process}) $P_{\vec{I}|\vec{O}}=\sum_{\gamma}p(\gamma)\delta_{i_k,\nu_{\gamma}(\vec{o})}$---where $\gamma$ denotes arbitrary quasi-process functions  $\{\nu_{\gamma}:\vec{O}\rightarrow \vec{I}\}_{\gamma}$---and local stochastic maps $\{P_{X_k,O_k|A_k,I_k}\}_k$ such that
\begin{align}
	p(\vec{x}|\vec{a})=\sum_{\gamma}p(\gamma)\sum_{\vec{i},\vec{o}}\prod_{k=1}^Np(x_k,o_k|a_k,i_k)\delta_{i_k,\nu_{\gamma}(\vec{o})}.
\end{align}
We denote this set of correlations by $\mathcal{qC}$. It is easy to see that \textit{every} multipartite correlation is a quasi-consistent correlation, \textit{i.e.}, if there are no restrictions on the signalling behaviour of the classical quasi-process $p(\vec{i}|\vec{o})$, then it can achieve arbitrary multipartite correlations. Here is a simple argument to this effect: one can define
\begin{align}
	&p(\vec{i}|\vec{o}):=\sum_{\vec{x}',\vec{a}'}p(\vec{x}'|\vec{a}')\delta_{\vec{i},\vec{x}'}\delta_{\vec{o},\vec{a}'},\\
	&p(x_k,o_k|a_k,i_k):=\delta_{x_k,i_k}\delta_{o_k,a_k},
\end{align}
so that 
\begin{align}
	\sum_{\vec{i},\vec{o}}\prod_{k=1}^N\delta_{x_k,i_k}\delta_{o_k,a_k}\sum_{\vec{x}',\vec{a}'}p(\vec{x}'|\vec{a}')\delta_{\vec{i},\vec{x}'}\delta_{\vec{o},\vec{a}'}=p(\vec{x}|\vec{a}).
\end{align}
That is, one can simply use a classical quasi-process that encodes any desired multipartite correlation $p(\vec{x}|\vec{a})$ perfectly\footnote{This is always possible because classical quasi-processes are simply conditional probability distributions with no additional constraints on their signalling behaviour.} and use local interventions that simply pass on the correlations implicit in $p(\vec{i}|\vec{o})$ to recover $p(\vec{x}|\vec{a})$. The set $\mathcal{qC}$ is the convex polytope describing the full set of correlations in any correlational scenario; all its vertices are deterministic.

We therefore have that
\begin{align}
	\mathcal{DC}\subseteq\mathcal{PC}\subseteq\mathcal{QP}\subseteq \mathcal{qC}.
\end{align}

\section{General properties of sets of correlations}

We now make some general observations---not tied to particular correlational scenarios---about the bounds on correlations that follow from assuming nomicity as one's notion of classicality. In subsequent sections, we will follow up these observations with a detailed analysis of the causal/noncausal separation vis-à-vis the classical/nonclassical separation in bipartite and tripartite scenarios with binary inputs and outputs.

The causal/noncausal separation coincides with the classical/nonclassical separation in the bipartite case. That is, in the bipartite case, every causal inequality is a witness of nonclassicality and conversely. This was proved in Ref.~\cite{OCB12}, where it was also suggested that these two types of separation will always coincide, \textit{i.e.}, definite causal order will emerge in the limit where the parties locally behave according to a classical probabilistic theory rather than quantum theory (\textit{e.g.}, because of a decoherence process). However, it was soon shown that the two types of separation fail to coincide in the tripartite case which admits causal inequality violations with classical processes \cite{BW16}. A simple scenario where one can see the distinction between the causal/noncausal separation and the classical/nonclassical separation is, therefore, the tripartite scenario with binary inputs and outputs. However, even before we investigate that case, we can already prove some general properties of the different sets of correlations we have defined.

\subsection{Characterization of deterministic correlations achievable in the process-matrix framework}
We are interested in characterizing deterministic correlations that are achievable via process functions. These correlations form vertices of the nomic polytope and, as such, are as fundamental to our notion of classicality as vertices of the Bell polytope are to Bell locality \cite{BCP14}, \textit{i.e.}, their convex hull defines the classical polytope. We will show below that a deterministic correlation is realizable via a process function if and only if it is realizable via a process matrix (Theorem \ref{thm:det}). 

A notion that will be key in our characterization (and subsequent results) will be that of a \textit{faithful} realization, which we define below.
\begin{definition}[Faithful realization of a nomic vertex via a process function]\label{def:faithful}
	A realization of a nomic vertex via a process function is said to be faithful if and only if the signalling graph of the vertex coincides with the causal structure of the process function, \textit{i.e.}, one party deterministically signals to another if and only if the former party can causally influence the latter.\footnote{In other words, there are no `hidden' causal influences in an underlying process function that realizes a given nomic vertex faithfully: all causal influences (or lack thereof) are `seen' by the vertex. Note that here `causal influence' refers to what we earlier also called `environmental signalling' and `signalling' refers to what we earlier also called `observed signalling'.}
\end{definition}

We are now ready to prove the following theorem that will be key to many of our results.
\begin{theorem}\label{thm:det}
	A deterministic correlation $p(\vec{x}|\vec{a})=\delta_{\vec{x},f(\vec{a})}$ 
	admits a process-matrix realization if and only if it admits a process function realization that is faithful, \textit{that is}, 
	\begin{enumerate}
		\item Firstly, there exists a process function $g$  that realizes $f$, \textit{i.e.}, $f:=(f_k:\vec{A}\rightarrow X_k)_{k=1}^N$ can be expressed as  
		\begin{align}
			\forall k\in [N]: f_k(\vec{a})=f'_k(a_k,g_k(\vec{a}_{\backslash k})),
		\end{align}
		where $g:=(g_k:\vec{A}_{\backslash k}\rightarrow B_k)_{k=1}^N$ defines a process function $g:\vec{A}\rightarrow\vec{B}$ and for all $k\in [N]$, $f'_k:A_k\times B_k\rightarrow X_k$ is such that $f'_k(a_k,b_k):=f_k(a_k,\vec{a}_{\backslash k})$ for all $\vec{a}_{\backslash k}$ satisfying $b_k=g_k(\vec{a}_{\backslash k})$. Here, for each $k\in[N]$, $B_k:=\{g_k(\vec{a}_{\backslash k})\}_{\vec{a}_{\backslash k}\in\vec{A}_{\backslash k}}$, $b_k\in B_k$, and $\vec{B}:=B_1\times B_2\times\dots\times B_N$.
		
		\item Secondly, $g$ realizes $f$ faithfully, \textit{i.e.},
		\begin{align}
			&\forall k\in [N], \forall \vec{a}_{\backslash k}, \vec{a}'_{\backslash k}\in \vec{A}_{\backslash k}:\nonumber\\
			&g_k(\vec{a}_{\backslash k})\neq g_k(\vec{a}'_{\backslash k}) \nonumber\\
			\Leftrightarrow &\exists a_k\in A_k: f_k(a_k,\vec{a}_{\backslash k})\neq f_k(a_k,\vec{a}'_{\backslash k}).
		\end{align}
	\end{enumerate}
	
\end{theorem}
Note that an analogous statement is easily seen to be true in the non-signalling case of Bell scenarios:  that is, any deterministic non-signalling correlation achievable by a (possibly entangled) quantum state is also achievable by a Bell-local model, \textit{i.e.}, vertices of the Bell polytope are classically realizable. We can therefore see the above theorem as a generalization of this property to the case of (potentially noncausal) signalling correlations achievable within the process-matrix framework. Theorem \ref{thm:det} also puts fundamental limits on the deterministic correlations achievable via the process-matrix framework: all vertices of the correlation polytope that fail to be nomic also fail to be achievable via any process matrix.

A function $f$ is realizable by some process matrix $W$ under local quantum instruments $\{M_{x_k|a_k}^{I_kO_k}\}_{x_k\in X_k}$, where $a_k\in A_k$ and $k\in [N]$, if and only if
\begin{align}
	\delta_{\vec{x},f(\vec{a})}=\Tr(W\bigotimes_{k=1}^NM_{x_k|a_k}^{I_kO_k}),
\end{align}
or equivalently,
\begin{align}\label{eq:perf1}
	1=\Tr(W\bigotimes_{k=1}^NM_{f_k(\vec{a})|a_k}^{I_kO_k})
\end{align}
and 
\begin{align}\label{eq:perf0}
	0=\Tr(W\bigotimes_{k=1}^NM_{x_k|a_k}^{I_kO_k}),\textrm{ if }x_k\neq f_k(\vec{a})\textrm{ for some }k\in[N].
\end{align}
\begin{figure*}
	\centering
	\includegraphics[scale=0.4]{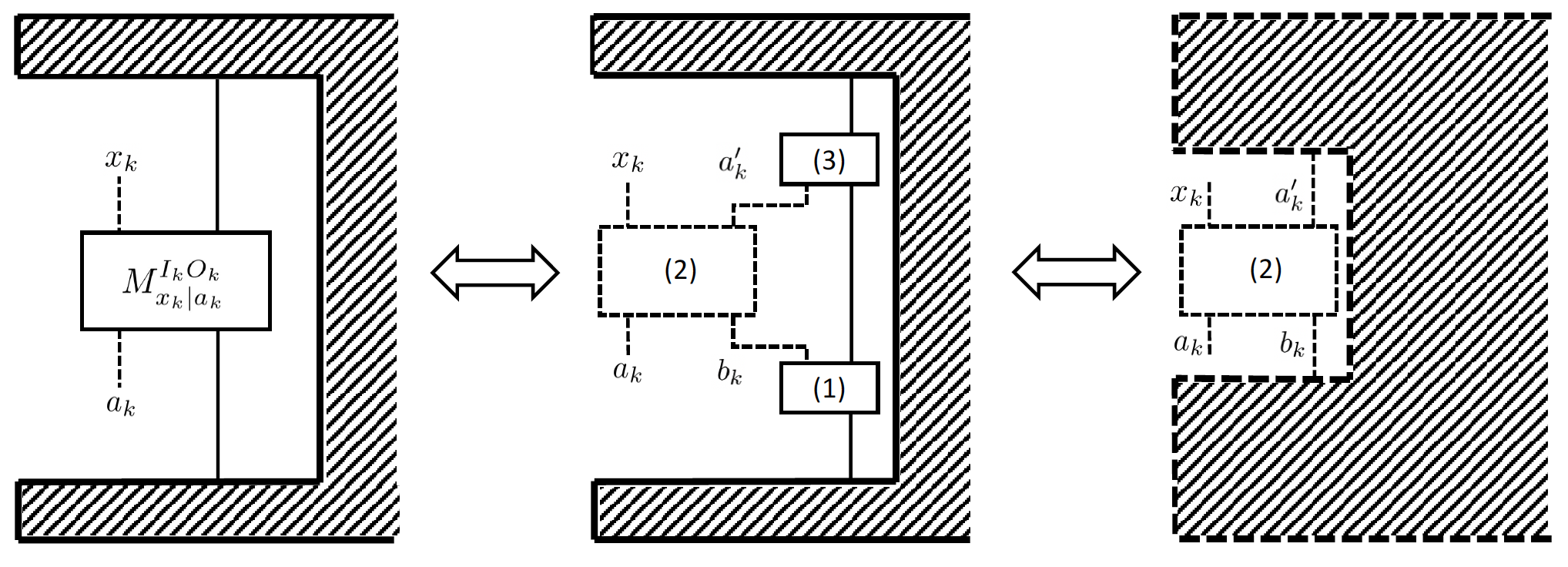}
	\caption{Transformation of the local operations of any party $S_k$ without changing the deterministic correlation $\vec{x}=f(\vec{a})$, following Lemma \ref{lem:replaceinstrument}. Solid lines indicate quantum instruments, systems, and processes, while dashed lines indicated classical instruments, systems, and processes.}\label{fig:seqofopns}
\end{figure*}

Our proof of Theorem \ref{thm:det} is based on the following Lemma, illustrated in Fig.~\ref{fig:seqofopns}.
\begin{lemma}\label{lem:replaceinstrument}
	In any process-matrix realization of a deterministic correlation $p(\vec{x}|\vec{a})=\delta_{\vec{x},f(\vec{a})}$, the local quantum operation performed by any party $S_k$ (where $k\in[N]$) can be replaced by a sequence of operations as described in Fig.~\ref{fig:seqofopns} without changing the deterministic correlation, where
	\begin{itemize}
		\item 	the instrument (1) is a fixed L\"uders instrument that projects the incoming state onto one of a set of mutually orthogonal subspaces $\mathcal{H}^{b_k}$, where $b_k\in\{1_k,\dots,N_k\}$, yielding classical outcome $b_k$, and outputting the projected state  $\rho_{b_k}\in\mathcal{L}(\mathcal{H}^{b_k})\subseteq\mathcal{L}(\mathcal{H}^{I_k})$;
		
		\item the instrument (2) is a classical instrument dependent on the local setting $a_k$ that takes as an input the outcome $b_k$ of instrument (1), sends out the output $a'_k = a_k$, and produces the outcome $x_k = f'_k(a_k, b_k)$;
		
		\item the instrument (3) is a fixed instrument that takes as input $a'_k$ and $\rho_{b_k}$ and performs exactly what the original quantum instrument $\{M^{I_kO_k}_{x_k|a_k}\}_{x_k\in X_k}$ would perform from $\mathcal{H}^{I_k}$ to $\mathcal{H}^{O_k}$ depending on $a_k$, but now from $\mathcal{H}^{b_k}$  to $\mathcal{H}^{O_k}$ conditionally on the value of $a’_k$, with the outcome of that operation traced out.
	\end{itemize}
\end{lemma}

\begin{proof}[Proof of Lemma \ref{lem:replaceinstrument}]
	Consider a given party $S_k$. For every combination of settings $\vec{a}_{\backslash k}$ of the other parties and corresponding instruments that they perform where we discard the (classical) outcomes of those instruments, party $S_k$ would receive a state $\rho_{\vec{a}_{\backslash k}}$ that may in general depend on $\vec{a}_{\backslash k}$ but not on $a_k$. This is a consequence of the fact that whenever we fix deterministic instruments (\textit{i.e.} CPTP maps) in the labs of the other parties, the probabilities for the outcomes of all local instruments that party $S_k$ may perform are the same as those that would be obtained if those instruments are applied on a fixed input density matrix \cite{OCB12}. For a given setting $a_k$, party $S_k$ performs a quantum instrument with outcomes labeled $x_k\in X_k$. By assumption, each such instrument gives a deterministic outcome---namely, $x_k = f_k(a_k, \vec{a}_{\backslash k})$---for any given state $\rho_{\vec{a}_{\backslash k}}$ that party $S_k$ may receive.
	
	We will first identify a complete set of mutually orthogonal subspaces $\mathcal{H}^{b_k}$ of the input Hilbert space $\mathcal{H}^{I_k}$ of party $S_k$, so that
	\begin{gather}\label{eq:directsum}
		\mathcal{H}^{I_k} = \bigoplus_{b_k}  \mathcal{H}^{b_k},
	\end{gather}
	and a function $g_k: \vec{A}_{\backslash k}\rightarrow B_k$, such that the support of each $\rho_{\vec{a}_{\backslash k}}$ is contained entirely in the subspace labelled by $b_k =g_k(\vec{a}_{\backslash k})$, and such that the classical outcome $x_k$ that would be obtained for a specific choice of the local setting $a_k$ is a function $f'_k$ of $a_k$ and the value $b_k$ labelling the subspace occupied by $\rho_{\vec{a}_{\backslash k}}$, \textit{i.e.}, $x_k=f'_k(a_k,b_k)$.
	
	For brevity of notation, we will temporarily relabel the different possible settings $\vec{a}_{\backslash k}$ of the other parties by $i$, so that the corresponding input states of $S_k$ are denoted by $\rho_i$. We now explain how to define the subspaces $\mathcal{H}^{b_k}$ and the functions $g_k$ and $f'_k$.
	
	Consider some value of $i$, say $i_*$. For every $i \neq i_*$, check whether there is any $a_k\in A_k$ such that $f_k(a_k, i)$ is different from $f_k(a_k, i_*)$. If not, we say that $i$ is not discriminable from ${i_*}$. Otherwise, we say that $i$ is discriminable from $i_*$: in this case, $\rho_i$ and $\rho_{i_*}$ must be perfectly distinguishable quantum states since by assumption there is an operation that yields deterministically two distinct outcomes---that is, $f_k(a_k,i)\neq f_k(a_k,i_*)$---when applied on these two different states. Two states are perfectly distinguishable if and only if their supports are orthogonal.
	
	Consider the set $\mathcal{S}_{1_k}$ of remote settings that contains $i_*$ and all $i$ such that $i$ is not discriminable from $i_*$. Observe that we would have obtained the same set of remote settings had we performed the same procedure starting from any other $i$ in $\mathcal{S}_{1_k}$. This is because, by definition: 1) $j$ is not discriminable from $j$ (\textit{reflexivity}), 2) if $j$ is not discriminable from $i$, then $i$ is not discriminable from $j$ (\textit{symmetry}), and 3) if $j$ is not discriminable from $i$ and $m$ is not discriminable from $i$, then $j$ is not discriminable from $m$ (\textit{transitivity}). Hence, the relation `not discriminable from' is an equivalence relation between remote settings and defines an equivalence class $\mathcal{S}_{1_k}$ that we refer to as an \textit{undiscriminable subset} of the set of all possible remote settings $\vec{A}_{\backslash k}$.
	
	Next, if $\mathcal{S}_{1_k}$ is not the full set of possible settings that the other parties can choose, consider some $n_*$ that is not in $\mathcal{S}_{1_k}$. Perform the same procedure as above and define the undiscriminable subset $S_{2_k}$, which contains $n_*$ and any $n$ that is not discriminable from $n_*$. Obviously, the states in the set $\{\rho_i\}_{i\in \mathcal{S}_{1_k}}$ are orthogonal to the states in the set $\{\rho_i\}_{i\in \mathcal{S}_{2_k}}$. 
	We iterate this procedure until each (remote) setting $i$ has been attributed to some undiscriminable subset $\mathcal{S}_{b_k}$, $b_k =1_k,\dots,N_k$. Note that, apart from their labels, the undiscriminated subsets $\mathcal{S}_{b_k}$ that we defined do not depend on which $i_*$, $n_*$, $\dots$, we selected in the above construction.
	
	The states in the set $\{\rho_i\}_{i\in \mathcal{S}_{b_k}}$ are all orthogonal to the states in the set 
	$\{\rho_i\}_{i\in \mathcal{S}_{b'_k}}$ for $b'_k \neq b_k$. Let us denote the span of the supports of all states in $\{\rho_i\}_{i\in \mathcal{S}_{b_k}}$ by $\mathcal{H}^{b_k}$. The subspaces $\mathcal{H}^{b_k}$ for the different values of $b_k = 1_k,\dots, N_k$, are all mutually orthogonal. If the span (direct sum) of these mutually orthogonal subspaces equals the full input Hilbert space $\mathcal{H}^{I_k}$, we define the decomposition in Eq.~\eqref{eq:directsum} in terms of these subspaces. Otherwise, let $\mathcal{K}$ denote the complement of the direct sum of these subspaces to the full input Hilbert space $\mathcal{H}^{I_k}$. We redefine $\mathcal{H}^{1_k}$ as $\mathcal{H}^{1_k} \rightarrow \mathcal{H}^{1_k} \oplus \mathcal{K}$. With this redefinition, the subspaces $\mathcal{H}^{b_k}$ now span the full input Hilbert space and we define the decomposition in Eq.~\eqref{eq:directsum} in terms of them.
	
	Now, observe that each value $\vec{a}_{\backslash k}$ of the settings of the other parties belongs to one and only one $\mathcal{S}_{b_k}$. This membership defines a function $g_k$ from $\vec{A}_{\backslash k}$ to $B_k:=\{1_k,2_k,\dots,N_k\}$. Since by definition $f(a_k, \vec{a}_{\backslash k})=f(a_k, \vec{a}'_{\backslash k})$ for all $a_k$ if and only if $g_k(\vec{a}_{\backslash k}) = g_k(\vec{a}'_{\backslash k})$, we can define $f(a_k, \vec{a}_{\backslash k})=:f'_k(a_k,b_k)$, where $f'$ is a function from $(a_k,b_k)$ to $x_k$ that takes exactly the same value as $f(a_k,\vec{a}_{\backslash k})$ for any $\vec{a}_{\backslash k}$ such that $g(\vec{a}_{\backslash k})=b_k$ (by definition there is at least one $\vec{a}_{\backslash k}$ such that $g_k(\vec{a}_{\backslash k})=b_k$).
	
	Consider now the quantum operation that party $S_k$ performs in the original protocol. We want to show that we can replace it by the sequence of local operations depicted in Fig.~\ref{fig:seqofopns} and described above, without changing the correlation $p(\vec{x} | \vec{a})=\delta_{\vec{x},f(\vec{a})}$. We first observe that by construction the sequence of operations in Fig.~\ref{fig:seqofopns} yields outcome $x_k$ with the same probability as the original instrument. Indeed, the Lüders instrument (1) would yield the correct value of $b_k$ by construction, which is then used together with $a_k$ in the classical instrument (2) to produce the same outcome $x_k=f'_k(a_k,b_k)=f_k(a_k,\vec{a}_{\backslash k})$ as the one that would result from the original operation. This means that party $S_k$ will produce the correct outcome $x_k=f(a_k, \vec{a}_{\backslash k})$ with certainty for all values of $a_k$ and $\vec{a}_{\backslash k}$.
	
	However, this only shows that the probabilities $p(x_k|a_k,\vec{a}_{\backslash k})=\delta_{x_k, f(a_k,\vec{a}_{\backslash k})}$ are the same as in the original protocol. We want to show that the full correlation
	
	\begin{gather}\label{eq:fullcorrelation}
		p(\vec{x} | a_k, \vec{a}_{\backslash k}) =\prod_{k'=1}^N \delta_{x_{k'}, f(a_k, \vec{a}_{\backslash k})}
	\end{gather} 
	remains the same if we replace the operation of party $S_{k}$ by the one in Fig.~\ref{fig:seqofopns}.
	
	Let us denote by $\tilde{p}(\vec{x} | a_k, \vec{a}_{\backslash k})$ the correlation that would be obtained if we replace the operation of party $S_k$ by the one in Fig.~\ref{fig:seqofopns}. We have already shown that $\tilde{p}(x_k| a_k,\vec{a}_{\backslash k}):= \sum_{\vec{x}_{\backslash k}}\tilde{p}(\vec{x}| a_k, \vec{a}_{\backslash k})=\delta_{x_k,f(a_k,\vec{a}_{\backslash k})}=p(x_k|a_k, \vec{a}_{\backslash k})$. We will next show that 
	\begin{align}
		\tilde{p}(\vec{x}_{\backslash k}| a_k, \vec{a}_{\backslash k}):= &\sum_{x_k}\tilde{p}(x_k, \vec{x}_{\backslash k} | a_k, \vec{a}_{\backslash k})\\
		= &\prod_{k’\neq k} \delta_{x_{k’}, f(a_k,\vec{a}_{\backslash k})} =:p(\vec{x}_{\backslash k} | a_k, \vec{a}_{\backslash k}).
	\end{align}
	
	Note that if both marginals $\tilde{p}(x_k| a_k, \vec{a}_{\backslash k})$  and $\tilde{p}(\vec{x}_{\backslash k}| a_k, \vec{a}_{\backslash k})$ are Kronecker delta functions as claimed here (still to be shown), this would mean that the full correlation $\tilde{p}(\vec{x} | a_k, 
	\vec{a}_{\backslash k})$ is a product of these marginals, \textit{i.e.}, $\tilde{p}(\vec{x} | a_k, \vec{a}_{\backslash k})  = \prod_k \delta_{x_k, f(a_k, \vec{a}_{\backslash k})}  = p(\vec{x} | a_k,\vec{a}_{\backslash k})$.  
	
	Recalling that the original instrument of party $S_k$ is described by CP maps with CJ matrices $M^{I_kO_k}_{x_k|a_k}$, let the new instrument of party $S_k$ realised via the sequence of operations in Fig.~\ref{fig:seqofopns} be described by CP maps with CJ matrices $M'^{I_kO_k}_{x_k|a_k}$. Let the CJ matrices of the original instruments of the other parties be denoted by $M^{I_{k'}O_{k'}}_{x_{k'}|a_{k'}}$, $k'\neq k$. We have that 
	\begin{gather}
		\tilde{p}(\vec{x}_{\backslash k}| a_k, \vec{a}_{\backslash k}) = \sum_{x_k} \textrm{Tr}[(M'^{I_kO_k}_{x_k|a_k} \bigotimes_{k’\neq k} M^{I_{k'}O_{k'}}_{x_{k’}|a_{k’}}) \cdot W],
	\end{gather}
	where $W$ is the process matrix and $\cdot$ denotes matrix multiplication that is made explicit for clarity. We can rewrite the previous expression as 
	\begin{gather}
		\tilde{p}(x_{\backslash k}| a_k, \vec{a}_{\backslash k}) = \textrm{Tr}(M'^{I_kO_k}_{a_k} \cdot W_{\vec{x}_{\backslash k}|\vec{a}_{\backslash k}}), 
	\end{gather}
	where $M'^{I_kO_k}_{a_k}=\sum_{x_k} M'^{I_kO_k}_{x_k|a_k}$ and $W_{\vec{x}_{\backslash k}|\vec{a}_{\backslash k}}= \textrm{Tr}_{\vec{I}_{\backslash k} \vec{O}_{\backslash k}} [(\id^{I_kO_k}\bigotimes_{k’\neq k} M^{I_{k'}O_{k'}}_{x_{k’}|a_{k’}}) \cdot W]$. We want to show that 
	\begin{gather}
		\textrm{Tr}(M'^{I_kO_k}_{a_k} \cdot W_{\vec{x}_{\backslash k}|\vec{a}_{\backslash k}}) = \textrm{Tr}(M^{I_kO_k}_{a_k} \cdot W_{\vec{x}_{\backslash k}|\vec{a}_{\backslash k}}), \label{???}
	\end{gather}
	where $M^{I_kO_k}_{a_k}=\sum_{x_k}  M^{I_kO_k}_{x_k|a_k}$. To this end, we will use the fact that for every $a_k$, the channel (with CJ operator) $M'^{I_kO_k}_{a_k}$ acts exactly like the channel $M^{I_kO_k}_{a_k}$ on any input state whose support is inside one of the subspaces $\mathcal{H}^{b_k}$ on which the Lüders instrument (1) projects. (This fact follows simply from the construction of the sequence of instruments in Fig.~\ref{fig:seqofopns}--- the Lüders instrument (1) passes on any such input state unaltered to instrument (3), which then applies on it the channel $M^{I_kO_k}_{a_k}$ determined by value of $a_k$. Note that when we sum over the outcomes $x_k$ of instrument (2), we get the classical identity channel from $a_k$ to $a’_k$.) 
	
	This fact means that if we apply the channel $M'^{I_kO_k}_{a_k}$ or the channel $M^{I_kO_k}_{a_k}$ on any input state $\sigma_{b_k} $with support in $\mathcal{H}^{b_k}$, and then apply a POVM on the output of the channel, any outcome of such a measurement, say described by a POVM element $E^T$ (the transposition is included simply for a direct correspondence with the following formula), must occur with the same probability irrespectively of whether we first applied the channel $M'^{I_kO_k}_{a_k}$ or the channel $M^{I_kO_k}_{a_k}$. In the CJ representation, this means that for any $\sigma_{b_k}\geq 0$ with support in $\mathcal{H}^{b_k}$ and any $E\geq 0 \in \mathcal{L}(\mathcal{H}^{O_k})$, we have
	
	\begin{gather}
		\textrm{Tr} [M'^{I_kO_k}_{a_k} \cdot (\sigma^{b_k} \otimes E)] = \textrm{Tr} [M^{I_kO_k}_{a_k} \cdot (\sigma^{b_k} \otimes E)]. \label{????}
	\end{gather}
	
	We will now argue that this implies Eq. \eqref{???}. Observe that $W_{\vec{x}_{\backslash k}|\vec{a}_{\backslash k}} \geq 0$ and $\sum_{\vec{x}_{\backslash k}} W_{\vec{x}_{\backslash k}|\vec{a}_{\backslash k}}  = \rho_{\vec{a}_{\backslash k}} \otimes \id^{O_k}$. This means that for all $\vec{x}_{\backslash k}$, $W_{\vec{x}_{\backslash k}|\vec{a}_{\backslash k}}$ must have support that is inside $\mathcal{H}^{b_k} \otimes \mathcal{H}^{O_k}$, where $b_k = g_k(\vec{a}_{\backslash k})$. Indeed, if any of the positive semidefinite operators $W_{\vec{x}_{\backslash k}|\vec{a}_{\backslash k}}$ had support that is not contained in $\mathcal{H}^{b_k} \otimes \mathcal{H}^{O_k}$, it would not be possible by summing these operators to obtain the operator $\rho_{a_{\backslash k}} \otimes \id^{O_k}$, which has support in $\mathcal{H}^{b_k} \otimes \mathcal{H}^{O_k}$. 
	
	The latter observation means that we can write each $W_{\vec{x}_{\backslash k}|\vec{a}_{\backslash k}}$ as a linear combination of operators of the form $\sigma^{b_k} \otimes E$, since such operators span the full space of operators over $\mathcal{H}^{b_k} \otimes \mathcal{H}^{O_k}$. By linearity, Eq. \eqref{????} then implies Eq. \eqref{???}. This completes the proof of Lemma \ref{lem:replaceinstrument}.
\end{proof}
We can now prove Theorem \ref{thm:det}.
\begin{proof}[Proof of Theorem \ref{thm:det}]
	The `if' direction of Theorem \ref{thm:det} is trivial: it follows from simply noting that a process function realization can always be seen as a process matrix realization in the diagonal limit of the process-matrix framework \cite{BW16}. We prove the `only if' direction below.
	
	Since Lemma \ref{lem:replaceinstrument} holds for an arbitrary party $S_k$, this implies that we can replace the operations of all parties by sequences of operations of the kind in Fig.~\ref{fig:seqofopns} without changing the deterministic correlation. Indeed, we can first take one party and replace their operations following the construction in Lemma \ref{lem:replaceinstrument}. We can then replace the operation of another party following the analogous construction, and so on until we cover all parties. At each step, the correlation remains the same.
	
	So far, we have found a realisation of the original correlation based on the same process matrix $W$ and a particular type of sequence of operations inside each lab. To show that the same correlation can be obtained by applying classical operations on a deterministic classical process, it suffices to notice that this is precisely what is happening effectively if we think that each party applies the classical operations (2) rather than the full operation in the original lab. The input variable $b_k$ and the output variable $a'_k$ can be thought of as, respectively, the input and output variables of a classical ‘lab’ in which the operation (2) is applied. The fact that these variables define a valid lab in the sense of the ‘classical’ (locally diagonal)  process-matrix framework follows simply from the fact that if we replace the operation (2) by any other classical operation we would obtain valid joint probabilities between all parties. This is because replacing instrument (2) by any other instrument would still amount to an overall valid quantum instrument inside the original lab (we have a sequential composition of instruments (1), (2), and (3), which yields a valid quantum instrument). 
	
	The locally diagonal process matrix on which the classical operations (2) of the parties act will now be shown to be deterministic, \textit{i.e.}, equivalent to a process function. This locally diagonal process matrix can be easily obtained by ‘unplugging’ the local classical operations (2) and composing the remaining circuit fragment (containing only the instruments (1) and (3)) in each original lab with the original process matrix $W$. (Since the link product in the CJ isomorphism is associative, the so-obtained locally diagonal process matrix on the input and output variables $b_k$ and $a'_k$ would yield the same correlation when composed with the local classical instruments (2).)
	
	It is well known that every process matrix is equivalent to a channel from the output systems of the parties to their input systems. Process functions correspond to deterministic classical examples of such channels. Therefore, what remains to be shown is that the channel from $a’_k$ to $b_k$ obtained when we compose the original process matrix $W$ with the local circuit fragments involving the instruments (1) and (3), is deterministic, \textit{i.e.}, for every $\vec{a'}$ that the parties may send out to the environment through their (classical) output systems, they receive a specific $\vec{b}$ through their input systems with probability $1$. We have already seen that when the parties output a specific $\vec{a}'=\vec{a}$ (which is what the instrument (2) in our protocol does), we obtain a specific $\vec{x}$ with probability 1. But we still have to argue that the outcome $b_k$ of the Lüders instrument (1) would be deterministic given $\vec{a}$. This was true by construction in our procedure for replacing the original instrument of party $S_k$, but it assumed fixed operations of the other parties. Although we proved that if we replace the operations of all parties the total correlation will not change, we have not shown explicitly that the supports of the states $\rho_{\vec{a}_{\backslash k}}$ received by a given party $S_k$ could not go out of the respective $\mathcal{H}^{b_k}$ when we replace the operations of other parties. It is \textit{a priori} conceivable that the outcome $b_k$ of the Lüders instrument (1) could become random upon such a replacement as long as all outcomes $b_k$ that may occur with nonzero probability for a fixed $\vec{a}$ produce the same $x_k$ when inserted in the classical instrument (2) together with $a_k$. But by construction the instrument (2) produces the outcome $x_k= f’_k(a_k, b_k)$ when applied on inputs $a_k$ and $b_k$, and by definition the function $f'_k$ is such that for $b_k \neq b'_k$ there would be at least one value of $a_k$ such that $f’_k(a_k, b_k) \neq  f’_k(a_k, b'_k)$. Since the received value of $b_k$ that would be plugged in instrument (2) is independent of $a_k$ (since in the process-matrix framework no party can signal to their own input), it is impossible for party $S_k$ to receive two (or more) different values of $b_k$ that would produce the same $x_k$ for every choice of $a_k$. Thus, the channel from $\vec{a}'$ to $\vec{b}$ that describes the newly defined locally diagonal process must be deterministic. It is given by 
	\begin{gather}
		p(\vec{b}|\vec{a}') = \prod_k\delta_{b_k, g_k(\vec{a}')} .
	\end{gather}
	This completes the proof of Theorem 4.
\end{proof}

\subsection{Strict inclusions}
The sets of correlations we have defined admit the following strict inclusion relations:
\begin{align}\label{eq:strictinclusions}
	\mathcal{DC}\subsetneq\mathcal{PC}\subsetneq\mathcal{QP}\subsetneq \mathcal{qC}.
\end{align}
The strict inclusion $\mathcal{PC}\subsetneq\mathcal{QP}$ follows from the bipartite case since quantum process correlations can violate bipartite causal inequalities \cite{OCB12, BAF15} but deterministically/probabilistically consistent correlations cannot do so in the bipartite case. Specifically, in the bipartite case, we have 
\begin{align}
	\mathcal{DC}^{(2)}=\mathcal{PC}^{(2)}\subsetneq\mathcal{QP}^{(2)}\subsetneq \mathcal{qC}^{(2)},
\end{align}
where we use the superscript to denote that these inclusion relations are specific to the bipartite case. The strict inclusion $\mathcal{QP}^{(2)}\subsetneq \mathcal{qC}^{(2)}$
already follows from Appendix A of Ref.~\cite{BAB19}: namely, that the binary input/output bipartite `Guess Your Neighbour's Input' (GYNI) game cannot be won perfectly by any correlation realizable via finite-dimensional process matrices. 
Here we adopt a different strategy to prove the strict inclusion $\mathcal{QP}\subsetneq \mathcal{qC}$: namely, as a corollary to Theorem \ref{thm:det}.
Hence, we generalize the aforementioned unachievability of the correlation winning the GYNI game perfectly (which follows from Ref.~\cite{BAB19}) to all deterministic noncausal correlations in the bipartite case and, indeed, all deterministic noncausal correlations unachievable by process functions in the general case.\footnote{Later (Corollary \ref{cor:tripartiteantinomy}), we will provide examples of provably unachievable deterministic noncausal correlations in the case of tripartite process functions: in principle, any such example can be used in proving Corollary \ref{cor:strictincl1}.}
\begin{corollary}\label{cor:strictincl1}
	The set of quantum process correlations is strictly contained within the set of quasi-consistent correlations, \textit{i.e.}, $\mathcal{QP}\subsetneq \mathcal{qC}$.
\end{corollary}
\begin{proof}
	Since the correlations achievable with bipartite process matrices are always causal \cite{OCB12}, it follows that bipartite process functions cannot achieve deterministic noncausal correlations, \textit{e.g.}, the correlation winning a bipartite GYNI game \cite{BAF15} perfectly, which is deterministic and noncausal. Hence, by Theorem \ref{thm:det}, it is also impossible to achieve such deterministic noncausal correlations with a bipartite process matrix.
	
	Indeed, \textit{any} example of a deterministic noncausal correlation that cannot be achieved by a process function allows, via Theorem \ref{thm:det}, to demonstrate the strict inclusion $\mathcal{QP}\subsetneq \mathcal{qC}$. Examples of such correlations in the tripartite case are shown further on in Corollary \ref{cor:tripartiteantinomy}.
\end{proof}

We defer a formal proof of the strict inclusion $\mathcal{DC}\subsetneq\mathcal{PC}$ until later (see Theorem \ref{thm:dcpcgap}).

\subsection{Sufficient condition for antinomicity of vertices in any correlational scenario}
We can also prove the following corollary of Theorem \ref{thm:det} that provides a sufficient condition for antinomicity of vertices of the correlation polytope.
\begin{corollary}\label{cor:siboncyc}
	Given any correlational scenario, each vertex of the correlation polytope with a signalling graph that fails to satisfy the siblings-on-cycles property is antinomic.
\end{corollary}
\begin{proof}
	Failure to satisfy the siblings-on-cycles property implies that the signalling graph of the vertex has at least one directed cycle such that no two nodes in that cycle are siblings. To assume that this vertex is nomic, then, is to assume (by Theorem \ref{thm:det}) its faithful realization via a process function whose causal structure fails the siblings-on-cycles property. However, we know from Theorem $2$ of Ref.~\cite{TB22} (reproduced as Theorem \ref{thm:siboncyc} above) that causal structures that fail the siblings-on-cycles property are inadmissible: in particular, they cannot arise from process functions. Hence, our assumption is flawed and the vertex is antinomic.
\end{proof}

\subsection{Sufficient condition for causality of vertices in any correlational scenario}

We can prove another corollary of Theorem \ref{thm:det} and Theorem 6 of Ref.~\cite{TB22} (reproduced as Theorem \ref{thm:chordless} above) that provides a sufficient condition for the causality of a vertex in any correlational scenario.

\begin{corollary}\label{cor:causalityofvertices}
	Given any correlational scenario, any vertex that is realizable within the process-matrix framework and has a signalling graph satisfying the chordless siblings-on-cycles property (Definition \ref{def:chordless}) is causal.
\end{corollary}
\begin{proof}
	From Theorem \ref{thm:det}, any vertex realizable within the process-matrix framework admits a faithful realization with a process function. Hence, if the signalling graph of this vertex satisfies the chordless siblings-on-cycles property, then it can be realized by a process function with a causal structure that satisfies the same property. All correlations achievable by such a process function---including the vertex in question---are then causal by Theorem \ref{thm:chordless}.
\end{proof}

\subsection{Antinomy weight as a measure of nonclassicality}

Any correlation $P_{\vec{X}|\vec{A}}$ can be written as a convex mixture of deterministic correlations, \textit{i.e.},
\begin{align}
	p(\vec{x}|\vec{a})=\sum_v q_v\delta_{\vec{x},f_v(\vec{a})},
\end{align}
where $v$ denotes the vertices of the correlation polytope where no equality constraints (such as non-signalling) beyond normalization are assumed and, as usual, we require non-negativity. The deterministic correlation corresponding to a given $v$ can be represented by a function $f_v:\vec{A}\rightarrow\vec{X}$, so that $p_v(\vec{x}|\vec{a}):=\delta_{\vec{x},f_v(\vec{a})}$, and we take the probabilistic weight of this correlation in the decomposition above as $q_v\geq 0, \sum_vq_v=1$.

If $P_{\vec{X}|\vec{A}}$ fails to be nomic, then at least some of the deterministic correlations in any convex decomposition of it must also fail to be nomic, \textit{i.e.}, they must be antinomic. We can then define the following weight-based measure for the nonclassicality of $P_{\vec{X}|\vec{A}}$ via a minimization over its convex decompositions: 

\begin{align}\label{eq:antinomyweight}
	w_a(P_{\vec{X}|\vec{A}}):=\min_{\{ q_v,\delta_{\vec{x},f_v(\vec{a})}\}_v}\sum_{v_*}q_{v_*},
\end{align}
where for any given convex decomposition $\sum_v q_v\delta_{\vec{x},f_v(\vec{a})}=p(\vec{x}|\vec{a})$, $v_*$ labels the vertices of the correlation polytope that are antinomic. As such, we term $w_a$ the \textit{antinomy weight}, \textit{i.e.}, the minimal fraction of the correlation $P_{\vec{X}|\vec{A}}$ that requires quasi-process functions (which allow for logical antinomies \cite{BT21}) for its realization. 

We have that $w_a(P_{\vec{X}|\vec{A}})\in[0,1]$, where $w_a(P_{\vec{X}|\vec{A}})=0$ if and only if $P_{\vec{X}|\vec{A}}$ is nomic and $w_a(P_{\vec{X}|\vec{A}})=1$ if and only if no fraction of the correlation $P_{\vec{X}|\vec{A}}$ can be explained by nomic vertices, \textit{e.g.}, all $P_{\vec{X}|\vec{A}}$ corresponding to antinomic vertices have $w_a(P_{\vec{X}|\vec{A}})=1$.

We will illustrate this measure by computing it in some examples of antinomic correlations we investigate below.

\section{Bipartite correlations}\label{sec:4}
In this section, we consider the $(2,2,2)$ scenario. Although this scenario doesn't provide any new examples of nonclassicality witnesses beyond causal inequalities \cite{OCB12}, it is instructive to consider it from the perspective of nonclassicality rather than noncausality. This will also help prepare the ground for the tripartite case, where there \textit{is} a distinction to be made between noncausality and nonclassicality. The correlation polytope for this scenario is defined by non-negativity and normalization, \textit{i.e.},
\begin{align}
	&\forall \vec{x},\vec{a}:&p(x_1,x_2|a_1,a_2)\geq 0,\\
	&\forall a_1,a_2:& \sum_{x_1,x_2}p(x_1,x_2|a_1,a_2)=1.
\end{align}
The lack of equality constraints beyond normalization means that the vertices of this polytope are all deterministic, \textit{i.e.}, they are given by all possible functions $f:\vec{A}\rightarrow \vec{X}$.

\subsection{Classification of bipartite vertices}
Some of the vertices of the bipartite polytope can be implemented in a scenario with definite causal order---these are the \textit{causal vertices}---while the remaining vertices cannot be implemented in this way, \textit{i.e.}, they are \textit{noncausal} vertices. The causal vertices---and the resulting causal polytope---have been previously studied in Ref.~\cite{BAF15} but we will present them here (for completeness and) to be able to refer to both causal and noncausal vertices when we analyze the structure of noncausal correlations in the bipartite case further on.

We will represent the correlation $P_{\vec{X}|\vec{A}}$ as a $4\times 4$ matrix with entries $P_{\vec{X}|\vec{A}}(\vec{x},\vec{a}):=p(x_1,x_2|a_1,a_2)$, where $\vec{x},\vec{a}\in\{00,01,10,11\}$.\footnote{Note that we will often represent a vector $v=(v_1,v_2,\dots,v_N)$ as the string $v_1v_2\dots v_N$ for brevity of notation.} The vertices of the correlation polytope can be classified via their signalling graphs as follows:

\begin{enumerate}
	
	\item \textit{Bell (\textit{i.e.}, non-signalling) vertices:} These are vertices that have a  two-node signalling graph without any edges, \textit{i.e.}, there is no signalling between the parties. As such, these vertices are defined by local choices of functions $f_1:a_1\rightarrow x_1$ and $f_2:a_2\rightarrow x_2$ such that 
	\begin{align}\label{eq:bellvertices}
		p(\vec{x}|\vec{a})=\delta_{x_1,f_1(a_1)}\delta_{x_2,f_2(a_2)}.
	\end{align}
	There are $4$ possible function choices for each party $S_i$: $f_i(a_i)=0$ ($\forall a_i$), $f_i(a_i)=1$ ($\forall a_i$), $f_i(a_i)=a_i$ (identity), and $f_i(a_i)=\bar{a}_i$ (bit-flip). Hence, there are $4\times 4=16$ Bell vertices, as expected.
	
	\item \textit{One-way signalling (causal) vertices, $S_1\rightarrow S_2$:} These are defined by functions $f_1:a_1\rightarrow x_1$ and $f_2:(a_1,a_2)\rightarrow x_2$, so that
	\begin{align}\label{eq:12S}
		p(\vec{x}|\vec{a})=\delta_{x_1,f_1(a_1)}\delta_{x_2,f_2(a_1,a_2)},
	\end{align}
	where the output of $S_2$ (\textit{i.e.}, $x_2$) depends non-trivially on the input of $S_1$ (\textit{i.e.}, $a_1$).
	
	As in the case of Bell vertices, there are $4$ possible choices for $f_1$. For $f_2$, there are $12$ new possible choices that aren't of the type that define Bell vertices: $2$ of these just arise from passing the value of $a_1$ or $\bar{a}_1$ to $S_2$, so that $f_2(a_1,a_2)=a_1$ or $\bar{a}_1$; the remaining  $10$ choices arise from combining the value of $a_1$ with $a_2$ in the 10 distinct ways possible, \textit{i.e.}, $f_2(a_1,a_2)$ is given by one of: $a_1.a_2, a_1.\bar{a}_2, \bar{a}_1.a_2,\bar{a}_1.\bar{a}_2, \bar{a}_1+\bar{a}_2, \bar{a}_1+a_2, a_1+\bar{a}_2,a_1+a_2, a_1\oplus a_2, \overline{a_1\oplus a_2}$.\footnote{Note that $a_1\oplus a_2=a_1.\bar{a}_2+\bar{a}_1.a_2$.} This gives us a total of $4\times 12=48$ vertices with signalling graph $S_1\rightarrow S_2$.

	\item \textit{One-way signalling (causal) vertices, $S_2\rightarrow S_1$:} Similar to the above case (with $S_1$ and $S_2$ exchanged), these vertices are given by correlations of the form 
	
	\begin{align}\label{eq:21S}
		p(\vec{x}|\vec{a})=\delta_{x_1,f_1(a_1,a_2)}\delta_{x_2,f_2(a_2)},
	\end{align}
	where the output of $S_1$ (\textit{i.e.}, $x_1$) depends non-trivially on the input of $S_2$ (\textit{i.e.}, $a_2$).
	
	There are 48 such vertices.
	
	\item \textit{Two-way signalling (\textit{noncausal}) vertices, $S_1\leftrightarrow S_2$:} These are given by correlations of the form 
	
	\begin{align}\label{eq:2wayS}
		p(\vec{x}|\vec{a})=\delta_{x_1,f_1(a_1,a_2)}\delta_{x_2,f_2(a_1,a_2)},
	\end{align}
	
	where the outcome of each party depends non-trivially on the input of the other party.
	
	There are $12\times 12=144$ new choices of $f_1$ and $f_2$ where the output of each party depends nontrivially on the input of the other party (besides, possibly, its own input): i.e., $f_1(a_1,a_2)=a_2, \bar{a}_2, a_1.a_2, a_1.\bar{a}_2, \bar{a}_1.a_2,\bar{a}_1.\bar{a}_2, \bar{a}_1+\bar{a}_2, \bar{a}_1+a_2, a_1+\bar{a}_2,a_1+a_2, a_1\oplus a_2, \overline{a_1\oplus a_2}$. Similarly $f_2(a_1,a_2)=a_1, \bar{a}_1, a_1.a_2, \bar{a}_1.a_2, a_1.\bar{a}_2,\bar{a}_1.\bar{a}_2, \bar{a}_1+\bar{a}_2, a_1+\bar{a}_2, \bar{a}_1+a_2, a_1+a_2, a_1\oplus a_2, \overline{a_1\oplus a_2}$.
	
	Hence, there are $144$ such noncausal vertices.
\end{enumerate}

Note that the total number of vertices is given by the number of $\{0,1\}$-valued column-stochastic matrices $P_{\vec{X}|\vec{A}}$, which is given by $4^4=256$. We have accounted for all of these vertices in our characterization, \textit{i.e.}, $16+48+48+144=256$. There are $112$ causal vertices and $144$ noncausal vertices. The causal polytope is defined by the convex hull of the $112$ causal vertices and any violation of a causal inequality by a correlation must necessarily result from a non-zero probabilistic support on at least one of the $144$ noncausal vertices.\footnote{In the case of the Bell-CHSH scenario, the non-signalling condition imposes additional restrictions, and there we have the Bell polytope given by the convex hull of the $16$ Bell vertices, and the non-signalling polytope given by the convex hull of the $16$ Bell vertices together with $8$ non-signalling vertices (that are indeterministic, given by different versions of the PR-box). In the scenario without non-signalling that we are considering, however, there is nothing special about the PR-box compared to any other point within the causal polytope. The nonclassicality we are after is outside the causal polytope and in the bipartite case it happens to be \textit{everywhere} outside the causal polytope, making it equivalent to the violation of a causal inequality. This equivalence fails in the tripartite case, where we will see a more subtle distinction between nonclassical vs.~noncausal correlations: noncausality will be necessary but not sufficient for nonclassicality.}

\subsection{Bipartite causal inequalities in the $(2,2,2)$ scenario and their maximal violations}

The bipartite causal polytope has been previously characterized in Ref.~\cite{BAF15}, which identified two inequivalent classes of facet inequalities bounding the causal polytope: one corresponds to the Guess Your Neighbour's Input (GYNI) game and the other to a Lazy GYNI (LGYNI) game, and we take their canonical examples to be

\begin{align}
	&\textrm{GYNI:}\nonumber\\
	&p_{\rm gyni}\nonumber\\
	:=&\frac{1}{4}\sum_{\vec{a},\vec{x}}\delta_{x_1,a_2}\delta_{x_2,a_1}p(x_1,x_2|a_1,a_2)\leq\frac{1}{2},\label{gyni}\\
	&\nonumber\\
	&\textrm{LGYNI:}\nonumber\\
	&p_{\rm lgyni}\nonumber\\
	:=&\frac{1}{4}\sum_{\vec{a},\vec{x}}\delta_{a_1(x_1\oplus a_2),0}\delta_{a_2(x_2\oplus a_1),0}p(x_1,x_2|a_1,a_2)\leq\frac{3}{4}.\label{lgyni}
\end{align}
The winning condition for these games can be expressed via the following equations:
\begin{align}
	x_1=a_2&\textrm{ and }x_2=a_1 \textrm{ (GYNI)},\label{cangyni}\\
	a_1(x_1\oplus a_2)=0&\textrm{ and }a_2(x_2\oplus a_1)=0 \textrm{ (LGYNI)}.\label{canlgyni}
\end{align}
In Appendix \ref{app:bipartite}, we prove some general results on the relationship between the different versions of the facet inequalities characterizing the bipartite causal polytope and the noncausal vertices that violate them maximally. In the next subsection, one of these vertices will turn out to be crucial in characterizing the nonclassicality of previously discovered process-matrix violations of bipartite causal inequalities \cite{BAF15}.

\subsection{Process matrix violations of bipartite causal inequalities}

Within the process-matrix framework, a violation of both inequalities, Eqs.~\eqref{gyni} (GYNI) and \eqref{lgyni} (LGYNI), is known to be achievable by the same process matrix and choice of local CP maps \cite{BAF15}. To obtain this violation, we take $d_{I_1}=d_{O_1}=d_{I_2}=d_{O_2}=2$, \textit{i.e.}, each party receives a qubit from the environment and returns a qubit to the environment, where the environment is described by the process matrix 
\begin{align}
	W=\frac{1}{4}\left[\id^{\otimes4}+\frac{Z^{I_1}Z^{O_1}Z^{I_2}\id^{O_2}+Z^{I_1}\id^{O_1}X^{I_2}X^{O_2}}{\sqrt{2}}\right]
\end{align}
and the local CP maps are specified by the CJ matrices 
\begin{align}
	&M_{0|0}^{I_1O_1}=M_{0|0}^{I_2O_2}=0,\label{lo00}\\
	&M_{1|0}^{I_1O_1}=M_{1|0}^{I_2O_2}=2|\Phi^+\rangle\langle\Phi^+|,\label{lo10}\\
	&M_{0|1}^{I_1O_1}=M_{0|1}^{I_2O_2}=|0\rangle\langle0|\otimes |0\rangle\langle0|,\label{lo01}\\
	&M_{1|1}^{I_1O_1}=M_{1|1}^{I_2O_2}=|1\rangle\langle1|\otimes |0\rangle\langle0|.\label{lo11}
\end{align}

Here, $\{|0\rangle, |1\rangle \}$ is the eigenbasis of $Z$ and $|\Phi^+\rangle=\frac{|00\rangle+|11\rangle}{\sqrt{2}}\in \mathcal{H}^{I_1}\otimes \mathcal{H}^{O_1}\cong \mathcal{H}^{I_2}\otimes \mathcal{H}^{O_2}\cong \mathcal{H}^{I_1}\otimes \mathcal{H}^{I_1}$.

Operationally, the local CP maps dictate that each party $S_i$ ($i\in\{1,2\}$) implements the following local operations for a given input $a_i\in\{0,1\}$:
\begin{enumerate}
	\item $a_i=0$: $S_i$ implements the identity channel locally and outputs the value $1$.\footnote{Recall that $M_{1|0}^{I_iO_i}=[\mathcal{I}^{I_i}\otimes \mathcal{M}_{1|0}^{S_i}(2|\Phi^+\rangle\langle\Phi^+|)]^{\rm T}$, where $\mathcal{M}_{1|0}^{S_i}=\mathcal{I}^{I_iO_i}: \mathcal{L}(\mathcal{H}^{I_i})\rightarrow\mathcal{L}(\mathcal{H}^{O_i})$.}
	\item $a_i=1$: $S_i$ measures in the $Z$ basis and assigns classical outputs accordingly (\textit{i.e.}, $0$ for $|0\rangle\langle0|$ and $1$ for $|1\rangle\langle1|$); it then outputs a fixed quantum state, $|0\rangle\langle0|$.
\end{enumerate}

The quantum probability of success in the two games (GYNI and LGYNI) is calculated using 
\begin{equation}
	p(x_1,x_2|a_1,a_2)=\Tr[(M_{x_1|a_1}^{I_1O_1}\otimes M_{x_2|a_2}^{I_2O_2})W].
\end{equation}

We then have, quantumly,
\begin{align}
	p_{\rm gyni}&=\frac{5}{16}\left(1+\frac{1}{\sqrt{2}}\right)\approx 0.53347,\\
	p_{\rm lgyni}&=\frac{1}{4}+\frac{5}{16}\left(1+\frac{1}{\sqrt{2}}\right)\approx 0.78347.
\end{align}

The full set of resulting probabilities, summarized by the matrix $P^{\rm Q}_{\vec{X}|\vec{A}}$ below, are given by 
\begin{align}\label{eq:quantumdist}
	\begin{pmatrix}
		0 & 0 & 0 & \frac{1}{4}\left(1+\frac{1}{\sqrt{2}}\right)\\
		0 & 0 & \frac{1}{2}\left(1+\frac{1}{\sqrt{2}}\right) & \frac{1}{4}\left(1-\frac{1}{\sqrt{2}}\right)\\		
		0 & \frac{1}{2}\left(1+\frac{1}{\sqrt{2}}\right) & 0&\frac{1}{4}\left(1-\frac{1}{\sqrt{2}}\right)\\
		1&\frac{1}{2}\left(1-\frac{1}{\sqrt{2}}\right)&\frac{1}{2}\left(1-\frac{1}{\sqrt{2}}\right)& \frac{1}{4}\left(1+\frac{1}{\sqrt{2}}\right)
	\end{pmatrix}
\end{align}

\subsubsection{Convex decompositions of the quantum process correlation and the antinomy weight}

The quantum process correlation of Eq.~\eqref{eq:quantumdist} can be expressed as a convex mixture of four vertices of the bipartite correlation polytope. The first vertex is a (causal) Bell vertex:
\begin{align}
	&\begin{pmatrix}
		0&0&0&1\\
		0&0&1&0\\
		0&1&0&0\\
		1&0&0&0
	\end{pmatrix}\textrm{ with probability }\frac{1}{4}\left(1+\frac{1}{\sqrt{2}}\right),\nonumber\\
\end{align}
given by $x_1=f_1(a_1,a_2)=\bar{a}_1, x_2=f_2(a_1,a_2)=\bar{a}_2$. 

The second vertex is a noncausal vertex (with signalling graph $S_1\leftrightarrow S_2$), achieving value $3/4$ for the GYNI game and value $1$ for the LGYNI game:
\begin{align}\label{eq:lgynivertex}
	&\begin{pmatrix}
		0&0&0&0\\
		0&0&1&0\\
		0&1&0&0\\
		1&0&0&1
	\end{pmatrix}\textrm{ with probability }\frac{1}{4}\left(1+\frac{1}{\sqrt{2}}\right),
\end{align}
given by 
$x_1=f_1(a_1,a_2)=(a_1\oplus a_2)a_2+\overline{a_1\oplus a_2}=\bar{a}_1+a_2$ and $x_2=f_2(a_1,a_2)=(a_1\oplus a_2)a_1+\overline{a_1\oplus a_2}=a_1+\bar{a}_2$,
\textit{i.e.}, for input parity $1$, the parties guess each other's inputs perfectly and for input parity $0$, they assign the fixed outcomes $x_1=x_2=1$.

The third vertex is a causal signalling vertex (with signalling graph $S_2\rightarrow S_1$):
\begin{align}
	&\begin{pmatrix}
		0&0&0&0\\
		0&0&0&1\\
		0&0&0&0\\
		1&1&1&0
	\end{pmatrix}\textrm{ with probability }\frac{1}{4}\left(1-\frac{1}{\sqrt{2}}\right),
\end{align}
given by $x_1=f_1(a_1,a_2)=\overline{a_1.a_2}$ and $x_2=f_2(a_1,a_2)=1$.

The fourth vertex is a causal signalling vertex (with signalling graph $S_1\rightarrow S_2$):
\begin{align}
	&\begin{pmatrix}
		0&0&0&0\\
		0&0&0&0\\
		0&0&0&1\\
		1&1&1&0
	\end{pmatrix}\textrm{ with probability }\frac{1}{4}\left(1-\frac{1}{\sqrt{2}}\right),
\end{align}
given by $x_1=f_1(a_1,a_2)=1$ and  $x_2=f_2(a_1,a_2)=\overline{a_1.a_2}$.

Hence, there is one non-signalling vertex, two one-way signalling (causal) vertices, and one noncausal vertex in this convex decomposition of the quantum correlation in Eq.~\eqref{eq:quantumdist}. As expected, it is the non-zero support of the quantum correlation on the noncausal vertex that is responsible for the violation of the causal inequalities. Is the non-zero support on \textit{this particular} noncausal vertex necessary? In general, a given point inside the bipartite correlation polytope can admit several distinct convex decompositions over the vertices. This is true of the correlation in Eq.~\eqref{eq:quantumdist} as well, with one special feature (which can be verified by enumerating all convex decompositions of the correlation): namely, \textit{every} convex decomposition of this correlation contains the LGYNI vertex given by $x_1=\bar{a}_1+a_2$ and $x_2=a_1+\bar{a}_2$ as the noncausal vertex with the highest weight among all noncausal vertices in the decomposition. Hence, a non-zero support on this particular noncausal vertex is, indeed, necessary. The (non-zero) probabilistic weight associated with this vertex, however, depends on the particular convex decomposition.

How nonclassical is the correlation of Eq.~\eqref{eq:quantumdist}? To answer this question, we compute its antinomy weight, defined in Eq.~\eqref{eq:antinomyweight}. Before we compute it, however, it will be useful to express this correlation in a more general form, \textit{i.e.},
\begin{align}\label{eq:qform}
	\begin{pmatrix}
		0 & 0 & 0 & \frac{q}{2}\\
		0 & 0 & q & \frac{1-q}{2}\\		
		0 & q & 0&\frac{1-q}{2}\\
		1&1-q&1-q& \frac{q}{2}
	\end{pmatrix},
\end{align} 
where $q=\frac{1}{2}\left(1+\frac{1}{\sqrt{2}}\right)$. By numerically minimizing the overall weight on the antinomic\footnote{Same as `noncausal' in the bipartite case.} vertices over all convex decompositions of the correlation, we obtain 
\begin{align}
	r_a(P^{\rm Q}_{\vec{X}|\vec{A}})&=\min_{ \{q_v,\delta_{\vec{x},f_v(\vec{a})}\}_v}\sum_{v_*}q_{v_*}\nonumber\\
	&\approx0.134,
\end{align}
where the unique convex decomposition that minimizes the weight is given by 
\begin{align}\label{allcols}
	&\frac{q}{2}\begin{pmatrix}
		0&0&0&1\\
		0&0&1&0\\
		0&1&0&0\\
		1&0&0&0
	\end{pmatrix}+
	\frac{1-q}{2}\begin{pmatrix}
		0&0&0&0\\
		0&0&1&1\\
		0&1&0&0\\
		1&0&0&0
	\end{pmatrix}\nonumber\\+
	&\frac{1-q}{2}\begin{pmatrix}
		0&0&0&0\\
		0&0&1&0\\
		0&1&0&1\\
		1&0&0&0
	\end{pmatrix}+
	(1-q)
	\begin{pmatrix}
		0&0&0&0\\
		0&0&1&0\\
		0&0&0&0\\
		1&1&0&1
	\end{pmatrix}\nonumber\\+
	&(1-q)
	\begin{pmatrix}
		0&0&0&0\\
		0&0&0&0\\
		0&1&0&0\\
		1&0&1&1
	\end{pmatrix}+
	\frac{5q-4}{2}
	\begin{pmatrix}
		0&0&0&0\\
		0&0&1&0\\
		0&1&0&0\\
		1&0&0&1
	\end{pmatrix},
\end{align}
and where exactly one antinomic vertex (winning the canonical LGYNI game perfectly) appears with probability 
$\frac{5q-4}{2}\approx 0.134$ for $q=\frac{1}{2}(1+\frac{1}{\sqrt{2}})\approx 0.8536$. We know that this vertex is antinomic because it does not admit a realization with a process function: all bipartite classical processes are causal and they cannot violate causal inequalities \cite{OCB12}. Hence, this correlation requires \textit{at least} a $13.4\%$ weight on a antinomic vertex.\footnote{Although every antinomic vertex is noncausal, and conversely, in the bipartite case, this is not true generally, \textit{e.g.}, in the tripartite case. Hence, we refer to the vertex as an `antinomic' vertex (or a `nonclassical' vertex) when concerned with nonclassicality. Robustness of antinomy is \textit{not} intended as a measure of noncausality, but rather of nonclassicality. The difference becomes explicit in the tripartite case.}

The above analysis shows that the quantum distribution violates  the canonical GYNI and LGYNI inequalities (Eqs.~\eqref{cangyni}, \eqref{canlgyni}) because it is necessarily supported on a noncausal vertex that violates canonical GYNI up to $\frac{3}{4}$ and canonical LGYNI maximally (up to $1$). The optimal decomposition (\textit{i.e.}, the one minimizing the probabilistic weight on antinomic vertices) shows that the quantum process correlation lies within the polytope obtained by extending the causal polytope, \textit{i.e.}, by taking the convex hull of the causal/classical vertices and this additional noncausal/antinomic vertex.

\subsubsection{Causal inequality violations from a family of process matrices}

The values of $q\in[0,1]$ in Eq.~\eqref{eq:qform} that violate the  causal inequalities in Eqs.~\eqref{gyni} and \eqref{lgyni} are given by the constraints
\begin{align}
	p_{\rm gyni}&=\frac{1}{4}\left(q+q+\frac{q}{2}\right)=\frac{5}{8}q>\frac{1}{2},\\
	p_{\rm lgyni}&=\frac{1}{4}\left(1+q+q+\frac{q}{2}\right)=\frac{1}{4}+\frac{5}{8}q>\frac{3}{4},
\end{align}
which means $q\in (0.8,1]$ is necessary and sufficient for the violation of both inequalities. For $q\leq0.8$, any correlation of the form in Eq.~\eqref{eq:qform} is causal, \textit{i.e.}, it \textit{can} be expressed as a convex mixture over causal vertices.

The general form of the correlation in Eq.~\eqref{eq:qform} raises the question of whether it is realizable via process matrices for values of $q$ other than $\frac{1}{2}\left(1+\frac{1}{\sqrt{2}}\right)$. It is easy to show that the process matrix given by 
\begin{align}\label{eq:Wq}
	W(q):=&\frac{1}{4}\Bigg[\id^{\otimes4}+(2q-1)\nonumber\\
	&\left(Z^{I_1}Z^{O_1}Z^{I_2}\id^{O_2}+Z^{I_1}\id^{O_1}X^{I_2}X^{O_2}\right)\Bigg],
\end{align}
for any $q\in \left[\frac{1}{2}\left(1-\frac{1}{\sqrt{2}}\right),\frac{1}{2}\left(1+\frac{1}{\sqrt{2}}\right)\right]$ yields the correlation in Eq.~\eqref{eq:qform} for the local operations in Eqs.~\eqref{lo00}, \eqref{lo01}, \eqref{lo10}, and \eqref{lo11}. The case $q=\frac{1}{2}\left(1+\frac{1}{\sqrt{2}}\right)\approx0.8536$ studied in Ref.~\cite{BAF15} belongs to the family of process matrices $W(q)$ parameterized by $q$ in Eq.~\eqref{eq:Wq}. This correlation violates the GYNI and LGYNI inequalities for $q\in\left(0.8,\frac{1}{2}\left(1+\frac{1}{\sqrt{2}}\right)\right]$.

Note that for $q=1$ the correlation of Eq.~\eqref{eq:qform} is not achievable by process matrices: In this case, it corresponds to the noncausal vertex of Eq.~\eqref{eq:lgynivertex}, which, by Corollary \ref{cor:siboncyc}, is antinomic and therefore unachievable via any process matrix as a consequence of Theorem \ref{thm:det}. It remains an open question whether a process matrix can realize the correlation of Eq.~\eqref{eq:qform} for a value of $q\in\left(\frac{1}{2}\left(1+\frac{1}{\sqrt{2}}\right),1\right)$. More generally, one might ask for upper bounds on violations of causal inequalities---such as those of Eqs.~\eqref{gyni} and \eqref{lgyni}---achievable within the process-matrix framework \cite{Brukner15}. Some partial results along these lines exist: for example, it has been shown in Ref.~\cite{BAB19} that for the GYNI inequality of Eq.~\eqref{gyni}, any finite-dimensional process-matrix realization with $d:=d_{I_1}d_{O_1}d_{I_2}d_{O_2}$ cannot exceed the value $1-\frac{1}{d+1}$. In the case $d=2^4=16$ that is of interest here, we have that no process matrix can exceed the value $\frac{16}{17}\approx 0.941$. More recently, Liu and Chiribella \cite{LC24} developed a general technique to bound process-matrix correlations, obtaining in particular an upper bound of $0.7592$ on $p_{\rm gyni}$ (not known to be saturated, but already an improvement over Ref.~\cite{BAB19}) and an upper bound of $0.8194$ on $p_{\rm lgyni}$ (known to be saturated by process matrices \cite{BAF15}). Finally, no process matrix correlation can achieve a  maximal violation of any facet causal inequality in the $(2,2,2)$ scenario (\textit{cf.}~Appendix \ref{app:bipartite}) because these maximal violations require antinomic vertices unachievable with process matrices (Corollary \ref{cor:siboncyc}).

\section{Tripartite correlations}\label{sec:5}
Having illustrated the key ideas underlying antinomicity in the bipartite case, we now consider the tripartite scenario of type $(3,2,2)$, \textit{i.e.}, three parties with binary inputs and outputs. Unlike the bipartite case, this scenario provides \textit{new} examples of nonclassicality witnesses beyond known causal inequalities. The correlation polytope for this scenario is also defined by non-negativity and normalization, \textit{i.e.},
\begin{align}
	&\forall \vec{x},\vec{a}:&p(x_1,x_2,x_3|a_1,a_2,a_3)\geq 0,\\
	&\forall \vec{a}:& \sum_{x_1,x_2,x_3}p(x_1,x_2,x_3|a_1,a_2,a_3)=1.
\end{align}
Like the bipartite case, the vertices of this polytope are also deterministic, \textit{i.e.}, they are given by all possible functions $f:\vec{A}\rightarrow \vec{X}$.

\subsection{Signalling classes}
Every vertex in the tripartite case has an associated signalling graph. While the bipartite case is simple enough (with $4$ signalling graphs, including the non-signalling case), already in the tripartite case, many more signalling graphs are possible. Although we defined signalling graphs as labelled directed graphs, here it will be useful to coarse-grain this description further: that is, at the most coarse-grained level, we will classify vertices according to the  \textit{unlabelled} directed graphs underlying their signalling graphs. The unlabelled directed graph underlying a signalling graph will be called the \textit{signalling class} of the signalling graph and of any vertex associated with the signalling graph. 
\begin{figure*}
	\centering
	\includegraphics[scale=0.3]{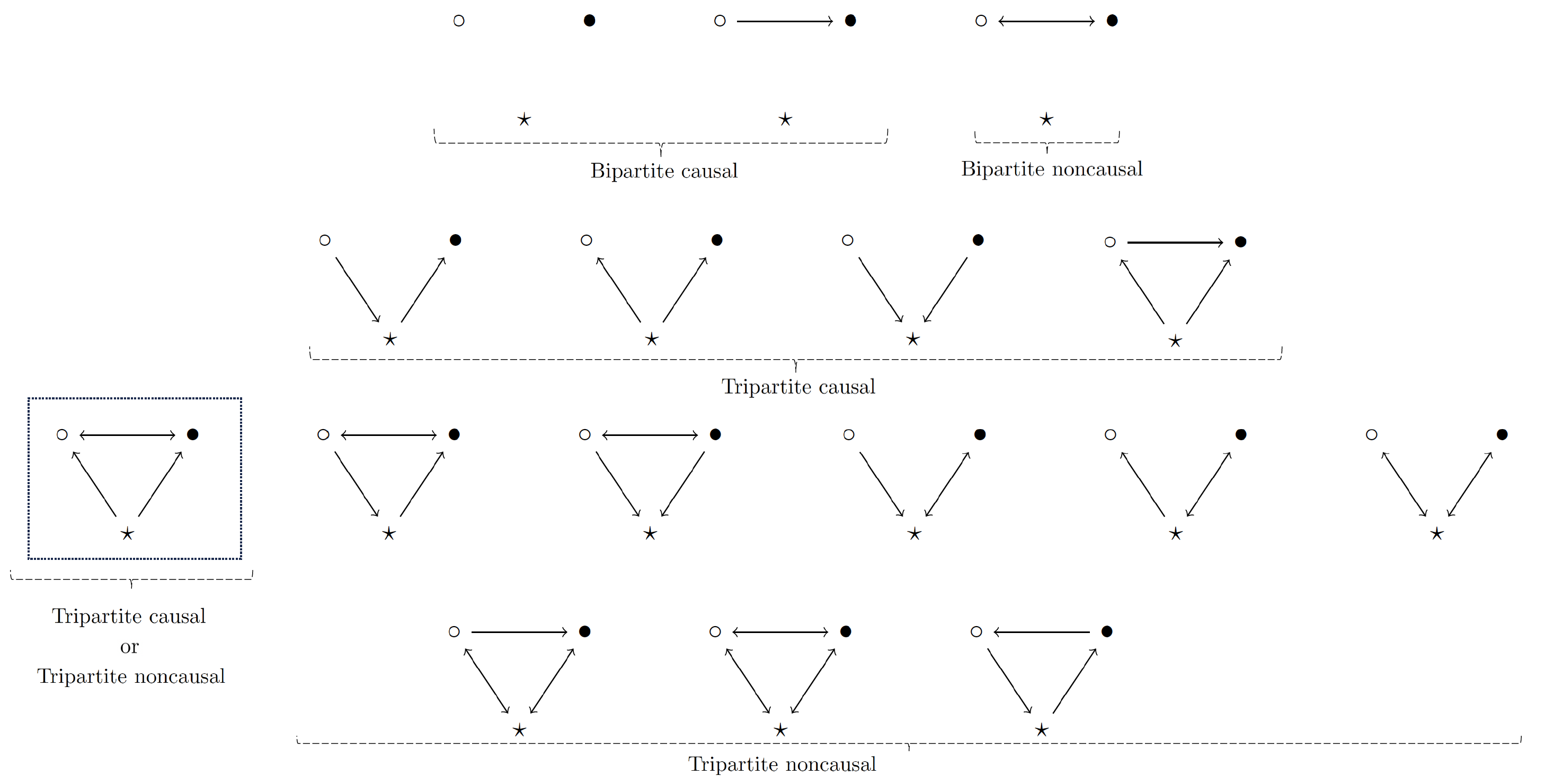}
	\caption{Tripartite signalling classes. The boxed signalling class is ambiguous: some vertices in this class are causal (\textit{e.g.}, a classical switch), while the others are noncausal.}\label{fig:DG3}
\end{figure*}
These signalling classes can be classified according to the following categories (\textit{cf.}~Fig.~\ref{fig:DG3}):
\begin{enumerate}
	\item Non-signalling: The directed graph has three nodes and no edges, \textit{i.e.}, no party signals to any other. All the vertices in this class are causal.
	\item Bipartite causal signalling: One of the three nodes in the directed graph is a by-stander, \textit{i.e.}, with no incoming or outgoing edges, while the other two nodes communicate causally, \textit{i.e.}, they have a unidirectional edge between them.
	All the vertices in this class are causal.
	\item Bipartite noncausal signalling: One of the three nodes in the directed graph is a by-stander while the other two nodes communicate noncausally, \textit{i.e.}, they have edges in both directions between them. All the vertices in this class are noncausal.
	\item Tripartite fixed-order causal signalling: The directed graph is connected (\textit{i.e.}, no by-stander nodes) and \textit{acyclic}, \textit{i.e.}, it is a connected directed acyclic graph (DAG). All the vertices in this class are fixed-order causal.
	\item Tripartite (dynamical-order) causal or Tripartite noncausal signalling: The directed graph has a root node that is connected to two nodes with a cycle between them. This ambiguous case is represented in the boxed signalling class of Fig.~\ref{fig:DG3}: vertices in this class may be (dynamical-order) causal or noncausal. 
	
	Note that the noncausal vertices in this class cannot be achieved by the process-matrix framework: this is ruled out by Corollary \ref{cor:causalityofvertices}, \textit{i.e.}, the conjunction of our Theorem \ref{thm:det} with Theorem 6 of Ref.~\cite{TB22} (reproduced as Theorem \ref{thm:chordless} above). Hence, vertices in this class achievable by the process-matrix framework are all (dynamical-order) causal. On the other hand, the remaining vertices in this class are (by Theorem \ref{thm:det}) antinomic and thus noncausal.
	
	\item Tripartite noncausal signalling: The directed graph is connected, with at least one cycle, \textit{i.e.}, a connected directed cyclic graph (DCG), and is not of the ambiguous type above. All the vertices in this class are noncausal.
\end{enumerate}
Note that this classification of signalling classes holds in the tripartite case beyond the $(3,2,2)$ scenario, \textit{i.e.}, it holds for all $(3,M,K)$ scenarios for $M,K\geq 2$. 

In Appendix \ref{app:322vertices}, we explicitly characterize all the vertices in the $(3,2,2)$ scenario. We use this characterization in the following subsection to identify the signalling class of nomic vertices that are noncausal in this scenario.

\subsection{Nomic vertices in the $(3,2,2)$ scenario}

Given the exhaustive classification of all the tripartite vertices in the $(3,2,2)$ scenario in Appendix \ref{app:322vertices}, a natural question to ask is: Which of these vertices are nomic, \textit{i.e.}, admit process function realizations? All causal vertices are nomic. However, unlike the bipartite case, there exists a noncausal vertex that is also nomic: namely, the deterministic correlation realized by the AF/BW process that is known to violate a causal inequality maximally \cite{BW16, KB22}.

\subsubsection{The AF/BW process}
The AF/BW process, $P^{\rm afbw}_{\vec{I}|\vec{O}}$, is given by the following matrix of probabilities:
\begin{align}\label{eq:afbwproc}
	P^{\rm afbw}_{\vec{I}|\vec{O}}=
	\begin{pmatrix}
		1    &  0  &  0  &  0  &  0  &  0  &  0  &1\\
		0	 &  0  &  1  &  1  &  0  &  0  &  0  &0\\
		0	 &  0  &  0  &  0 &   1  &  0  &  1  &0\\
		0  	 &  0  &  0  &  0  &  0  &  0  &  0  &0\\
		0    &  1  &  0  &  0  &  0  &  1  &  0  &0\\
		0    &  0  &  0  &  0  &  0  &  0  &  0  &0\\
		0    &  0  &  0  &  0  &  0  &  0  &  0  &0\\
		0	 &  0  &  0  &  0  &  0  &  0  &  0  &0
	\end{pmatrix}.
\end{align}
In a more compact form, this process function is given by 
\begin{align}\label{eq:afbwfn}
	i_1=\bar{o}_2o_3, i_2=\bar{o}_3o_1,i_3=\bar{o}_1o_2.
\end{align}
This process can win the following tripartite game with certainty: on receiving the inputs $a_1,a_2,a_3$ uniformly randomly, the parties must produce outputs $x_1,x_2,x_3$ that satisfy the constraints
\begin{align}
	&x_1=a_3, x_2=a_1, x_3=a_2,\nonumber\\
	&\textrm{ when } {\rm maj}(a_1,a_2,a_3)=0,\nonumber\\
	&x_1=\bar{a}_2, x_2=\bar{a}_3, x_3=\bar{a}_1,\nonumber\\
	&\textrm{ when } {\rm maj}(a_1,a_2,a_3)=1,
\end{align}
where ${\rm maj}(a_1,a_2,a_3)$ assigns value $0$ if the majority vote among the three inputs is $0$, and $1$ otherwise. No causal strategy can win this game perfectly \cite{BW16}. Hence, the AF/BW process can maximally violate the following causal inequality:
\begin{align}\label{eq:afbwcausineq}
	p_{\rm afbw}\nonumber\\
	:=\frac{1}{2}\sum_{\vec{x},\vec{a}}&p(x_1,x_2,x_3|a_1,a_2,a_3)\nonumber\\
	&\Big(\delta_{x_1,a_3}\delta_{x_2,a_1}\delta_{x_3,a_2}\delta_{{\rm maj}(a_1,a_2,a_3),0}\nonumber\\
	&+\delta_{x_1,\bar{a}_2}\delta_{x_2,\bar{a}_3}\delta_{x_3,\bar{a}_1}\delta_{{\rm maj}(a_1,a_2,a_3),1}\Big)\\
	\leq\frac{3}{4}.
\end{align}

The maximal violation of this inequality occurs for the following tripartite noncausal vertex:
\begin{align}\label{eq:afbwcorr}
	P_{\vec{X}|\vec{A}}=\begin{pmatrix}
		1    &  0  &  0  &  0  &  0  &  0  &  0  &1\\
		0	 &  0  &  1  &  1  &  0  &  0  &  0  &0\\
		0	 &  0  &  0  &  0 &   1  &  0  &  1  &0\\
		0  	 &  0  &  0  &  0  &  0  &  0  &  0  &0\\
		0    &  1  &  0  &  0  &  0  &  1  &  0  &0\\
		0    &  0  &  0  &  0  &  0  &  0  &  0  &0\\
		0    &  0  &  0  &  0  &  0  &  0  &  0  &0\\
		0	 &  0  &  0  &  0  &  0  &  0  &  0  &0
	\end{pmatrix}
\end{align}
The correlations between $\vec{X}$ and $\vec{A}$ specified by this vertex are exactly the same as those specified by $P^{\rm afbw}_{\vec{I}|\vec{O}}$ between $\vec{I}$ and $\vec{O}$. Hence, this vertex can be achieved via the AF/BW process if the parties make the following local interventions:
\begin{align}\label{eq:afbwlo}
	p(x_k,o_k|a_k,i_k)=\delta_{x_k,i_k}\delta_{o_k,a_k},
\end{align}
\textit{i.e.}, pass $i_k$ to $x_k$ and $a_k$ to $o_k$.

In the tripartite case where the classical process has binary inputs and outputs, Baumeler and Wolf \cite{BW16} show that the AF/BW process is essentially the only process function that can violate causal inequalities. It belongs to a family of $64$ process functions that can be seen to arise from $P^{\rm afbw}_{\vec{I}|\vec{O}}$ (\textit{cf.}~Eq.~\eqref{eq:afbwproc}) as follows: considering local relabellings (of inputs and outputs) $L_k, L'_k$ ($k\in\{1,2,3\}$), where $L_k,L'_k\in\{\id,\mathbb{F}\}$ and $\id=\begin{pmatrix}
	1&0\\
	0&1
\end{pmatrix}$, $\mathbb{F}=\begin{pmatrix}
	0&1\\
	1&0
\end{pmatrix}$, the $64$ process functions are given by 
\begin{align}
	\forall k, L_k, L'_k: (L_1\otimes L_2\otimes L_3) P^{\rm afbw}_{\vec{I}|\vec{O}}(L'_1\otimes L'_2\otimes L'_3).
\end{align}

Each version of the AF/BW process violates the corresponding causal inequality (a variant of Eq.~\eqref{eq:afbwcausineq}) maximally by perfectly realizing the corresponding noncausal vertex (a variant of Eq.\eqref{eq:afbwcorr}). Hence, $64$ of the tripartite noncausal vertices admit deterministically consistent realization via some version of the AF/BW process under local operations of Eq.~\eqref{eq:afbwlo}.

\subsubsection{Sufficient condition for antinomicity of vertices}
In our enumeration of the vertices in the $(3,2,2)$ scenario in Appendix \ref{app:322vertices}, we note that Corollary \ref{cor:siboncyc}---which follows from our Theorem \ref{thm:det} and Theorem 2 of Ref.~\cite{TB22} (reproduced as Theorem \ref{thm:siboncyc} above)---rules out all vertices whose signalling graphs fail the siblings-on-cycles property as nomic vertices. This leaves us with two tripartite signalling classes, namely, $\{\{\circ\leftarrow\star\}, \{\star\rightarrow\bullet\}, \{\circ\leftrightarrow\bullet\}\}$ and $\{\{\circ\leftrightarrow\star\}, \{\star\leftrightarrow\bullet\}, \{\circ\leftrightarrow\bullet\}\}$ as potential candidates for nomic vertices that are also noncausal. Furthermore, as we note in our enumeration in Appendix \ref{app:322vertices}, Corollary \ref{cor:causalityofvertices}---which follows from our Theorem \ref{thm:det} and Theorem 6 of Ref.~\cite{TB22} (reproduced as Theorem \ref{thm:chordless} above)---implies that  process matrices cannot realize noncausal vertices from the signalling class $\{\{\circ\leftarrow\star\}, \{\star\rightarrow\bullet\}, \{\circ\leftrightarrow\bullet\}\}$ (\textit{cf.}~Fig.~\ref{fig:DG3}).

As such, in the $(3,2,2)$ correlational scenario, we have the following sufficient condition for the antinomicity of vertices of the correlation polytope:
\begin{corollary}[Sufficient condition for antinomicity of vertices in the $(3,2,2)$ correlational scenario.]\label{cor:tripartiteantinomy}
	All noncausal vertices of the correlation polytope in the $(3,2,2)$ scenario that fall into signalling classes other than $\{\{\circ\leftrightarrow\star\}, \{\star\leftrightarrow\bullet\},\{\circ\leftrightarrow\bullet\}\}$ are antinomic, \textit{i.e.}, they cannot be realized by process functions.
\end{corollary}
This means that the only nomic vertices that are also noncausal must belong to the signalling class $\{\{\circ\leftrightarrow\star\}, \{\star\leftrightarrow\bullet\},\{\circ\leftrightarrow\bullet\}\}$ (\textit{i.e.}, exactly one of the signalling classes in Fig.~\ref{fig:DG3}). However, the converse does not hold: there exist antinomic vertices with this signalling graph, \textit{e.g.}, the vertex given by $x_1=a_2\oplus a_3$, $x_2=a_3\oplus a_1$, and $x_3=a_1\oplus a_2$ is not realizable by a process function since such a realization requires that $g:=(g_k:\vec{A}_{\backslash k}\rightarrow\vec{I})$, given by $g_1(a_2,a_3):=a_2\oplus a_3$, $g_2(a_3,a_1):=a_3\oplus a_1$, $g_3(a_1,a_2):=a_1\oplus a_2$ (where $\vec{I}:=\{0,1\}^3$), is a process function, which it isn't: a simple identity intervention $o_k=h_k(i_k)=i_k$ by each party $S_k$ results in the set of constraints $i_1=i_2\oplus i_3, i_2=i_3\oplus i_1, i_3=i_1\oplus i_2$, which do not admit a unique solution, \textit{i.e.}, there is no unique fixed point for $g\circ h$, where $h:=(h_1,h_2,h_3)$.\footnote{All $\vec{i}\in\{(0,0,0), (0,1,1), (1,0,1), (1,1,0)\}$ satisfy this system of equations and constitute (non-unique) fixed points.}

\subsection{Nonclassicality beyond causal inequalities}
The first thing to note about bounds on correlations that arise from assuming nomicity is that these bounds also apply to (but may not be saturated by) causal correlations. Hence, any witness of antinomicity we construct will also be a witness of noncausality but not conversely. 

We now construct a simple example of a nonclassicality witness based on a tripartite game that we refer to as the `Guess Your Neighbour's Input or NOT' (GYNIN) game. It is defined as follows: on receiving the inputs $\{a_k\}_{k=1}^3$ uniformly randomly, the parties must 
output $\{x_k\}_{k=1}^3$ such that 
\begin{align}\label{eq:gyninwin}
	&(x_1,x_2,x_3)=(a_3,a_1,a_2),\nonumber\\
	\textrm {OR }&(x_1,x_2,x_3)=(\bar{a}_3, \bar{a}_1, \bar{a}_2).
\end{align}

The success probability in this game is given by 
\begin{align}
	&p_{\rm gynin}\nonumber\\
	:=&\frac{1}{8}\sum_{\vec{x},\vec{a}}p(\vec{x}|\vec{a})\Big(\delta_{x_1,a_3}\delta_{x_2,a_1}\delta_{x_3,a_2}+\delta_{x_1,\bar{a}_3}\delta_{x_2,\bar{a}_1}\delta_{x_3,\bar{a}_2}\Big).
\end{align}
The following inequality serves as our witness of antinomicity:
\begin{align}\label{eq:nonclasswitness}
	p_{\rm gynin}\leq\frac{5}{8}.
\end{align}
\begin{proof}
	Since any correlation $p(\vec{x}|\vec{a})$ achievable by classical processes within the deterministic-extrema polytope can be decomposed as a convex mixture of deterministic correlations achievable via its vertices (\textit{i.e.}, process functions), by convexity the upper bound on $p_{\rm gynin}$ follows from maximizing it over process functions. In particular, $p_{\rm gynin}=1$ is impossible to achieve via process functions: achieving this value with a process function requires that the associated $p(\vec{x}|\vec{a})$ is deterministic and is given by either $\delta_{x_1,a_3}\delta_{x_2,a_1}\delta_{x_3,a_2}$ or $\delta_{x_1,\bar{a}_3}\delta_{x_2,\bar{a}_1}\delta_{x_3,\bar{a}_2}$, neither of which is achievable via process functions since their signalling graphs fail to satisfy the siblings-on-cycles property (\textit{cf.}~Corollary \ref{cor:siboncyc}). We therefore have 
	\begin{align}
		p_{\rm gynin}<1
	\end{align}
	for nomic correlations. The precise bound on this quantity follows from two facts: firstly, from Theorem \ref{thm:det}, that any deterministic correlation given by the function $f:\vec{A}\rightarrow\vec{X}$ that is realizable by a process function admits a faithful realization by the process function $g:\vec{A}\rightarrow\vec{I}$, where   $I_k=\{g_k(\vec{a}_{\backslash k})\}_{\vec{a}_{\backslash k}}$ and $x_k=f_k(\vec{a})=f'_k(a_k,g(\vec{a}_{\backslash k}))$; secondly, that in the case of a deterministic correlation in the $(3,2,2)$ scenario that aims to maximize $p_{\rm gynin}$, we have that we can, without loss of generality, assume $I_k\in\{0,1\}$ (for all $k\in\{1,2,3\}$) for the associated process function.\footnote{This follows from the definition of the GYNIN game (Eq.~\eqref{eq:gyninwin}), which requires that each outcome $x_k$ is a function only of remote settings, \textit{i.e.}, $x_1=f_1(a_3), x_2=f_2(a_1),x_3=f_3(a_2)$, where the $f_k$'s are either all identity functions or all bit-flips.} Hence, all we need to do is to maximize the value of $p_{\rm gynin}$ over deterministic correlations realizable by all tripartite process functions where the parties have binary inputs and outputs. Noting that the $x_k$ should be independent of $a_k$ in any deterministic correlation that hopes to maximize $p_{\rm gynin}$, we can assume that the local interventions set $x_k=i_k$ $(=g_k(\vec{a}_{\backslash k}))$ in the process function realization of the correlation. Using the Baumeler-Wolf characterization \cite{BW16} of all the $64$ variants of the AF/BW process function and maximizing $p_{\rm gynin}$ over them, we obtain $p_{\rm gynin}\leq\frac{5}{8}$, where the upper bound is saturated by the canonical AF/BW process of Eq.~\eqref{eq:afbwfn} (where $\vec{o}=\vec{a}$ and $\vec{i}=\vec{x}$).
\end{proof}
Hence, the inequality $p_{\rm gynin}\leq\frac{5}{8}$ is saturated by the AF/BW correlation of Eq.~\eqref{eq:afbwcorr}. To see this, note that the AF/BW correlation is given by $x_1=\bar{a}_2a_3,x_2=\bar{a}_3a_1,x_3=\bar{a}_1a_2$. This means it satisfies the winning condition of Eq.~\eqref{eq:gyninwin} for the cases where $(\bar{a}_1,\bar{a}_2,\bar{a}_3)\in\{(1,1,1), (0,0,0),(1,1,0),(1,0,1), (0,1,1)\}$ and fails to satisfy the winning condition for $(\bar{a}_1,\bar{a}_2,\bar{a}_3)\in\{(0,0,1), (0,1,0), (1,0,0)\}$. 

The causal bound on the success probability is given by 
\begin{align}
	p_{\rm gynin}\leq\frac{1}{2},
\end{align}
so there is a gap between the causal and the classical bound in this case. The causal bound follows from the fact that in any causal strategy, one of the parties (say $S_3$) is in the global past and so it must guess the input of its neighbour (say, $S_2$) uniformly randomly (\textit{i.e.}, $p(x_3)=\frac{1}{2}\delta_{x_3,a_2}+\frac{1}{2}\delta_{x_3,\bar{a}_2}$); assuming $S_3$ signals to $S_1$, we have that $S_1$ can guess the input of $S_3$ (or its negation) perfectly, \textit{i.e.}, $p(x_1)=\delta_{x_1,a_3}$ (or $p(x_1)=\delta_{x_1,\bar{a}_3}$); assuming $S_1$ signals to $S_2$, $S_2$ can guess the input of $S_1$ (or its negation) perfectly, \textit{i.e.}, $p(x_2)=\delta_{x_2,a_1}$ (or $p(x_2)=\delta_{x_2,\bar{a}_1}$). Hence, for each $\vec{a}$,  the three parties will satisfy the winning condition of the game (\textit{i.e.}, Eq.~\eqref{eq:gyninwin}) with probability at most $\frac{1}{2}$; this gives us the overall causal bound of $p_{\rm gynin}\leq\frac{1}{2}$.

Furthermore, the inequality of Eq.~\eqref{eq:nonclasswitness} is maximally violated (up to probability $1$) by a correlation achievable by the BFW process \cite{BFW14}, \textit{i.e.}, by
\begin{align}
	P_{\vec{X}|\vec{A}}=\begin{pmatrix}
		1/2&  0  &  0  &  0  &  0  &  0  &  0  &1/2\\
		0	 &  0  &1/2&  0  &  0  &1/2&  0  &0\\
		0	 &  0  &  0  &1/2&1/2&  0  &  0  &0\\
		0  	 &1/2&  0  &  0  &  0  &  0  &1/2&0\\
		0    &1/2&  0  &  0  &  0  &  0  &1/2&0\\
		0    &  0  &  0  &1/2&1/2&  0  &  0  &0\\
		0    &  0  &1/2&  0  &  0  &1/2&  0  &0\\
		1/2&  0  &  0  &  0  &  0  &  0  &  0  &1/2
	\end{pmatrix}.
\end{align}
The BFW process is given by 
\begin{align}
	P^{\rm bfw}_{\vec{I}|\vec{O}}=\begin{pmatrix}
		1/2&  0  &  0  &  0  &  0  &  0  &  0  &1/2\\
		0	 &  0  &1/2&  0  &  0  &1/2&  0  &0\\
		0	 &  0  &  0  &1/2&1/2&  0  &  0  &0\\
		0  	 &1/2&  0  &  0  &  0  &  0  &1/2&0\\
		0    &1/2&  0  &  0  &  0  &  0  &1/2&0\\
		0    &  0  &  0  &1/2&1/2&  0  &  0  &0\\
		0    &  0  &1/2&  0  &  0  &1/2&  0  &0\\
		1/2&  0  &  0  &  0  &  0  &  0  &  0  &1/2
	\end{pmatrix},
\end{align}
and it maximally violates the GYNIN inequality under the local interventions 
\begin{align}
	p(x_k,o_k|a_k,i_k)=\delta_{x_k,i_k}\delta_{o_k,a_k}.
\end{align}
Hence, the GYNIN inequality illustrates gaps between causal, nomic, and quantum process correlations (a special case of which are the probabilistically consistent correlations). This inequality thus reveals the sense in which the BFW process can lead to nonclassical correlations. We therefore have the following:
\begin{theorem}[$\mathcal{DC}\subsetneq\mathcal{PC}$]\label{thm:dcpcgap}
	The set of deterministically consistent (\textit{i.e.}, nomic) correlations is strictly contained within the set of probabilistically consistent correlations, \textit{i.e.}, $\mathcal{DC}\subsetneq\mathcal{PC}$.
\end{theorem}
Hence, we have the full set of strict inclusions between the hierarchy of correlation sets defined in Section \ref{subsec:hierarchy}, given by 
\begin{align}
	\mathcal{DC}\subsetneq \mathcal{PC}\subsetneq\mathcal{QP}\subsetneq\mathcal{qC}.
\end{align}
We also have an antinomicity witness that can separate causal, classical, and antinomic correlations, \textit{i.e.},
\begin{align}\label{eq:sepbounds}
	p_{\rm gynin}\overset{\rm causal}{\leq}\frac{1}{2}\overset{\rm classical}{\leq}\frac{5}{8}\overset{\rm antinomic}{\leq} 1.
\end{align}
We depict the inclusion relations between different sets of correlations in Figure \ref{fig:setinclusions}.
\begin{figure}
	\centering
	\includegraphics[scale=0.36]{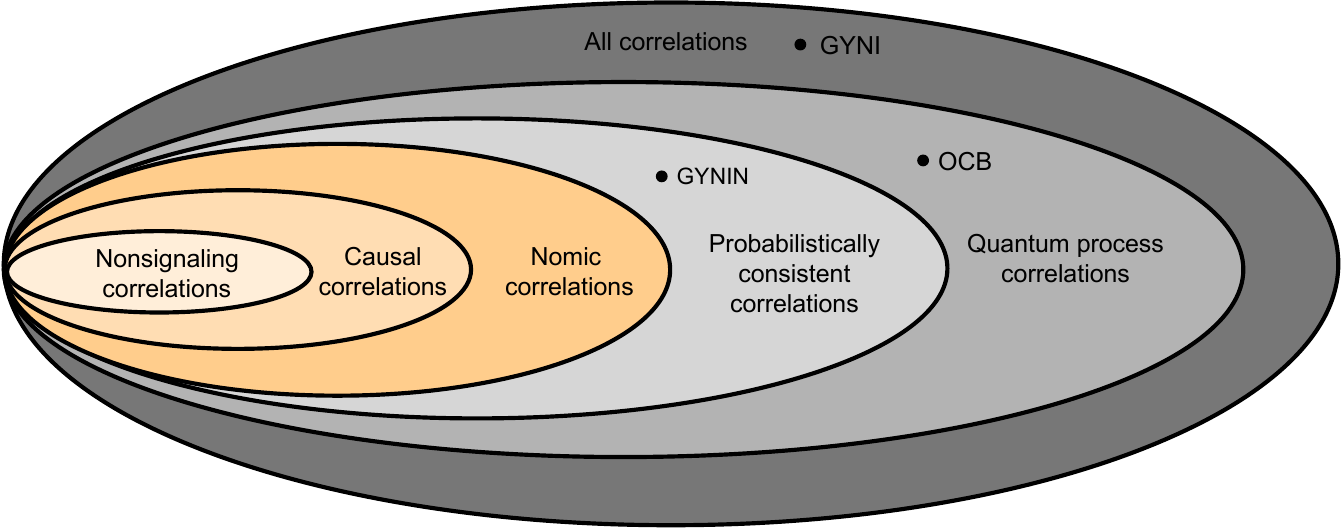}
	\caption{Inclusion relations between different sets of correlations. The point GYNIN refers to the correlation for which $p_{\rm gynin}=1$, achieved by the BFW process \cite{BFW14}. The point OCB refers to the correlation that wins the Oreshkov-Costa-Brukner game with probability $\frac{2+\sqrt{2}}{4}$ \cite{OCB12}. The point GYNI \cite{BAF15} refers to the correlation that wins the GYNI game with probability $1$. Antinomicity witnesses such as Eq.~\eqref{eq:nonclasswitness} separate nomic correlations from the rest.}\label{fig:setinclusions}
\end{figure} 

\section{Discussion}\label{sec:discussion}

We have defined a notion of nonclassicality ---namely, \textit{antinomicity}---for correlations without causal order. We argued that this notion of nonclassicality is a natural generalization of the failure of local causality in (non-signalling) Bell scenarios to multipartite scenarios without any non-signalling constraints, and proposed \textit{antinomy weight} as a measure of antinomicity of a correlation.  We also demonstrated a strict hierarchy between four sets of correlations, \textit{i.e.}, deterministically consistent (nomic), probabilistically consistent, quantum process, and quasi-consistent correlations (Figure \ref{fig:setinclusions}). We characterized vertices of the correlation polytope that are classically realizable (Theorem \ref{thm:det}) in any correlational scenario with single-round communication and made connections with admissible causal structures in the process-matrix framework, providing a sufficient condition for antinomicity of deterministic correlations \cite{TB22} (Corollary \ref{cor:siboncyc}). To illustrate the general ideas and their connections with prior literature, we analyzed bipartite $(2,2,2)$ and tripartite $(3,2,2)$ correlational scenarios is extensive detail and studied an antinomicity witness based on the GYNIN game that goes beyond causal inequalities.

Several new questions arise as a result of our investigation and we highlight some of them here. 

Firstly, we need to  characterize the bounds on correlations that follow from the assumption of nomicity in generic correlational scenarios. Theorem \ref{thm:det} and Corollary \ref{cor:siboncyc} go quite some way towards characterizing vertices of the nomic polytope. In particular, we used Theorem \ref{thm:det} to prove the classical (nomic) bound on the antinomicity witness of Eq.~\eqref{eq:nonclasswitness} based on the GYNIN game (\textit{cf.}~Eq.~\eqref{eq:gyninwin}). Of particular interest are tight bounds that are not saturated by causal correlations (\textit{i.e.}, they are not tight causal bounds) and are, furthermore, violated by process matrices. Our antinomicity witness provides such an example (\textit{cf.},~Eq.~\eqref{eq:sepbounds}) and it would be interesting and important to identify other games or antinomicity witnesses which can demonstrate such strict separations between different sets of correlations. Since the first appearance of our results, a technique for obtaining Tsirelson-type bounds on causal inequality violations in the  process-matrix framework has been proposed by Liu and Chiribella \cite{LC24}. It would be interesting to investigate whether this technique can also provide explicit bounds on the antinomicity achievable in the process-matrix framework, \textit{e.g.}, by looking at antinomicity witnesses (unlike the GYNIN witness we presented) where process-matrix correlations \textit{cannot} achieve the algebraic maximum. Furthermore, there could be qualitatively different types of nonclassicality depending on the types of vertices that necessarily participate in the support of an antinomic correlation. The antinomy weight, for example, doesn't discriminate between different types of antinomic vertices and it might be interesting to develop finer classifications of antinomicity as well as corresponding measures.
We should also note a key distinction between nonclassicality witnessed by Bell inequality violations in Bell scenarios and the nonclassicality witnessed via antinomicity in causally unrestricted correlational scenarios. In Bell scenarios, the causal structure is fixed \textit{a priori} for both classical and quantum (or post-quantum) causal models and the violation of a Bell inequality certifies the inability of a classical causal model \textit{with that} causal structure to reproduce correlations that can arise in quantum (or a post-quantum) theory under the same causal structure. Similarly, in more general network scenarios inspired by Bell nonlocality \cite{WSF19}, classical causal compatibility constraints generalizing Bell inequalities are also defined relative to an \textit{a priori} causal structure; this is in keeping with the usual philosophy of assessing claims of nonclassicality relative to a presumed causal structure \cite{SFK20}. In the correlational scenarios we envisage, by contrast, there is no \textit{a priori} assumption of a causal structure---not even that it respects a definite causal order \cite{OCB12}---and the violation of a nomic inequality such as Eq.~\eqref{eq:nonclasswitness} certifies the inability of \textit{any} classical split-node causal model \cite{BLO19,BLO21} to reproduce the correlations that can arise in the process-matrix framework (or beyond).

Secondly, it has been conjectured that any process that is physically realizable without post-selection would be unitarily extendible \cite{AFN17}. At present, all processes with indefinite causal order for which a potential realization is known, be it with gravity \cite{ZCP19, MSY23} or time-delocalized subsystems \cite{Oreshkov19, WBO23}, are unitarily extendible, and we do not know of any example that displays antinomicity. Hence, we need to address the following question: can correlations realizable via unitarily-extendible process matrices display antinomicity?  If so, one important implication would be a concrete demonstration that antinomicity is of interest within the cyclic quantum causal models framework \cite{BLO21}, which is applicable to the case of unitarily-extendible process matrices: that is, there exists \textit{some} admissible causal structure \cite{TB22} (associated to a unitarily-extendible process matrix) relative to which there exist process-matrix correlations that are unachievable via \textit{any} classical split-node causal model (regardless of its causal structure). If not, this would mean that antinomicity of a process-matrix correlation can operationally certify the fact that the underlying process matrix is not unitarily-extendible. In either case, the answer to the above question would have important implications for the device-independent analysis of correlations without causal order and their realizability in the process-matrix framework.

Finally, we know that the AF/BW process \cite{BW16} is a resource for perfect local discrimination of an ensemble of product quantum states---the SHIFT ensemble---that displays the phenomenon of quantum nonlocality without entanglement (QNLWE) \cite{KB22}. By the lights of our notion of classicality, this is a task that requires noncausal-but-still-classical resources with indefinite causal order. The question then arises: how does one identify information-theoretic tasks exploiting indefinite causal order where antinomic correlations always provide an advantage over nomic ones? Are there independently interesting examples of such tasks in quantum information theory, just as QNLWE provides an example where classical noncausality is a resource \cite{KB22}? This would help clarify the operational interpretation of antinomicity beyond that offered by the GYNIN game.

Our work here represents only the first step in a broader research programme that aims at characterizing the nonclassicality of multipartite correlations achievable within and outside the process-matrix framework. It is our hope that antinomicity of multipartite correlations will be as fundamental and useful a characterization of nonclassicality for correlations without causal order as Bell nonlocality is for non-signalling correlations.

\ack
R.K.~thanks Ämin Baumeler, Hippolyte Dourdent, Tobias Fritz, and Andreas Winter for discussions that informed his thinking while pursuing this project. This work was made possible through the support of the ID\# 62312 grant from the John Templeton Foundation, as part of the project \href{https://www.templeton.org/grant/the-quantum-information-structure-of-spacetime-qiss-second-phase}{`The Quantum Information Structure of Spacetime' (QISS)}. The opinions expressed in this work are those of the author(s) and do not necessarily reflect the views of the John Templeton Foundation. R.K.~was supported by the Chargé de Recherche fellowship of the Fonds de la Recherche Scientifique (F.R.S.-FNRS), Belgium and the Program of Concerted Research Actions (ARC) of the Université libre de Bruxelles. O.O.~is a Research Associate of the Fonds de la Recherche Scientifique (F.R.S.-FNRS), Belgium. This work also received support from the French government under the France 2030 investment plan, as part of the Initiative d'Excellence d'Aix-Marseille Université-A*MIDEX, AMX-22-CEI-01.

\bibliographystyle{apsrev4-2}
\bibliography{masterbibfilev2.bib}

\appendix
\section{Bipartite causal inequalities and their maximal violations}\label{app:bipartite}
We prove some general properties of the correlation polytope in the $(2,2,2)$ correlational scenario.

\subsection{All the different versions of GYNI and LGYNI games}

\textit{The $16$ versions of GYNI:} Note that party $S_1$ (with local access to $a_1$) can access information about $a_2$ in four distinct ways, \textit{i.e.},
\begin{align}
	x_1&=a_2,\\
	x_1&=\bar{a}_2,\\
	x_1&=a_1\oplus a_2,\\
	x_1&=a_1\oplus \bar{a}_2.
\end{align}
Similarly, party $S_2$ (with local access to $a_2$) can access information about $a_1$ in four distinct ways, \textit{i.e.},
\begin{align}
	x_2&=a_1,\\
	x_2&=\bar{a}_1,\\
	x_2&=a_1\oplus a_2,\\
	x_2&=\bar{a}_1\oplus a_2.
\end{align}
Hence, there are $4\times 4=16$ ways in which the two parties can be asked to guess their neighbour's input, each corresponding to a different version of the GYNI game. These can be easily parameterized, as in Ref.~\cite{BAF15}, via four-bit strings $(\alpha_0,\alpha_1,\beta_0,\beta_1)\in\{0,1\}^4$ so that each such string corresponds to a version of the GYNI game as follows:
\begin{align}\label{paramgyni}
	&x_1\oplus \alpha_1 a_1\oplus \alpha_0=a_2,\nonumber\\
	&x_2\oplus \beta_1 a_2\oplus\beta_0=a_1. 
\end{align}
Here $(\alpha_0,\alpha_1,\beta_0,\beta_1)=(0,0,0,0)$ is the canonical version in Eq.~\eqref{cangyni}.

\textit{The $16$ versions of LGYNI:} Here party $S_1$ (with local access to $a_1$) needs to access information about $a_2$ if $a_1$ takes a particular value, but can freely choose $x_1$ otherwise. This can be done in four ways, \textit{i.e.},
\begin{align}
	a_1(x_1\oplus a_2)&=0,\\
	a_1(x_1\oplus \bar{a}_2)&=0,\\
	\bar{a}_1(x_1\oplus a_2)&=0,\\
	\bar{a}_1(x_1\oplus \bar{a}_2)&=0.
\end{align}

Similarly, party $S_2$ (with local access to $a_2$) can access information about $a_1$ in four distinct ways depending on whether $a_2$ takes a particular value, \textit{i.e.},
\begin{align}
	a_2(x_2\oplus a_1)&=0,\\
	a_2(x_2\oplus \bar{a}_1)&=0,\\
	\bar{a}_2(x_2\oplus a_1)&=0,\\
	\bar{a}_2(x_2\oplus \bar{a}_1)&=0.
\end{align}
Hence, there are $4\times 4=16$ ways in which the two parties can be asked (or not) to guess their neighbour's input depending on their own input, each corresponding to a different version of the LGYNI game. These can again be parameterized, as in Ref.~\cite{BAF15}, via four-bit strings $(\alpha'_0,\alpha'_1,\beta'_0,\beta'_1)\in\{0,1\}^4$ so that each such string corresponds to a version of LGYNI given by
\begin{align}\label{paramlgyni}
	(a_1\oplus\alpha'_1)(x_1\oplus \alpha'_0 \oplus a_2)&=0,\nonumber\\
	(a_2\oplus\beta'_1)(x_2\oplus \beta'_0 \oplus a_1)&=0.
\end{align}
Here $(\alpha'_0,\alpha'_1,\beta'_0,\beta'_1)=(0,0,0,0)$ is the canonical version in Eq.~\eqref{canlgyni}.

\subsection{Maximal violation of GYNI and LGYNI inequalities via noncausal vertices}
The noncausal vertices (with signalling graph $S_1\leftrightarrow S_2$) of the bipartite correlation polytope fall into two overlapping sets:
\begin{itemize}
	\item those violating a LGYNI inequality maximally, this set denoted by $V_{\rm LGYNI}$; \textit{e.g.}, the noncausal vertex
	\begin{align}
		\begin{pmatrix}\label{lgyni_0}
			0&1&1&0\\
			1&0&0&0\\
			0&0&0&0\\
			0&0&0&1
		\end{pmatrix}
	\end{align}
	which violates the canonical LGYNI inequality (Eq.~\eqref{canlgyni}) maximally.
	
	\item those violating a GYNI inequality maximally, this set denoted by $V_{\rm GYNI}$; \textit{e.g.}, the noncausal vertex
	\begin{align}\label{gyni_0}
		\begin{pmatrix}
			1&0&0&0\\
			0&0&1&0\\
			0&1&0&0\\
			0&0&0&1
		\end{pmatrix}
	\end{align}
	violates the canonical GYNI inequality (Eq.~\eqref{cangyni}) maximally.
\end{itemize}

Some of the vertices that maximally violate a GYNI inequality also maximally violate a corresponding set of LGYNI inequalities. We characterize this correspondence below.

\begin{theorem}\label{vgynitolgyni}
	A vertex maximally violating a GYNI inequality also violates some LGYNI inequality maximally if and only if the GYNI inequality is of type $(\alpha_0,0,\beta_0,0)$ and the LGYNI inequality is of type $(\alpha_0,\alpha'_1,\beta_0,\beta'_1)$, where $(\alpha'_1,\beta'_1)\in\{(0,0),(0,1),(1,0),(1,1)\}$.
\end{theorem}

\begin{proof}
	Consider a vertex $P_{\vec{X}|\vec{A}}\in V_{\rm GYNI}$ violating a GYNI inequality parameterized by  $(\alpha_0,\alpha_1,\beta_0,\beta_1)$ according to Eq.~\eqref{paramgyni}. We then have that $p(\vec{x}|\vec{a})=1$ for 
	\begin{align}
		x_1&=a_2\oplus\alpha_1a_1\oplus\alpha_0,\\
		x_2&=a_1\oplus\beta_1a_2\oplus\beta_0.
	\end{align}	
	Substituting these conditions in Eq.~\eqref{paramlgyni}, we obtain that for $P_{\vec{X}|\vec{A}}\in V_{\rm GYNI}$ to maximally violate an LGYNI inequality parameterized by $(\alpha'_0,\alpha'_1,\beta'_0,\beta'_1)$, at least one of the following four sets of equations must hold
	\begin{align}
		&\alpha'_1=a_1, \beta'_1=a_2,\\
		&\alpha'_1=a_1,\beta'_0=\beta_0\oplus \beta_1a_2,\\
		&\alpha'_0=\alpha_0\oplus \alpha_1a_1, \beta'_1=a_2,\\
		&\alpha'_0=\alpha_0\oplus \alpha_1a_1, \beta'_0=\beta_0\oplus \beta_1a_2.
	\end{align}
	Since the first three of these conditions require, respectively, $(\alpha'_1,\beta'_1)=(a_1,a_2)$, $\alpha'_1=a_1$, and $\beta'_1=a_2$, none of them can be satisfied for \textit{all} $(a_1,a_2)$---which is required for a maximal violation from a set of equations---since LGYNI inequality fixes $(\alpha'_1,\beta'_1)$ to take exactly one value in $\{(0,0),(0,1),(1,0),(1,1)\}$. Hence, the only feasible way to violate the LGYNI inequality maximally is to have $(\alpha'_1,\beta'_1)$ satisfy the fourth set of equations (for all $(a_1,a_2)$), namely,
	\begin{align}\label{maxlgyni}
		&\alpha'_0=\alpha_0\oplus \alpha_1a_1, \beta'_0=\beta_0\oplus \beta_1a_2,
	\end{align}
	and this is only possible if the parameters of the inequality are independent of the setting choices $(a_1,a_2)$. This independence is \textit{only} achieved when $\alpha_1=\beta_1=0$. Hence, the GYNI inequality should, in fact, be of the type $(\alpha_0,0,\beta_0,0)$ for a vertex maximally violating it to also maximally violate an LGYNI inequality which must, in turn, be of the type $(\alpha_0,\alpha'_1,\beta_0,\beta'_1)$, where $(\alpha'_1,\beta'_1)\in\{(0,0),(0,1),(1,0),(1,1)\}$. 
\end{proof}
Some examples:
\begin{enumerate}
	\item The vertex in Eq.~\eqref{gyni_0} that maximally violates the canonical GYNI inequality (Eq.~\eqref{gyni}) also maximally violates the canonical LGYNI inequality (Eq.~\eqref{lgyni}). 
	
	\item The vertex in Eq.~\eqref{lgyni_0} that maximally violates the canonical LGYNI inequality does not violate the canonical GYNI inequality.
	
	\item The vertex
	\begin{equation}
		x_1=a_1\oplus a_2, x_2=a_1\oplus a_2
	\end{equation}
	maximally violates the GYNI inequality of type $(0,1,0,1)$ but it does not maximally violate any LGYNI inequality. This can be verified by substituting $x_1$ and $x_2$ for this vertex in Eq.~\eqref{lgyni}, which results in constraints on $(\alpha'_0,\alpha'_1,\beta'_0,\beta'_1)$ that are not independent of $(a_1,a_2)$.		
\end{enumerate}

We now characterize the sets of vertices $V_{\rm GYNI}$ and $V_{\rm LGYNI}$ that (respectively) maximally violate the different versions of the GYNI and LGYNI inequalities. Generic GYNI and LGYNI inequalities, following the winning conditions of Eqs.~\eqref{paramgyni} and \eqref{paramlgyni}, are given by 

\begin{align}
	&\textrm{Generic GYNI:}\nonumber\\
	&\frac{1}{4}\sum_{\vec{a},\vec{x}}\delta_{x_1,a_2\oplus\alpha_1a_1\oplus\alpha_0}\delta_{x_2,a_1\oplus\beta_1a_2\oplus\beta_0}p(\vec{x}|\vec{a})\leq\frac{1}{2}\label{gengyni}\\
	&\textrm{Generic LGYNI:}\nonumber\\
	&\frac{1}{4}\sum_{\vec{a},\vec{x}}\delta_{(a_1\oplus\alpha'_1)(x_1\oplus \alpha'_0\oplus a_2),0}\delta_{(a_2\oplus \beta'_1)(x_2\oplus\beta'_0\oplus a_1),0}p(\vec{x}|\vec{a})\leq\frac{3}{4}.\label{genlgyni}
\end{align}
\textbf{Characterizing $V_{\rm GYNI}$:} A generic GYNI inequality of type $(\alpha_0,\alpha_1,\beta_0,\beta_1)$ (Eq.~\eqref{gengyni}) is maximally violated by the unique vertex corresponding to Eq.~\eqref{paramgyni}, \textit{i.e.}, the vertex given by
\begin{align}
	&x_1=a_2\oplus \alpha_1 a_1\oplus \alpha_0,\nonumber\\
	&x_2=a_1\oplus \beta_1 a_2\oplus\beta_0. 
\end{align}
Hence, there are exactly $16$ noncausal vertices in $V_{\rm GYNI}$, one corresponding to each GYNI inequality. Of these $16$ vertices, each of the following $4$ violates $4$ LGYNI inequalities maximally (from Theorem \ref{vgynitolgyni}): 
\begin{align}
	&x_1=a_2,x_2=a_1\\
	&x_1=a_2,x_2=a_1\oplus 1\\
	&x_1=a_2\oplus1,x_2=a_1\\
	&x_1=a_2\oplus1,x_2=a_1\oplus1
\end{align}
\textbf{Characterizing $V_{\rm LGYNI}$:} We now characterize the vertices that maximally violate the LGYNI inequalities maximally. Note that there are a total of $144$ noncausal vertices and $16$ of them violate the GYNI inequality maximally. Hence, there exist $128$ other vertices that violate LGYNI inequalities maximally but not the GYNI inequalities.
\begin{theorem}
	A generic LGYNI inequality of type $(\alpha'_0,\alpha'_1,\beta'_0,\beta'_1)$ (Eq.~\eqref{genlgyni}) is maximally violated by exactly $16$ vertices. Exactly one of these $16$ vertices is also a GYNI vertex, \textit{i.e.}, it violates the GYNI inequality of type $(\alpha'_0,0,\beta'_0,0)$ maximally. 
\end{theorem}
\begin{proof}
	A generic LGYNI inequality of type $(\alpha'_0,\alpha'_1,\beta'_0,\beta'_1)$ (Eq.~\eqref{genlgyni}) is maximally violated by exactly $16$ vertices. Recall that a vertex maximally violating this inequality must satisfy Eq.~\eqref{genlgyni}, \textit{i.e.},
	\begin{align}
		(a_1\oplus\alpha'_1)(x_1\oplus \alpha'_0 \oplus a_2)&=0,\nonumber\\
		(a_2\oplus\beta'_1)(x_2\oplus \beta'_0 \oplus a_1)&=0.
	\end{align}
	We will now argue that for any correlation with deterministic outcomes given by 
	\begin{align}
		&(x_1,x_2|a_1,a_2)\in\{(x_1^{(00)},x_2^{(00)}|0,0), (x_1^{(01)},x_2^{(01)}|0,1),\nonumber\\
		&(x_1^{(10)},x_2^{(10)}|1,0), (x_1^{(11)},x_2^{(11)}|1,1)\},
	\end{align}
	and which maximally violates the LGYNI inequality, only $4$ of the $8$ variables in  $\{x_1^{(a_1a_2)},x_2^{(a_1a_2)}\}_{a_1a_2}$ are independent of each other. The condition of maximal violation requires that
	\begin{align}
		&\alpha'_1(x_1^{(00)}\oplus\alpha'_0)=0,\quad \beta'_1(x_2^{(00)}\oplus\beta'_0)=0,\label{eq:maxlgyni_1}\\
		&\alpha'_1(x_1^{(01)}\oplus\bar{\alpha}'_0)=0,\quad \bar{\beta}'_1(x_2^{(01)}\oplus\beta'_0)=0,\label{eq:maxlgyni_2}\\
		&\bar{\alpha}'_1(x_1^{(10)}\oplus\alpha'_0)=0,\quad \beta'_1(x_2^{(10)}\oplus\bar{\beta}'_0)=0,\label{eq:maxlgyni_3}\\
		&\bar{\alpha}'_1(x_1^{(11)}\oplus\bar{\alpha}'_0)=0,\quad \bar{\beta}'_1(x_2^{(11)}\oplus\bar{\beta}'_0)=0.\label{eq:maxlgyni_4}
	\end{align}
	For any value of $(\alpha'_1,\beta'_1)$, exactly one of the $4$ pairs of equations fixes $2$ of the $8$ variables, namely, 
	\begin{align}
		&x_1^{(\bar{\alpha}'_1\bar{\beta}'_1)}=\alpha'_0\oplus\bar{\beta}'_1,\label{eq:fixed_1}\\ &x_2^{(\bar{\alpha}'_1\bar{\beta}'_1)}=\beta'_0\oplus\bar{\beta}'_1,\label{eq:fixed_2}
	\end{align}
	\textit{e.g.}, for $(\alpha'_1,\beta'_1)=(1,1)$, Eqs.~\eqref{eq:maxlgyni_1} fix $x_1^{(00)}=\alpha'_0, x_2^{(00)}=\beta'_0$. That leaves us with at most $6$ independent variables. 
	
	Since $\alpha'_1$ appears in another pair of equations (where $\beta'_1$ is flipped in the other equation in the pair), it also fixes one of the remaining $6$ variables, namely, 
	\begin{align}
		x_1^{(\bar{\alpha}'_1\beta'_1)}=\alpha'_0\oplus\beta'_1,\label{eq:fixed_3}	
	\end{align}
	\textit{e.g.},  for $(\alpha'_1,\beta'_1)=(1,1)$, Eq.~\eqref{eq:maxlgyni_2} fixes $x_1^{(01)}=\bar{\alpha}'_0$. Similarly, since $\beta'_1$ also appears in another pair of equations (where $\alpha'_1$ is flipped in the other equation in the pair), it fixes another variable, namely, 
	\begin{align}
		x_2^{(\alpha'_1\bar{\beta}'_1)}=\beta'_0\oplus\alpha'_1,\label{eq:fixed_4}
	\end{align}
	\textit{e.g.}, $x_2^{(10)}=\bar{\beta}'_0$ for $(\alpha'_1,\beta'_1)=(1,1)$. There are no further linear dependencies for any given $(\alpha'_1,\beta'_1)$, hence we have that $4$ of the $8$ variables are independent, namely,
	\begin{align}\label{eq:freetuple}
		x_2^{(\bar{\alpha}'_1\beta'_1)}, x_1^{(\alpha'_1\bar{\beta}'_1)}, x_1^{(\alpha'_1\beta'_1)}, x_2^{(\alpha'_1\beta'_1)}\in\{0,1\},
	\end{align}
	\textit{e.g.}, for $(\alpha'_1,\beta'_1)=(1,1)$, these are the variables $x_2^{(01)}, x_1^{(00)}, x_1^{(11)}, x_2^{(11)}\in\{0,1\}$. Hence, we have $2^4=16$ vertices maximally violating the given LGYNI inequality, each corresponding to a particular choice values in Eq.~\eqref{eq:freetuple}, with the remaining values fixed by Eqs.~\eqref{eq:fixed_1}, \eqref{eq:fixed_2}, \eqref{eq:fixed_3}, and \eqref{eq:fixed_4}.
\end{proof}

\textbf{Example:} There are $16$ noncausal vertices violating the canonical LGYNI inequality (Eq.~\eqref{canlgyni}) maximally, given by assignments of the type
\begin{align}
	&(x_1,x_2|a_1,a_2)\nonumber\\
	\in&\{(x_1^{(00)},x_2^{(00)}|0,0), (x_1^{(01)},0|0,1), (0,x_2^{(10)}|1,0),\nonumber\\
	&(1,1|1,1)\},
\end{align}
where $x_1^{(00)},x_2^{(00)},x_1^{(01)},x_2^{(10)}\in\{0,1\}$. 

The assignment $(x_1^{(00)},x_2^{(00)}|0,0)$ corresponds to the case where no GYNI condition is required of either party (so, four possibilities for this assignment). The assignment $(x_1^{(01)},0|0,1)$ requires the GYNI condition of the second party (so, two possibilities) and the assignment $(0,x_2^{(10)}|1,0)$ requires the GYNI condition of the first party (so, two possibilities). Hence, in all, there are $16$ possible assignments for all four settings $(a_1,a_2)\in\{(0,0),(0,1),(1,0),(1,1)\}$, each constituting a LGYNI vertex.  Of these $16$ vertices, the specific assignment corresponding to $x_1^{(00)}=0,x_2^{(00)}=0,x_1^{(01)}=1,x_2^{(10)}=1$ is the vertex violating the canonical GYNI inequality (Eq.~\eqref{cangyni}) maximally. 

Hence, the unique noncausal vertex $\mathbb{P}_{\vec{X},\vec{A}}$ violating the canonical GYNI inequality maximally (and therefore also the canonical LGYNI inequality maximally) is
\begin{align}
	P_{\vec{X}|\vec{A}}=\begin{pmatrix}
		1&0&0&0\\
		0&0&1&0\\
		0&1&0&0\\
		0&0&0&1
	\end{pmatrix}.
\end{align}

\section{Classification of all tripartite vertices in the $(3,2,2)$ scenario}\label{app:322vertices}
In the tripartite scenario where each party has a binary choice of inputs $a_i\in\{0,1\}$, each resulting in a binary outcome $x_i\in\{0,1\}$ ($i=1,2,3$), the object of interest is the conditional probability distribution $p(x_1,x_2,x_3|a_1,a_2,a_3)$. We will represent this correlation as an $8\times 8$ column-stochastic matrix $P_{\vec{X}|\vec{A}}$ with entries $P_{\vec{X}|\vec{A}}(\vec{x}|\vec{a}):=p(x_1,x_2,x_3|a_1,a_2,a_3)$, where $\vec{x},\vec{a}\in\{000,001,010,011,100,101,110,111\}$. The vertices correspond to $\{0,1\}$-valued column-stochastic matrices, \textit{i.e.}, there are $8^8=16,777,216$ vertices. We can now classify them according to their signalling classes and determine whether they are causal or noncausal, following Eq.~\eqref{eq:causalcorrdet}.

\subsection{Bell, \textit{i.e.}, non-signalling, vertices}

Here no party signals to any other, \textit{i.e.}, 
\begin{align}
	x_1=f_1(a_1), x_2=f_2(a_2),
	x_3=f_3(a_3).
\end{align}
That is, we have $4^3=64$ Bell vertices, since there are four choices of $f_i\in\{\textrm{Id}, \textrm{Flip}, 0,1\}$ for each $i\in\{1,2,3\}$. The signalling class is a $3$-node DAG with no edges, \textit{i.e.}, $\{\{\circ\},\{\star\},\{\bullet\}\}$.

There is only one signalling graph in this class, \textit{i.e.}, relabelling parties doesn't change the signalling relations between the parties.

\subsection{Bipartite causal signalling vertices}
Here, 
one party, $S_i$, is a by-stander, not communicating with the other two in any way and for each choice of $i\in\{1,2,3\}$ and $f_i\in\{\textrm{Id}, \textrm{Flip},0,1\}$, we have two families of signalling graphs:
\begin{itemize}
	\item $S_{i+1}\rightarrow S_{i+2}$, \textit{i.e.},
	\begin{align}
		x_i&=f_i(a_i),\nonumber\\
		x_{i+1}&=f_{i+1}(a_{i+1}),\nonumber\\ x_{i+2}&=f_{i+2}(a_{i+1},a_{i+2}),
	\end{align} 
	There are $48$ of vertices with this signalling graph, namely, one-way signalling vertices inherited from the bipartite case.
	
	\item $S_{i+2}\rightarrow S_{i+1}$, \textit{i.e.}, 
	\begin{align}
		x_i&=f_i(a_i),\nonumber\\
		x_{i+1}&=f_{i+1}(a_{i+1},a_{i+2}),\nonumber\\
		x_{i+2}&=f_{i+2}(a_{i+2}).
	\end{align}
	Again, there are $48$ of these vertices inherited from the bipartite case.
\end{itemize}

The signalling class of these vertices is a $3$-node disconnected DAG given by  $\{\{\star\}, \{\circ\rightarrow \bullet\}\}$. There are $6$ signalling graphs in this class under relabellings: there are $3$ choices of by-stander node, and for each such choice, $2$ choices of signalling direction for the remaining two nodes.

Each signalling graph contributes $4\times48=192$ bipartite causal signalling vertices since the by-stander party, $S_i$, has $4$ choices of functions $f_i$. Hence, the total number of bipartite causal signalling vertices is $6\times 4\times 48=1152$.
\subsection{Bipartite noncausal signalling vertices}
Here again 
one party, $S_i$, is a by-stander, not communicating with the other two in any way and for each choice of $i\in\{1,2,3\}$ and $f_i\in\{\textrm{Id}, \textrm{Flip},0,1\}$, we have $S_{i+1}\leftrightarrow S_{i+2}$, so that
\begin{align}
	x_i&=f_i(a_i),\\
	x_{i+1}&=f_{i+1}(a_{i+1},a_{i+2}),\\
	x_{i+2}&=f_{i+2}(a_{i+1},a_{i+2})
\end{align}
There are $144$ choices of functions $f_{i+1}$ and $f_{i+2}$ inherited from the bipartite case. Hence, each signalling graph contributes $4\times 144=576$ vertices.

The signalling class is a $3$-node directed graph (DG) given by $\{\{\star\}, \{\circ\leftrightarrow \bullet\}\}$. There are $3$ signalling graphs in this class under all possible labellings: one each for a choice of by-stander node. Hence, total number of bipartite noncausal signalling vertices is $3\times4\times 144=1728$.

\subsection{Tripartite fixed-order causal signalling vertices}
In the case of tripartite signalling vertices, we have that 
no party is a by-stander: all of them communicate by sending or receiving signals. Further, for tripartite \textit{fixed-order causal} signalling vertices, the signalling graph is a connected DAG. There are $4$ tripartite fixed-order causal signalling classes associated with the following DAGs:
\begin{enumerate}
	\item $\{\{\circ\rightarrow\star\}, \{\star\rightarrow\bullet\}\}$ 
	\item $\{\{\circ\leftarrow\star\}, \{\star\rightarrow\bullet\}\}$
	\item $\{\{\circ\rightarrow\star\}, \{\star\leftarrow\bullet\}\}$
	\item $\{\{\circ\leftarrow\star\}, \{\star\rightarrow\bullet\}, \{\circ\rightarrow\bullet\} \}$
\end{enumerate}

We can now characterize the vertices in each class.
\begin{enumerate}
	\item $\{\{\circ\rightarrow\star\}, \{\star\rightarrow\bullet\}\}$: There are $6$ signalling graphs in this class under all possible labellings (each permutation of parties counts).
	
	For each ordered set $(i_1,i_2,i_3)$, the vertices are given by
	\begin{align}
		x_{i_1}&=f_{i_1}(a_{i_1}),\\
		x_{i_2}&=f_{i_2}(a_{i_1},a_{i_2}),\\
		x_{i_3}&=f_{i_3}(a_{i_2},a_{i_3}),
	\end{align}
	with $2^2=4$ choices of $f_{i_1}$, $2^4-2^2=12$ choices of $f_{i_2}$ (excluding the cases where $x_{i_2}$ is independent of $a_{i_1}$), and $2^4-2^2=12$ choices of $f_{i_3}$ (excluding the cases where $x_{i_3}$ is independent of $a_{i_2}$).
	
	Hence, we have $6\times 4\times 12\times 12=3456$ vertices of this type.
	
	\item $\{\{\circ\leftarrow\star\}, \{\star\rightarrow\bullet\}\}$: There are $3$ signalling graphs in this class, one each for a choice of common cause node `$\star$'.
	
	For each choice of common cause node $i_2\in\{1,2,3\}$, the vertices are given by
	\begin{align}
		x_{i_1}&=f_{i_1}(a_{i_1},a_{i_2}),\\
		x_{i_2}&=f_{i_2}(a_{i_2}),\\
		x_{i_3}&=f_{i_3}(a_{i_2},a_{i_3}),
	\end{align}
	with $2^2=4$ choices of $f_{i_2}$, $2^4-2^2=12$ choices of $f_{i_1}$, and $2^4-2^2=12$ choices of $f_{i_3}$.
	
	Hence, we have $3\times 4\times 12\times 12=1728$ vertices of this type.
	
	\item $\{\{\circ\rightarrow\star\}, \{\star\leftarrow\bullet\}\}$: There are $3$ signalling graphs in this class, one each for a choice of common effect node `$\star$'.
	
	For each choice of common effect node $i_2\in\{1,2,3\}$, the vertices are given by
	\begin{align}
		x_{i_1}&=f_{i_1}(a_{i_1}),\\
		x_{i_2}&=f_{i_2}(a_{i_1},a_{i_2},a_{i_3}),\\
		x_{i_3}&=f_{i_3}(a_{i_3}),
	\end{align}
	with $2^2=4$ choices of $f_{i_1}$, $2^8-2^2-2\times 12=256-28=228$ choices of $f_{i_2}$,\footnote{Excluding the cases where $f_{i_2}$ is a function only of $a_{i_2}$ or only of $\{a_{i_1},a_{i_2}\}$ or only of $\{a_{i_2},a_{i_3}\}$} and $2^2=4$ choices of $f_{i_3}$.
	
	Hence, we have $3\times 4\times 228\times 4=10944$ vertices of this type.
	
	\item $\{\{\circ\leftarrow\star\}, \{\star\rightarrow\bullet\}, \{\circ\rightarrow\bullet\} \}$: There are $6$ signalling graphs in this class, one for each choice of common cause node (`$\star$') and the direction of signalling between the remaining nodes.
	
	For each choice of common cause node $i_1\in\{1,2,3\}$ and direction of signalling $S_{i_2}\rightarrow S_{i_3}$ between the remaining nodes, the vertices are given by
	\begin{align}
		x_{i_1}&=f_{i_1}(a_{i_1}),\\
		x_{i_2}&=f_{i_2}(a_{i_1},a_{i_2}),\\
		x_{i_3}&=f_{i_3}(a_{i_1},a_{i_2},a_{i_3}),
	\end{align}
	with $2^2=4$ choices of $f_{i_1}$, $2^4-2^2=12$ choices of $f_{i_2}$, and $2^8-2^2-2\times 12=256-28=228$ choices of $f_{i_3}$.\footnote{Excluding the cases where $f_{i_3}$ is a function only of $a_{i_3}$ or only of $\{a_{i_1},a_{i_3}\}$ or only of $\{a_{i_2},a_{i_3}\}$.}
	
	Hence, we have $6\times 4\times 12\times 228=65664$ vertices of this type.
\end{enumerate}
\subsection{Tripartite (dynamical-order) causal or Tripartite noncausal signalling vertices}
These vertices correspond to the following connected and cyclic directed graph representing their signalling class:
$\{\{\circ\leftarrow\star\}, \{\star\rightarrow\bullet\}, \{\circ\leftrightarrow\bullet\}\}$. 

We now characterize the vertices in this class. There are $3$ signalling graphs in this class, one for each choice of $2$-cycle, `$\{\circ\leftrightarrow\bullet\}$'. For each choice of $2$-cycle $(i_1, i_2)\in\{(1,2),(1,3), (2,3)\}$, the vertices are given by
\begin{align}
	x_{i_1}&=f_{i_1}(a_{i_1},a_{i_2},a_{i_3}),\\
	x_{i_2}&=f_{i_2}(a_{i_1},a_{i_2},a_{i_3}),\\
	x_{i_3}&=f_{i_3}(a_{i_3}),
\end{align}
with $2^8-2^2-2\times 12=256-4-24=228$ choices of $f_{i_1}$, $228$ choices of $f_{i_2}$, and $2^2=4$ choices of $f_{i_3}$.

Hence, we have $3\times228\times228\times4=623808$ vertices in this signalling class. Note that these vertices are neither all causal nor all noncausal: in particular, consider the following vertex: 
\begin{align}
	x_{i_1}&=a_{i_2}.a_{i_3},\nonumber\\
	x_{i_2}&=a_{i_1}.\bar{a}_{i_3},\nonumber\\
	x_{i_3}&=0.
\end{align}
This describes a classical switch and is, therefore, a (dynamical-order) causal signalling vertex: for $a_{i_3}=0$, we have $S_{i_1}\rightarrow S_{i_2}$ and for $a_{i_3}=1$, we have $S_{i_2}\rightarrow S_{i_1}$. This can also be verified by checking that this vertex satisfies the definition of a deterministic causal correlation in Eq.~\eqref{eq:causalcorrdet}.

On the other hand, consider another vertex:

\begin{align}
	x_{i_1}&=a_{i_2}\oplus a_{i_3},\nonumber\\
	x_{i_2}&=a_{i_1}\oplus a_{i_3},\nonumber\\
	x_{i_3}&=0.
\end{align}
This is a noncausal vertex since it fails the characterization of Eq.~\eqref{eq:causalcorrdet}: the proper subset $\{S_{i_1},S_{i_2}\}$ contains no party whose outcome is independent of the other party's setting in the subset, given \textit{any} setting of the party $S_{i_3}$ outside this subset. As noted previously, noncausal vertices in this signalling class (including the particular one here) are not achievable by a process matrix on account of Corollary \ref{cor:causalityofvertices}.
\subsection{Tripartite noncausal signalling vertices}
Tripartite noncausal signalling classes correspond to connected and cyclic DGs and there are $8$ of them:
\begin{enumerate}
	
	\item $\{\{\circ\rightarrow\star\}, \{\star\rightarrow\bullet\}, \{\circ\leftrightarrow\bullet\}\}$
	
	\item $\{\{\circ\rightarrow\star\}, \{\star\leftarrow\bullet\}, \{\circ\leftrightarrow\bullet\}\}$
	
	\item $\{\{\circ\rightarrow\star\}, \{\star\leftrightarrow\bullet\}\}$
	
	\item $\{\{\circ\leftarrow\star\}, \{\star\leftrightarrow\bullet\}\}$
	
	\item $\{\{\circ\leftrightarrow\star\}, \{\star\leftrightarrow\bullet\}\}$
	
	\item $\{\{\circ\leftrightarrow\star\}, \{\star\leftrightarrow\bullet\},\{\circ\rightarrow\bullet\}\}$
	
	\item $\{\{\circ\leftrightarrow\star\}, \{\star\leftrightarrow\bullet\},\{\circ\leftrightarrow\bullet\}\}$
	
	\item $\{\{\circ\rightarrow\star\}, \{\star\rightarrow\bullet\}, \{\bullet\rightarrow\circ\}\}$
\end{enumerate}
We can now characterize the vertices in each class.

\begin{enumerate}
	
	\item $\{\{\circ\rightarrow\star\}, \{\star\rightarrow\bullet\}, \{\circ\leftrightarrow\bullet\}\}$: There are $6$ signalling graphs in this class, one for each choice of $2$-cycle ($3$ such choices) and for each choice of which node in the cycle receives a signal from the external node ($2$ such choices; the node sending a signal to the external node is fixed by this choice).
	
	For each choice of $2$-cycle $(i_1, i_2)\in\{(1,2),(1,3), (2,3)\}$ and reception node $i_1$ (\textit{i.e.}, $S_{i_3}\rightarrow S_{i_1}$ and $S_{i_2}\rightarrow S_{i_3}$), the vertices are given by
	\begin{align}
		x_{i_1}&=f_{i_1}(a_{i_1},a_{i_2},a_{i_3}),\\
		x_{i_2}&=f_{i_2}(a_{i_1},a_{i_2}),\\
		x_{i_3}&=f_{i_3}(a_{i_2},a_{i_3}),
	\end{align}
	with $2^8-2^2-2\times 12=256-4-24=228$ choices of $f_{i_1}$, $2^4-2^2=12$ choices of $f_{i_2}$, and $2^4-2^2=12$ choices of $f_{i_3}$.
	
	Hence, we have $6\times228\times 12\times 12=196992$ vertices in this signalling class. All these vertices are noncausal: this follows from noting that the directed graph contains no root node, \textit{i.e.}, there does not exist any party whose outcome is independent of the settings of all the other parties; hence, all signalling graphs in this class fail to satisfy this necessary condition for a deterministic correlation to be causal (\textit{cf.}~Eq.~\eqref{eq:causalcorrdet}).
	
	From Corollary \ref{cor:siboncyc}, none of the vertices in this class is realizable by a process matrix.
	
	\item $\{\{\circ\rightarrow\star\}, \{\star\leftarrow\bullet\}, \{\circ\leftrightarrow\bullet\}\}$: There are $3$ signalling graphs in this class, one for each choice of $2$-cycle.
	
	For each choice of $2$-cycle $(i_1, i_2)\in\{(1,2),(1,3), (2,3)\}$, the vertices are given by
	\begin{align}
		x_{i_1}&=f_{i_1}(a_{i_1},a_{i_2}),\\
		x_{i_2}&=f_{i_2}(a_{i_1},a_{i_2}),\\
		x_{i_3}&=f_{i_3}(a_{i_1},a_{i_2},a_{i_3}),
	\end{align}
	with $2^4-2^2=12$ choices of $f_{i_1}$, $12$ choices of $f_{i_2}$, and $228$ choices of $f_{i_3}$.
	
	Hence, we have $3\times12^2\times 228=98496$ vertices of this type. All these vertices are noncausal since the signalling class contains no root nodes.
	
	From Corollary \ref{cor:siboncyc}, none of the vertices in this signalling class is realizable by a process-matrix.
	
	\item $\{\{\circ\rightarrow\star\}, \{\star\leftrightarrow\bullet\}\}$: There are $6$ signalling graphs in this class, one for each choice of $2$-cycle and of the node in the $2$-cycle that receives signal from the external node.
	
	For each choice of $2$-cycle $(i_1, i_2)\in\{(1,2),(1,3), (2,3)\}$ and of the reception node (say $i_1$), the vertices are given by
	\begin{align}
		x_{i_1}&=f_{i_1}(a_{i_1},a_{i_2},a_{i_3}),\\
		x_{i_2}&=f_{i_2}(a_{i_1},a_{i_2}),\\
		x_{i_3}&=f_{i_3}(a_{i_3}),
	\end{align}
	with $228$ choices of $f_{i_1}$, $12$ choices of $f_{i_2}$, and $4$ choices of $f_{i_3}$.
	
	Hence, we have $6\times228\times12\times4=65664$ vertices in this signalling class. This signalling class contains a root node, namely party $S_{i_3}$, so it satisfies this necessary condition for its vertices to be causal. However, the non-trivial dependence of $x_{i_1}$ on $a_{i_2}$ and $a_{i_3}$ entails that there exists \textit{at least one} value of $a_{i_3}$ such that $x_{i_1}$ depends non-trivially on $a_{i_2}$ for all the vertices in this signalling class. On the other hand, for all values of $a_{i_3}$, $x_{i_2}$ depends non-trivially on $a_{i_1}$ for all the vertices in this signalling class. Hence, we have that all vertices in this signalling class are noncausal since the proper subset of parties $\{S_{i_1},S_{i_2}\}$ does not admit any party whose outcome is independent of the setting of the other party for all settings of party $S_{i_3}$: in particular, there is at least one setting of $S_{i_3}$ for which this independence cannot hold.
	
	From Corollary \ref{cor:siboncyc}, none of the vertices in this signalling class is achievable by a process matrix.
	
	\item $\{\{\circ\leftarrow\star\}, \{\star\leftrightarrow\bullet\}\}$: There are $6$ signalling graphs in this class, one for each choice of $2$-cycle and of the node in the $2$-cycle that signals to the external node.
	
	For each choice of $2$-cycle $(i_1, i_2)\in\{(1,2),(1,3), (2,3)\}$ and of the sender node (say $i_1$), the vertices are given by
	\begin{align}
		x_{i_1}&=f_{i_1}(a_{i_1},a_{i_2}),\\
		x_{i_2}&=f_{i_2}(a_{i_1},a_{i_2}),\\
		x_{i_3}&=f_{i_3}(a_{i_1},a_{i_3}),
	\end{align}
	with $12$ choices of $f_{i_1}$, $12$ choices of $f_{i_2}$, and $12$ choices of $f_{i_3}$.
	
	Hence, we have $6\times12^3=10368$ vertices in this signalling class. All these vertices are noncausal since the directed graph admits no root node.
	
	From Corollary \ref{cor:siboncyc}, none of the vertices in this signalling class is achievable by a process matrix.
	
	\item $\{\{\circ\leftrightarrow\star\}, \{\star\leftrightarrow\bullet\}\}$: There are $3$ signalling graphs in this class, one for each choice of common node (that sends and receives signals from the remaining two nodes).
	
	For each choice of common node $i_1\in\{1,2,3\}$, the vertices are given by
	\begin{align}
		x_{i_1}&=f_{i_1}(a_{i_1},a_{i_2},a_{i_3}),\\
		x_{i_2}&=f_{i_2}(a_{i_1},a_{i_2}),\\
		x_{i_3}&=f_{i_3}(a_{i_1},a_{i_3}),
	\end{align}
	with $228$ choices of $f_{i_1}$, $12$ choices of $f_{i_2}$, and $12$ choices of $f_{i_3}$.
	
	Hence, we have $3\times228\times 12^2=98496$ vertices in this signalling class. All these vertices are noncausal since the directed graph admits no root node.
	
	From Corollary \ref{cor:siboncyc}, none of the vertices in this signalling class is achievable by a process matrix.
	
	\item $\{\{\circ\leftrightarrow\star\}, \{\star\leftrightarrow\bullet\},\{\circ\rightarrow\bullet\}\}$: There are $6$ signalling graphs in this class, one for each choice of common node (that sends and receives signals from the remaining two nodes) and the choice of signalling direction between the remaining nodes.
	
	For each choice of common node $i_1\in\{1,2,3\}$ and signalling direction between the remaining nodes (say $S_{i_2}\rightarrow S_{i_3}$), the vertices are given by
	\begin{align}
		x_{i_1}&=f_{i_1}(a_{i_1},a_{i_2},a_{i_3}),\\
		x_{i_2}&=f_{i_2}(a_{i_1},a_{i_2}),\\
		x_{i_3}&=f_{i_3}(a_{i_1},a_{i_2},a_{i_3}),
	\end{align}
	with $228$ choices of $f_{i_1}$, $12$ choices of $f_{i_2}$, and $228$ choices of $f_{i_3}$.
	
	Hence, we have $6\times228^2\times 12=3,742,848$ vertices of this type. All these vertices are noncausal since the directed graph admits no root node.
	
	From Corollary \ref{cor:siboncyc}, none of the vertices in this signalling class is achievable by a process matrix.
	
	\item $\{\{\circ\leftrightarrow\star\}, \{\star\leftrightarrow\bullet\},\{\circ\leftrightarrow\bullet\}\}$: There is $1$ signalling graph in this class (all labellings are equivalent).
	
	The vertices are given by
	\begin{align}
		x_{i_1}&=f_{i_1}(a_{i_1},a_{i_2},a_{i_3}),\\
		x_{i_2}&=f_{i_2}(a_{i_1},a_{i_2},a_{i_3}),\\
		x_{i_3}&=f_{i_3}(a_{i_1},a_{i_2},a_{i_3}),
	\end{align}
	with $228$ choices for each $f_{i_k}$, $k\in\{1,2,3\}$.
	
	Hence, we have $228^3=11,852,352$ vertices in this signalling class. All these vertices are noncausal since the directed graph admits no root node.
	
	This is the only tripartite noncausal signalling class that admits vertices realizable by the process-matrix framework: this follows from our Theorem \ref{thm:det} and the result of Ref.~\cite{TB22} that the directed graph associated with this signalling class is the unique one that is an admissible causal structure in the process-matrix framework and also realizes a noncausal vertex.
	
	In particular, the AF/BW correlation \cite{BW16} falls in this class and is given by
	\begin{align}
		&f_{i_1}(a_{i_1},a_{i_2},a_{i_3})=(a_{i_2}\oplus 1)a_{i_3},\\ 
		&f_{i_2}(a_{i_1},a_{i_2},a_{i_3})=(a_{i_3}\oplus 1)a_{i_1},\\ 
		&f_{i_3}(a_{i_1},a_{i_2},a_{i_3})=(a_{i_1}\oplus 1)a_{i_2}.
	\end{align} 
	
	\item $\{\{\circ\rightarrow\star\}, \{\star\rightarrow\bullet\}, \{\bullet\rightarrow\circ\}\}$: There are $2$ signalling graphs in this class (cyclic permutations of the labels are equivalent).
	
	For cyclic permutations of $S_1\rightarrow S_2\rightarrow S_3\rightarrow S_1$, the vertices are given by
	\begin{align}
		x_{1}&=f_{1}(a_{1},a_{3}),\\
		x_{2}&=f_{2}(a_{1},a_{2}),\\
		x_{3}&=f_{3}(a_{2},a_{3}),
	\end{align}
	with $12$ choices for each $f_{k}$, $k\in\{1,2,3\}$. There are $12^3=1728$ vertices with this signalling graph.
	
	For cyclic permutations of $S_1\rightarrow S_3\rightarrow S_2\rightarrow S_1$, the vertices are given by
	\begin{align}
		x_{1}&=f_{1}(a_{1},a_{2}),\\
		x_{2}&=f_{2}(a_{2},a_{3}),\\
		x_{3}&=f_{3}(a_{1},a_{3}),
	\end{align}
	with $12$ choices for each $f_{k}$, $k\in\{1,2,3\}$. There are $12^3=1728$ vertices with this signalling graph.
	
	Overall, we have $2\times 1728=3456$ vertices in this signalling class. All these vertices are noncausal since the directed graph admits no root node.
	
	From Corollary \ref{cor:siboncyc}, none of the vertices in this signalling class is achievable by a process matrix.
\end{enumerate}
As can be readily verified, the total number of vertices from the above classification scheme adds up to $8^8$, \textit{i.e.}, the total number of vertices in the $(3,2,2)$ correlational scenario.
\end{document}